\documentclass[11pt, a4paper]{article}
\usepackage{amsmath}
\usepackage{amsfonts}
\usepackage{amssymb}
\usepackage{amsthm}
\usepackage[inline]{enumitem}
\usepackage{graphicx}
\usepackage{fullpage}
\usepackage{hyperref}
\newcounter{favoriteitem}
\numberwithin{equation}{section}
\newtheorem{thm}{Theorem}[section]
\newtheorem{lem}[thm]{Lemma}

\newtheorem{prop}[thm]{Proposition}

\newtheorem{RHP}[thm]{RH Problem}

\theoremstyle{remark}
\newtheorem{rmk}{Remark}
\newtheorem{req}{Requirement}

\newcommand{\C}{\tilde{C}}
\newcommand{\compC}{\mathbb{C}}
\newcommand{\E}{\mathbb{E}}
\newcommand{\gfn}{\mathbf{g}}
\newcommand{\gfntilde}{\tilde{\mathbf{g}}}

\newcommand{\Iinv}{\mathbf{I}}
\newcommand{\J}[1][c]{\mathbf{J}_{#1}}
\newcommand{\tildeJ}{\tilde{\mathbf{J}}_{c}}
\newcommand{\Jhat}{\hat{J}}
\newcommand{\Mconst}{\mathbb{M}}
\newcommand{\natN}{\mathbb{N}}
\newcommand{\N}{\mathrm{N}}
\renewcommand{\hat}{\widehat}
\renewcommand{\tilde}{\widetilde}
\renewcommand{\i}{\mathit{i}}
\newcommand{\bigO}{\mathcal{O}}
\newcommand{\smallo}{\mathfrak{o}}
\newcommand{\p}{\widetilde{p}^{(n)}}
\newcommand{\phat}{\hat{p}^{(n)}}
\renewcommand{\P}{\mathcal{P}}
\newcommand{\Pmodel}{\mathsf{P}}
\newcommand{\paraP}{\mathbb{P}}
\newcommand{\q}{\widetilde{q}^{(n)}}
\newcommand{\qhat}{\hat{q}^{(n)}}
\newcommand{\R}{\mathcal{R}}
\newcommand{\realR}{\mathbb{R}}
\newcommand{\SigmaR}{\pmb{\mathit{\Sigma}}}
\newcommand{\intZ}{\mathbb{Z}}
\DeclareMathOperator{\Ai}{Ai}
\DeclareMathOperator{\Be}{Be}

\DeclareMathOperator{\arcosh}{arcosh}

\DeclareMathOperator{\interior}{int}
\DeclareMathOperator{\sol}{sol}
\DeclareMathOperator{\supp}{supp}
\DeclareMathOperator{\Var}{Var}
\DeclareMathOperator{\dist}{dist}
\DeclareMathOperator{\Res}{Res}
\title{Biorthogonal polynomials related to quantum transport theory of disordered wires}
\author{Dong Wang\thanks{School of Mathematical Sciences, University of Chinese Academy of Sciences, Beijing, P.~R.~China 100049 \newline
		email: \href{mailto:wangdong@wangd-math.xyz}{\protect\nolinkurl{wangdong@wangd-math.xyz}}}
\and
	Dong Yao\thanks{School of Mathematics and Statistics/RIMS, Jiangsu Normal University, Xuzhou, P.~R.~China 221116 \newline
		email: \href{dongyao@jsnu.edu.cn }{\protect\nolinkurl{dongyao@jsnu.edu.cn}} } 
}

\begin{document}
	\maketitle
	
	\begin{abstract}
		We consider the Plancherel-Rotach type asymptotics of the biorthogonal polynomials associated with the biorthogonal ensemble having the joint probability density function
		\begin{equation*}
			\frac{1}{C} \prod_{1 \leq i < j \leq n} (\lambda_j -\lambda_i)(f(\lambda_j) - f(\lambda_i)) \prod^n_{j = 1} W^{(n)}_{\alpha}(\lambda_j) d\lambda_j, 
		\end{equation*}
		where
		\begin{align*}
			f(x) = {}& \sinh^2(\sqrt{x}), & W^{(n)}_{\alpha}(x) = {}& x^{\alpha} h(x) e^{-nV(x)}.
		\end{align*}
		In the special case where the potential function $V$ is linear, this biorthogonal ensemble arises in the quantum transport theory of disordered wires. We analyze the asymptotic problem via $2$-component vector-valued Riemann-Hilbert problems and solve it under the one-cut regular with a hard edge condition. We use the asymptotics of the biorthogonal polynomials to establish sine universality for the correlation kernel in the bulk and provide a central limit theorem with a specific variance for holomorphic linear statistics.

		As an application of our theory, we establish Ohm's law \eqref{eq:Ohm_law} and universal conductance fluctuation \eqref{eq:UCF} for the disordered wire model, thereby rigorously confirming predictions from experimental physics \cite{Washburn-Webb86}.
	\end{abstract}
	
	\section{Introduction} \label{sec:Intro}
	
	\subsection{Setup of the model}
 Let $V$ be a real analytic function on $[0, \infty)$, and let $h$ be a positive-valued real analytic function on $[0, \infty)$ satisfying the limit condition
	\begin{equation} \label{eq:growth_at_infty}
		\lim_{x \to +\infty} \frac{V(x)}{\max(\log h(x), \sqrt{x})} = +\infty.
	\end{equation}
	We then denote the weight function $W^{(n)}_{\alpha}(x)$, which depends on the parameter $\alpha > -1$, as
	\begin{equation} \label{eq:general_weight}
		W^{(n)}_{\alpha}(x) = x^{\alpha} h(x) e^{-nV(x)}.
	\end{equation}
	
	Let
	\begin{equation} \label{eq:defn_f}
		f(x) = \frac{1}{4} \left( e^{2\sqrt{x}} - 2 + e^{-2\sqrt{x}} \right) = \sinh^2(\sqrt{x}), \quad x\in [0, \infty),
	\end{equation}
which will be fixed throughout the paper.
	We consider the monic polynomials $p^{(n)}_j(x)$ and $q^{(n)}_j(x)$, each of degree $j \geq 0$, determined by the orthogonality conditions
	\begin{equation} \label{eq:biorthogonality}
		\int_{\realR_+} p^{(n)}_j(x) f(x)^k W^{(n)}_{\alpha}(x) dx = 0, \quad \int_{\realR_+} x^k q^{(n)}_j(f(x)) W^{(n)}_{\alpha}(x) dx = 0, \quad k = 0, 1, \dotsc, j - 1,
	\end{equation}
	and define
	\begin{equation} \label{eq:defn_h}
		h^{(n)}_j = \int_{\realR_+} p^{(n)}_j(x) q^{(n)}_j(f(x)) W^{(n)}_{\alpha}(x) dx.
	\end{equation}
	Since $p^{(n)}_j$ and $q^{(n)}_j$ satisfy the biorthogonality condition \eqref{eq:biorthogonality}, they are called the biorthogonal polynomials with respect to $W^{(n)}_{\alpha}$ (the dependence on $f$ is suppressed because $f$ is fixed by \eqref{eq:defn_f}). 
	
	These biorthogonal polynomials are related to the point process consisting of $n$ particles at $\lambda_1, \dotsc, \lambda_n \in [0, \infty)$, with the joint probability density function (pdf) proportional to
	\begin{equation} \label{eq:jpdf}
		 \prod_{1 \leq i < j \leq n} (\lambda_j -\lambda_i)(f(\lambda_j) - f(\lambda_i)) \prod^n_{j = 1} W^{(n)}_{\alpha}(\lambda_j) d\lambda_j.
	\end{equation}
	Below, we call the point process defined by \eqref{eq:jpdf} the biorthogonal ensemble.

		In this paper, we are concerned with the Plancherel-Rotach type asymptotics of $p^{(n)}_{n + k}(x)$ and $q^{(n)}_{n + k}(f(x))$ as $n \to \infty$, where $k$ is any fixed integer. We will also use such asymptotics to study the limiting distribution of the linear statistics $\sum^n_{j = 1} \phi(\lambda_j)$ for analytic test functions $\phi$, as well as the limiting local distribution of the $\lambda_j$'s, as $n \to \infty$.
	
	The first relation between $p^{(n)}_j, q^{(n)}_j$ and the biorthogonal ensemble \eqref{eq:jpdf} is:
	\begin{prop} \label{prop:char}
		The polynomials $p_n^{(n)}(x)$ and $q_n^{(n)}(f(x))$ have the following representations:
		\begin{align}
			p^{(n)}_n(z) = {}& \E \left[ \prod^n_{j = 1} (z - \lambda_j) \right], & q^{(n)}_n(f(z)) = {}& \E \left[ \prod^n_{j = 1} (f(z) - f(\lambda_j)) \right],
		\end{align}
		where $\E[ \cdot ]$ is taken with respect to the joint pdf \eqref{eq:jpdf}.
	\end{prop}
	
	The proof of this proposition is analogous to the proofs of \cite[Proposition 2.1]{Bleher-Kuijlaars04a} and \cite[Proposition 1]{Claeys-Wang11}, so we omit it. The most important consequence of this proposition in our paper is the existence and uniqueness of $p^{(n)}_n, q^{(n)}_n$, as well as $p^{(n)}_j, q^{(n)}_j$ for general $j$. 
	More precisely, by replacing $V$ with $\hat{V}(x)=nV(x)/j$, we obtain the relation
	\begin{equation}\label{eq:wnj}
			W_{\alpha}^{(n)}(x)= x^{\alpha}h(x)e^{-nV(x)}=
		x^{\alpha}h(x)e^{-j\hat{V}(x)}:=\hat{W}^{(j)}_{\alpha}(x).
	\end{equation}
	Applying Proposition \ref{prop:char} to the weight $\hat{W}^{(j)}_{\alpha}$, we obtain the existence and uniqueness of $\hat{p}^{(j)}_j$, which is exactly equal to $p^{(n)}_j$ by \eqref{eq:wnj}.
	
	The next and more important relation between $p^{(n)}_j, q^{(n)}_j$ and the biorthogonal ensemble \eqref{eq:jpdf} is:
	\begin{prop} \label{prop:DPP}
		The biorthogonal ensemble with the joint pdf \eqref{eq:jpdf} is a determinantal point process, and the following $K_n(x, y)$ is its correlation kernel:
		\begin{equation} \label{eq:corr_kernel}
			K_n(x, y) = \sqrt{W^{(n)}_{\alpha}(x) W^{(n)}_{\alpha}(y)} \sum^{n - 1}_{j = 0} \frac{1}{h^{(n)}_j} p^{(n)}_j(x) q^{(n)}_j(f(y)),\quad x,y\geq 0.
		\end{equation}
	\end{prop}
	
	The proof of Proposition \ref{prop:DPP} is based on the general theory of determinantal point processes. See \cite{Tracy-Widom98} for the framework and \cite[Section 2]{Borodin99} for a short discussion on the kernel formula of a similar model. We omit the details of the proof.
	
	 Proposition \ref{prop:DPP} sets the stage for computing the local limiting density of particles (Theorem \ref{thm:sinekernel} below) through the Plancherel-Rotach type asymptotics for $p^{(n)}_j$ and $q^{(n)}_j$ (Theorem \ref{thm:main}). See Section \ref{sec:bulk universality} for details. 
	
	\subsection{Motivation} \label{subsec:Motivation}
	
	Our study of the biorthogonal ensemble defined by the general weight function \eqref{eq:general_weight} is inspired by the concrete case where the weight function is specialized as
	\begin{align} \label{eq:special_W}
		V(x) = {}& \frac{x}{\Mconst}, && \alpha =  0, && \text{and} && h(x) = f'(x)^{\frac{1}{2}} = \left( \frac{\sinh(\sqrt{x}) \cosh(\sqrt{x})}{\sqrt{x}} \right)^{\frac{1}{2}}, \quad  x\geq 0,
	\end{align}
where $\Mconst$ is a positive real number.
	This model is proposed in the study of the quantum transport theory of a disordered wire. See \cite[Formula (19)]{Beenakker-Rejaei93} and \cite[Section III A, Formula (3.4)]{Beenakker-Rejaei94} for the physical derivation of the biorthogonal ensemble. (In \cite{Beenakker-Rejaei94}, $\Mconst$ is denoted as $s$ and $n$ is denoted as $N$.) We also refer interested readers to the review article \cite{Beenakker97} for the physical theory that relates the biorthogonal ensemble with the specialization \eqref{eq:special_W}. 
	
	For the purpose of our paper, it suffices to note that the biorthogonal ensemble with the specialization \eqref{eq:special_W} approximates the distribution of particles (representing transmission eigenvalues) in the Dorokhov-Mello-Pereyra-Kumar (DMPK) equation with the parameter $\beta = 2$ and the ballistic initial condition. The DMPK equation is the evolution equation for the density function $P(\tilde{\lambda}_1, \dotsc, \tilde{\lambda}_n; \Mconst)$, where $\Mconst$ is the time parameter:
	\begin{equation} \label{eq:DMPK}
		\frac{\partial P}{\partial \Mconst} = \frac{1}{n} \sum^n_{j = 1} \frac{\partial}{\partial \tilde{\lambda}_j} \left( \tilde{\lambda}_j (1 + \tilde{\lambda}_j) J \frac{\partial}{\partial \tilde{\lambda}_j} \frac{P}{J} \right), \quad J = \prod^n_{i = 1} \prod^n_{j = i + 1} \lvert \tilde{\lambda}_i - \tilde{\lambda}_j \rvert^2.
	\end{equation}
	(See \cite[Equation (145)]{Beenakker97}. Here, we take $\beta = 2$ in that formula and use $\Mconst$ to denote $L/l$ there. In many parts of \cite{Beenakker97}, $L/l$ is denoted as $s$.) The $\tilde{\lambda}_j$ in \eqref{eq:DMPK} corresponds to $\sinh^2(\sqrt{\lambda_j})$ in \eqref{eq:jpdf}. We remind readers that the claim that the joint probability density function of $\{ \sinh^2(\sqrt{\lambda_j}) \}$ approximates $P(\tilde{\lambda}_1, \dotsc, \tilde{\lambda}_n; \Mconst)$ has been justified only in the regime where $\Mconst$ and $n$ are large and $1 \ll \Mconst \ll n$.
	
	The distribution of particles $\tilde{\lambda}_1, \dotsc, \tilde{\lambda}_n$ in \eqref{eq:DMPK} (resp.\@ $\lambda_1, \dotsc, \lambda_n$ in \eqref{eq:jpdf}) represents (resp.\@ approximately represents) the distribution of the transmission eigenvalues of the disordered wire, and the sum
	\begin{equation}\label{eq:conduct}
		C_n(\Mconst) := \sum^n_{j = 1} \frac{1}{\cosh^2(\sqrt{\lambda_j})}
	\end{equation}
	yields the conductance of the disordered wire, at least in the regime $1 \ll \Mconst \ll n$. Results from the experimental physics literature \cite{Washburn-Webb86} (see also \cite{Beenakker97}) indicate the following mathematical results:
	\begin{itemize}
		\item Ohm's law
		\begin{equation} \label{eq:Ohm_law}
			\lim_{\Mconst \to \infty} \Mconst \lim_{n \to \infty} \frac{1}{n} \E[C_n(\Mconst)] = 1.
		\end{equation}
		\item Universal conductance fluctuation
		\begin{equation} \label{eq:UCF}
			\lim_{\Mconst \to \infty} \lim_{n \to \infty} \Var[C_n(\Mconst)] = \frac{1}{15}.
		\end{equation}
	\end{itemize}
We remark that the universal conductance fluctuation \eqref{eq:UCF} has previously only been justified by physical arguments \cite{Mello88}, \cite{Lee-Stone85}, \cite{Mello-Stone91}, while the counterpart of \eqref{eq:UCF} for a simpler model (the quantum dot) has been rigorously proved in \cite{Baranger-Mello94}, \cite{Jalabert-Pichard-Beenakker94}, \cite{Iida-Weidenmuller-Zuk90}, \cite{Savin-Sommers06}; see also \cite{Jiang09a}.  
In this paper, we develop a framework for analyzing biorthogonal polynomials and derive a Plancherel-Rotach type asymptotic result (Theorem \ref{thm:main} below). This enables us to rigorously analyze the limiting distribution of $\lambda_1, \dotsc, \lambda_n$ and prove \eqref{eq:Ohm_law} and \eqref{eq:UCF}. See Theorems \ref{thm:eq_measure_linear} and \ref{thm:UCF1} below.  
	
	\subsection{Main results}
	In this section, we state the main contributions of the paper. To help readers grasp the main results without getting lost in technical details, we have deferred some technical definitions and detailed statements to later sections.
	\subsubsection{Global results: Properties of the equilibrium measure}
	
	Analogously to determinantal point processes associated with orthogonal polynomials (i.e., replacing $f(\lambda_j) - f(\lambda_i)$ with $\lambda_j - \lambda_i$ in the joint pdf \eqref{eq:jpdf}), we define the equilibrium measure supported on $[0, \infty)$ as the minimizer of the functional 
	\begin{equation} \label{eq:energy_functional}
		I_V(\mu) := \frac{1}{2} \iint \log \lvert t - s \rvert^{-1} d\mu(t)d\mu(s) + \frac{1}{2} \iint \log \lvert f(t) - f(s) \rvert^{-1} d\mu(t)d\mu(s) + \int V(s) d\mu(s).
	\end{equation}
	\begin{prop} \label{thm:potential_theory}
		Let $V$ be a continuous function on $[0, \infty)$ satisfying
		$$
		\lim_{x\to +\infty} \frac{V(x)}{\sqrt{x}}=\infty.
		$$
		 Then there exists a unique measure $\mu = \mu_V$ on $\realR_+$ with compact support that minimizes the functional \eqref{eq:energy_functional} among all Borel probability measures on $\realR_+$.
	\end{prop}
	
	The basic idea of the proof is contained in \cite{Deift99}, which provides a detailed proof of the result if $f(x)$ is replaced by $x$ and the domain of integration is replaced by $\realR$. If $f(x)$ is replaced by $e^x$ and the domain of integration is replaced by $\realR$, an explanation is given in \cite[Section 2]{Claeys-Wang11}. Hence, we omit the proof here.
	
	\begin{rmk}
		Proposition \ref{thm:potential_theory} is a potential-theoretic result concerned only with the existence and uniqueness of $\mu = \mu_V$. To explicitly construct $\mu$ (see Theorem \ref{eq:find_b_psi} below), we instead use a Riemann-Hilbert approach.


	\end{rmk}
	
	\begin{rmk}
	Under the condition \eqref{eq:growth_at_infty}, $\mu$ is the limit of the empirical law of the biorthogonal ensemble \eqref{eq:jpdf} (see \cite[Theorem 2.1]{Borot-Guionnet-Kozlowski13}). It is also possible to employ the argument given in \cite[Chapter 6]{Deift99} to establish the equivalence of $\mu$, the density of the ensemble \eqref{eq:jpdf}, and the zeros of $p^{(n)}_n(z), q^{(n)}_n(f(z))$ (as $n\to\infty$), but we do not pursue this approach in the paper. 
	\end{rmk}
	We say that a potential function $V$ satisfying \eqref{eq:growth_at_infty} is ``one-cut regular with a hard edge'' if its equilibrium measure $\mu$ (given by Proposition \ref{thm:potential_theory}) is absolutely continuous with
	\begin{equation} \label{eq:density_formula}
		d\mu(x) = \psi(x)dx \quad \text{on $\realR_+$},
	\end{equation}
where the function $\psi$ satisfies:
	\begin{req} \label{req:one-cut_reg} \hfill
		\begin{enumerate}
			\item \label{enu:req:one-cut_reg:1}
			$\supp \mu = [0, b]$ for some $b > 0$ that depends on $V$, and $\int d\mu(x) = 1$.
			\item \label{enu:req:one-cut_reg:2}
			$\psi(x)$ is continuous on $(0, b)$ and $\psi(x) > 0$ for all $x \in (0, b)$.
			\item \label{enu:req:one-cut_reg:4}
			For all $x \in [0, b]$, there exists a constant $\ell$ depending on $V$ such that
			\begin{equation} \label{eq:defn_ell}
				\int \log \lvert t - x \rvert^{-1} d\mu(t) + \int \log \lvert f(t) - f(x) \rvert^{-1} d\mu(t) + V(x) + \ell = 0.
			\end{equation}
			\item \label{enu:req:one-cut_reg:5}
			For all $x > b$,
			\begin{equation} \label{eq:defn_ell:2}
				\int \log \lvert t - x \rvert^{-1} d\mu(t) + \int \log \lvert f(t) - f(x) \rvert^{-1} d\mu(t) + V(x) + \ell > 0.
			\end{equation}
			\item \label{enu:req:one-cut_reg:3}
			The two limits
			\begin{align} \label{eq:reg:one-cut_reg:3}
				\psi_0 := {}& \lim_{x \to 0_+} x^{\frac{1}{2}}\psi(x) & & \text{and} & \psi_b := {}& \lim_{x \to b_-} (b - x)^{-\frac{1}{2}} \psi(x)
			\end{align}
			exist and are both positive.
		\end{enumerate}
	\end{req}
	
	It is clear that Item \ref{enu:req:one-cut_reg:1} implies that $\mu$ is a probability measure, and Item \ref{enu:req:one-cut_reg:2} further implies that $d\mu(x) = \psi(x) dx$ is a ``one-cut'' probability measure in the sense that its support is a compact interval and its probability density is positive everywhere in the interior of the support. Items \ref{enu:req:one-cut_reg:4} and \ref{enu:req:one-cut_reg:5} are slightly stronger than the Euler-Lagrange equation of the variational problem \eqref{eq:energy_functional}, so if $\psi(x)$ satisfies all of Items \ref{enu:req:one-cut_reg:1}, \ref{enu:req:one-cut_reg:2}, \ref{enu:req:one-cut_reg:4}, and \ref{enu:req:one-cut_reg:5}, then $d\mu(x) = \psi(x) dx$ is the unique equilibrium measure defined by the minimization of $I_V$, as stated in Proposition \ref{thm:potential_theory}. Finally, Item \ref{enu:req:one-cut_reg:3} means that the equilibrium measure is regular at $0$ (the ``hard edge'') and $b$ (the ``soft edge''). (The regularity of the equilibrium measure also includes Item \ref{enu:req:one-cut_reg:2}, which means it is regular in the interior of the support, and Item \ref{enu:req:one-cut_reg:5}, which means it is regular outside the support.) Throughout this paper, we only consider potential functions $V$ that are one-cut regular with a hard edge.

	Given a potential function $V$, it is generally difficult to determine whether $V$ is one-cut regular with a hard edge. In this paper, the following partial result suffices for our purposes.
	\begin{thm} \label{thm:one-cut_rgular_w_hard_edge}
		If the potential function $V$ is real analytic on $[0, \infty)$, satisfies \eqref{eq:growth_at_infty}, and
		\begin{equation} \label{eq:one-cut_regular_cond}
			U'(x) > 0 \quad \text{for all} \quad x \in (0, \infty), \quad \text{where} \quad U(x) = V'(x) \sqrt{x},
		\end{equation}
		then $V$ satisfies the one-cut regular with a hard edge condition stated in Requirement \ref{req:one-cut_reg}.
	\end{thm}
	
	The proof of this theorem is given in Section \ref{sec:constr_eq_measure}. We note that Theorem \ref{thm:one-cut_rgular_w_hard_edge} only provides a qualitative description of the equilibrium measure when \eqref{eq:one-cut_regular_cond} is satisfied. In fact, we can also derive an exact formula for the equilibrium measure. Since this involves some notations to be given later, we defer the result to Theorem \ref{eq:find_b_psi} in Section \ref{subsubsec:J_trans}.
 
	\subsubsection{Local results: Plancherel-Rotach asymptotics and their consequences} \label{subsubsec:local_result}
	
        Assuming certain regularity conditions on the potential function $V$, the biorthogonal polynomials $p^{(n)}_j(z)$ and $q^{(n)}_j(f(z))$ have the same type of Plancherel-Rotach asymptotics as orthogonal polynomials of the Laguerre type, as studied in \cite{Vanlessen07}:
        \begin{thm} \label{thm:main_summary}
          Suppose that $V$ is analytic on $[0, \infty)$, satisfies the condition \eqref{eq:growth_at_infty}, and meets Requirement \ref{req:one-cut_reg} (which defines the one-cut regular with a hard edge condition). As $n \to \infty$ and $k \in \intZ$ is fixed, both $p^{(n)}_{n + k}(z)$ and $q^{(n)}_{n + k}(f(z))$ converge to the Bessel function $I_{\alpha}$ as $z \to 0$ and to the Airy function $\Ai$ as $z \to b$, after appropriate rescalings, respectively. Additionally, both $p^{(n)}_{n + k}(z)$ and $q^{(n)}_{n + k}(f(z))$ exhibit sinusoidal limiting behavior for $z \in (0, b)$. 
        \end{thm}
        The precise details of Theorem \ref{thm:main_summary} will be given in Theorem \ref{thm:main}, after the necessary notations are defined.
	The proof of Theorems \ref{thm:main_summary} and \ref{thm:main} is given in Section \ref{sec:proofs_main}, which is based on the Riemann-Hilbert analysis in Sections \ref{sec:asy_p} and \ref{sec:asy_q}.

		Using Theorem  \ref{thm:main_summary}, we can adapt the approach in \cite{Claeys-Wang22} (see Section \ref{subsec:strategypf} for more details) to our current setting and show that the correlation kernel $K_n(x, y)$ in \eqref{eq:corr_kernel} has the sine kernel as its limit in the bulk (the interior $(0,b)$ of the support of $\mu$). 
For $x$ belonging to a neighborhood of the positive real line in $\mathbb{C}$, we define 
	\begin{equation}\label{deflambda}
		\Lambda(x) = \frac{1}{2} \int_0^b \log \left( \frac{f(x) - f(y)}{x - y} \right) d\mu(y).
	\end{equation}
	\begin{thm}\label{thm:sinekernel} 
		 Suppose that $V$ satisfies the same conditions as in Theorem \ref{thm:one-cut_rgular_w_hard_edge}. Fix any $x^*\in (0,b)$ and a compact subset $\Xi \subset \compC$, and let $\xi,\eta \in \Xi$. Set
                 \begin{align}
                   u = {}& u_n(\xi)= x^* + \frac{\xi}{\pi \psi(x^*) n}, & v = {}& v_n(\eta)=  x^* + \frac{\eta}{\pi \psi(x^*) n}.
                 \end{align} 
	 Uniformly over $\xi,\eta \in \Xi$, we have
		\begin{equation} \label{eq:_sine_kernel_pre}
			\lim_{n \to \infty} (f(u ) - f(v)) e^{n\Lambda(u)} K_n(u, v) e^{-n\Lambda(v)} = \frac{f'(x^*)}{\pi} \sin (\xi - \eta).
		\end{equation}
		Consequently, for each fixed $\xi$ and $\eta$, 
		\begin{equation} \label{eq:sine_kernel_nonzero}
			\lim_{n \to \infty} \frac{e^{n(\Lambda(u) - \Lambda(v))}}{\pi \psi(x^*)n} K_n\left(x^* + \frac{\xi}{ \psi(x^*) n}, x^* + \frac{\eta}{\pi \psi(x^*) n} \right)  = \frac{\sin (\xi - \eta)}{\xi - \eta}.
		\end{equation}
	By taking $\xi=\eta=0$ in \eqref{eq:sine_kernel_nonzero}, we verify that the limiting density in the bulk $(0,b)$ is given by $\psi$:
		\begin{equation}\label{eq:limitpsi}
			\lim_{n \to \infty} \frac{1}{ n}K_n(x^*,x^*)=\psi(x^*).
		\end{equation} 
	\end{thm}
Theorem \ref{thm:sinekernel} will be proved in Section \ref{sec:bulk universality}. Note that for 
	 $\xi\neq \eta$, \eqref{eq:sine_kernel_nonzero} directly follows from \eqref{eq:_sine_kernel_pre} via the expansion $f(u)-f(v)=f'(x^*) \frac{\xi-\eta}{n}+\bigO(n^{-2})$. 
	 The validity of \eqref{eq:sine_kernel_nonzero} for $\xi=\eta$ then follows from the analyticity of both sides of \eqref{eq:sine_kernel_nonzero} in $\xi$ and $\eta$. We also use the convention that $\frac{\sin 0}{0}=1$.

	Analogously to Theorem \ref{thm:sinekernel}, we can also show that the rescaled correlation kernel $K_n(x, y)$ in \eqref{eq:corr_kernel} converges to the Airy kernel around $b$ and to the Bessel kernel around $0$, respectively. In fact, the connections between a related model (the Muttalib-Borodin ensemble) and the Meijer-G kernel have been established in
	\cite{Kuijlaars-Stivigny14, Wang-Zhang21, Zhang15}.
	In \cite{Wang-Zhang21}, it is shown that the limit of the Muttalib-Borodin correlation kernel is given by a Meijer-G kernel with a parameter $\theta$. If $\theta = 1$, the Meijer G kernel specializes to the Bessel kernel. 
	
	Existing limit theories \cite{Borot-Guionnet-Kozlowski13, Breuer-Duits13}  for general biorthogonal ensembles provide a central limit theorem for holomorphic linear statistics. (See Propositions \ref{thm:mean}, \ref{prop:clt1} and \ref{thm:var} below.) Building on their results, we can combine  Theorem \ref{thm:main_summary} and a double contour integral representation to obtain a clear expression for the limiting variance of analytic test functions. See Theorem \ref{thm:variance_general} below. Our approach also suggests that analogous formulas  to \eqref{eq:sigform} hold  for other biorthogonal ensembles. See  Section \ref{subsec:strategypf}  and Section \ref{subsec:formula1} below for more details.

	 Let the region $D$ be defined in \eqref{eq:simplified} and the mapping $\J$ be given in \eqref{eq:defn_J_c} and \eqref{eq:determine_c}.
	
	
	\begin{thm} \label{thm:variance_general}
          Suppose $V$ satisfies the same conditions as in Theorem \ref{thm:one-cut_rgular_w_hard_edge}, and assume that $\phi$ is a holomorphic function in a neighborhood of $[0,b]$ in $\mathbb{C}$. Let $\mathcal{C}(D)$ and $\mathcal{C}'(D)$ be two smooth, positively oriented Jordan curves in $\compC \setminus \overline{D}$ that enclose $D$, with $\mathcal{C}(D)$ enclosing $\mathcal{C}'(D)$. Assume that both $\phi(\J(w))$ and\footnote{The function $h$ appears in the definition of the weight in \eqref{eq:general_weight}.} $h(\J(w))$ are analytic in the region between $\mathcal{C}(D)$ and $\partial D$. Moreover, suppose that\footnote{If $\phi$ is a polynomial, the condition $\alpha = 0$ can be removed, by using the result from \cite{Breuer-Duits13}. See Section \ref{subsec:formula1} below for more discussions.} $\alpha = 0$. 
                Then we have 
		\begin{equation} \label{eq:general_CLT}
			\sum_{i=1}^n \phi(\lambda_i) -\mathbb{E} \left(\sum_{i=1}^n \phi(\lambda_i)\right) \to \N(0,\sigma^2(\phi)) \mbox{ in distribution}, \mbox{ as }n\to\infty, 
		\end{equation} 
		where the right-hand side is the Gaussian distribution with mean $0$ and variance 
		\begin{equation}\label{eq:sigform}
			\sigma^2(\phi) = \frac{1}{(2\pi i)^2} \oint_{\mathcal{C}'(D)} du \oint_{\mathcal{C}(D)} dv \phi(\J(u)) \phi(\J(v)) \frac{1}{(u - v)^2}.
		\end{equation}
	\end{thm}
	
	Theorem \ref{thm:variance_general} implies that the limiting variance of linear statistics depends on the potential $V(x)$ only through its support $[0,b]$, similar to the case of orthogonal ensembles. This theorem will be proved in Section \ref{subsec:formula1}. 
	
	\subsubsection{Applications to the conductance of disordered wires} \label{subsubsec:disordered_wires_results}
By providing quantitative descriptions of the equilibrium measure (Theorem \ref{eq:find_b_psi} in Section \ref{subsubsec:J_trans} below), we obtain the following explicit form of the limiting density and the mean conductance when $V$ is linear. In particular, we rigorously verify Ohm's law \eqref{eq:Ohm_law}.
\begin{thm} \label{thm:eq_measure_linear}
	Suppose $V(x)$, $\alpha$, and $h(x)$ are specified in \eqref{eq:special_W}. Then the parameter $c$ defined in \eqref{eq:relation_b_c} is equal to $2\Mconst$, and the endpoint $b$ is given by
	\begin{equation} \label{eq:b_linear_case}
		b = \left( \frac{1}{2} \sqrt{(2\Mconst + 1)^2 - 1} + \frac{1}{2} \arcosh(2\Mconst + 1) \right)^2.
	\end{equation}
	Furthermore, the density function $\psi(x)$ in \eqref{eq:density_formula} is given by
	\begin{equation} \label{eq:psi_linear_case}
		\psi(x) = \frac{1}{\pi} \Im \sqrt{\frac{\Iinv_+(x)}{x}},
	\end{equation}
where $\Iinv_+$ is given by \eqref{eq:Iinv_+}, \eqref{eq:simplified}, and \eqref{eq:determine_c}.
Consequently, the conductance $C_n(\Mconst)$ defined in \eqref{eq:conduct} satisfies  
\begin{equation}\label{eq:meanconduc}
	\lim_{n \to \infty} \frac{1}{n} \E[C_n(\Mconst)] = \frac{1}{\Mconst+1},
\end{equation}
which implies Ohm's law \eqref{eq:Ohm_law} by taking $\Mconst\to\infty$.
\end{thm}
Equations \eqref{eq:b_linear_case} and \eqref{eq:psi_linear_case} of Theorem \ref{thm:eq_measure_linear} will be proved in Section \ref{sec:constr_eq_measure}. Equation \eqref{eq:meanconduc} will be shown in Section \ref{subsubsec:meanconduc}. 
Readers may compare \eqref{eq:b_linear_case} and \eqref{eq:meanconduc} to \cite[Equation (205)]{Beenakker97} and \cite[Equation (176)]{Beenakker97}, respectively. 
From our formula \eqref{eq:psi_linear_case}, we obtain the following limit results. In the limit as $\Mconst \to 0$, the support of the equilibrium measure is $[0, (4+\smallo(1))\Mconst ]$, and the density function satisfies the limiting formula that for all $\epsilon > 0$,
\begin{gather}
	\psi(x) = \frac{1}{2\pi \Mconst} \sqrt{\frac{4\Mconst - x}{x}} (1 + \bigO(\Mconst)), \quad \mbox{uniformly over } x \in (\epsilon\Mconst, (4 - \epsilon)\Mconst), \label{eq:MP_law_limit} \\
	\begin{aligned}
		\psi_0 = {}& \frac{\Mconst^{-\frac{1}{2}}}{\pi} (1 + \bigO(\Mconst)), & \psi_b = {}& \frac{\Mconst^{-\frac{3}{2}}}{4\pi} (1 + \bigO(\Mconst)).
	\end{aligned}
\end{gather}
In the limit as $\Mconst \to \infty$, the support of the equilibrium measure is $[0, \Mconst^2 + \bigO(\Mconst)]$, and the density function satisfies the limiting formula that for all $\epsilon > 0$,
\begin{gather}
	\psi(x) = \frac{1}{2\Mconst \sqrt{x}} (1 + \bigO(\Mconst^{-1})),  \quad \mbox{uniformly over } x \in (\epsilon \Mconst^2, (1 - \epsilon) \Mconst^2), \label{eq:classical_limit} \\
	\begin{aligned}
		\psi_0 = {}& \frac{1}{2\Mconst} (1 + \bigO(\Mconst^{-1})), & \psi_b = {}& \frac{1}{\sqrt{2} \pi \Mconst^{\frac{5}{2}}} (1 + \bigO(\Mconst^{-1})).
	\end{aligned}
\end{gather}
Here, we note that \eqref{eq:MP_law_limit} is comparable to the Mar\v{c}enko-Pastur law \cite[Chapter 3]{Bai-Silverstein10}, and \eqref{eq:classical_limit} is comparable to \cite[Equation (191)]{Beenakker97}. Moreover, Theorem \ref{thm:sinekernel} implies that the local statistics of the $\lambda_i$'s in the bulk exhibit sine universality.

Using Theorem \ref{thm:variance_general}, we can also prove the universal conductance fluctuation \eqref{eq:UCF}. 
\begin{thm}\label{thm:UCF1}
	Suppose $V(x)$, $\alpha$, and $h(x)$ are specified in \eqref{eq:special_W}. The conductance $C_n(\Mconst)$ defined in \eqref{eq:conduct} has a Gaussian limit:
		\begin{equation}\label{eq:wireclt1} 
C_n(\Mconst)		 -\mathbb{E} \left( C_n(\Mconst)	\right) \to \N\left(0,\frac{1}{15}\left(1-\frac{6\Mconst+1}{(\Mconst+1)^6}\right)  \right) \mbox{ in distribution}, \mbox{ as }n\to\infty, 
	\end{equation} 
We also have 
\begin{equation}\label{eq:wireclt2}
	\lim_{n\to\infty}\Var(C_n(\Mconst))=
	\frac{1}{15}\left(1-\frac{6\Mconst+1}{(\Mconst+1)^6}\right), 
\end{equation}
which implies the universal conductance fluctuation \eqref{eq:UCF} by taking $\Mconst\to\infty$.
\end{thm}
Our result \eqref{eq:wireclt2} coincides with the formula suggested by physical methods in \cite[Equation (178)]{Beenakker97}.
We will prove Theorem \ref{thm:UCF1} in Section \ref{subsubsec:pfucf}. 
	\subsection{Related models and previous results}
	
        \paragraph{Biorthogonal ensembles}
       An important class of biorthogonal polynomials is the \emph{Muttalib-Borodin biorthogonal polynomials}, which are characterized by $f(x) = x^{\theta}$ in \eqref{eq:biorthogonality}. Correspondingly, the particle model given in \eqref{eq:jpdf} with $f(x) = x^{\theta}$ is called the \emph{Muttalib-Borodin ensemble}.
	This ensemble was proposed by the physicist Muttalib \cite{Muttalib95} as a simplification of the biorthogonal ensemble considered in our paper, and the study of the Muttalib-Borodin ensemble \cite{Borodin99}, \cite{Cheliotis14}, \cite{Forrester-Wang15}, \cite{Zhang15}, \cite{Husson-Mazzuca-Occelli25} is partially motivated by its indirect relation to the quantum transport theory of disordered wires.
	We also remark that technically, the Muttalib-Borodin ensemble is more challenging, as its limit behavior at the hard edge is the more complex Meijer G kernel, rather than the Bessel kernel.
	Finally, we note that biorthogonal ensembles are also motivated by several other areas, such as \cite{Borot-Guionnet-Kozlowski13}, \cite{Lubinsky-Sidi-Stahl15}, \cite{Kuijlaars16}, \cite{Kuijlaars-Molag19}, \cite{Molag20}.
        \paragraph{Variance of linear statistics}
        A striking feature of Theorem \ref{thm:variance_general} is that the limiting variance of the linear statistic depends only on the support of the equilibrium measure, not on the density function itself. This feature is well known for linear statistics in orthogonal polynomial ensembles that can be realized by \eqref{eq:jpdf} with $f(x) = x$; see, for example, \cite{Johansson98} and \cite[Section 2.2]{Breuer-Duits13}. For biorthogonal ensembles, to the best of the authors' knowledge, this feature has not been discussed in the mathematical literature but has been indicated in physical arguments (e.g., \cite[Equations (50) and (55)]{Beenakker97}).
        \paragraph{Disordered wires}
        The biorthogonal ensemble considered in Section \ref{subsubsec:disordered_wires_results} is an approximate solution to the DMPK equation \eqref{eq:DMPK}, and the approximation has been justified by physical arguments only in the regime $1 \gg \Mconst \gg n$, as explained in Section \ref{subsec:Motivation}. However, the mean conductance formula \eqref{eq:meanconduc} and the conductance fluctuation formula \eqref{eq:wireclt2} for finite $\Mconst$ coincide with the formulas derived directly from the DMPK equation by physical arguments \cite[Equations (205), (176), and (178)]{Beenakker97}. This suggests that the biorthogonal ensemble is a good approximation in the regime $\Mconst = \bigO(1)$.
	
        \subsection{Strategy of proofs}\label{subsec:strategypf}
        
        In this section, we highlight some of the key ideas in the proofs of the main results.
        
        Historically, orthogonal polynomials were related to Riemann-Hilbert problems by \cite{Fokas-Its-Kitaev91} and \cite{Fokas-Its-Kitaev92}. The powerful Deift-Zhou nonlinear steepest-descent method was then successfully applied to such Riemann-Hilbert problems \cite{Deift-Kriecherbauer-McLaughlin-Venakides-Zhou99}, \cite{Deift-Kriecherbauer-McLaughlin-Venakides-Zhou99a}, opening the door to many limiting results for orthogonal polynomials and their generalizations; see \cite{Kuijlaars10} for a review. These various Riemann-Hilbert problems are all matrix-valued, some of size $2 \times 2$ and others larger.
Later, in the study of a special type of biorthogonal polynomials related to the random matrix model with an equispaced external source \cite{Claeys-Wang11}, Claeys and the first-named author related the biorthogonal polynomials to $2$-component vector-valued Riemann-Hilbert problems and applied the Deift-Zhou method to them to obtain limiting results. The biorthogonal polynomials in \cite{Claeys-Wang11} have, in our notation, $f(x) = e^x$ in \eqref{eq:biorthogonality}, where the domain of integration is replaced by $\realR$. After \cite{Claeys-Wang11}, the method of vector-valued Riemann-Hilbert problems was applied to the Muttalib-Borodin biorthogonal polynomials \cite{Claeys-Romano14}, \cite{Wang-Zhang21}, \cite{Charlier21}.
        In our paper, we similarly apply the vector Riemann-Hilbert approach to prove Theorem \ref{thm:main_summary}, the Plancherel-Rotach asymptotics of the biorthogonal polynomials \eqref{eq:biorthogonality}.
        
        The universality of local statistics for orthogonal polynomial ensembles is usually proved via the Christoffel-Darboux formula. However, for our biorthogonal ensemble, the Christoffel-Darboux formula is not available. In \cite{Claeys-Wang22}, the sine universality of the biorthogonal ensemble studied in \cite{Claeys-Wang11} was proved using an approximate version of the Christoffel-Darboux formula. Here, we adopt (essentially) the same approach, which consists of two steps. First, we consider the decomposition of $K_n(x,y)$ in \eqref{eq:ess2}. We will show that the main contribution to the right-hand side of \eqref{eq:ess2} comes from the third term $J_3^{(M)}$, which contains those $p^{(n)}_k$ and $q^{(n)}_j$ where both $k$ and $j$ are close to $n$. This step is achieved using pointwise and $L^2$ bounds on the biorthogonal polynomials in Section \ref{subsec:upperbds}. Second, we evaluate $J_3^{(M)}$ using Theorem \ref{thm:main_summary} and certain algebraic identities (Lemma \ref{lem:cancel}) to obtain the sine kernel limit.
         We note that universality at the edges can be proved in a more straightforward manner; see \cite{Claeys-Wang22} for the soft edge, \cite{Wang-Zhang21} for the hard edge, and \cite{Wang-Xu25} for a transition case.
        
        For linear statistics of biorthogonal ensembles, thanks to the explicit density formula in Theorem \ref{eq:find_b_psi}, the limiting mean formula in \cite[Theorem 2.2]{Borot-Guionnet-Kozlowski13} yields Ohm's law in Theorem \ref{thm:eq_measure_linear}. Regarding the fluctuations of holomorphic linear statistics, \cite[Proposition 8.2]{Borot-Guionnet-Kozlowski13} provides a systematic way to show that they converge to a normal distribution if the ensemble satisfies certain regularity conditions. However, it is not easy to rewrite the limiting variance formula in \cite{Borot-Guionnet-Kozlowski13} into explicit results like \eqref{eq:wireclt1}. On the other hand, we find that the results in \cite{Breuer-Duits13, Lambert18} (see \cite[Theorem 2.1]{Breuer-Duits13} and \cite[Theorem 2.6]{Lambert18}) can be adapted to explicitly compute the limiting variances of linear statistics, provided that the test function is a polynomial. Hence, we adopt a mixed approach in our paper to derive \eqref{eq:wireclt1}: first, we show that \eqref{eq:general_CLT} holds for all polynomials by using Theorem \ref{thm:main_summary} to compute the right limit of the recurrence matrices of the biorthogonal polynomials (Lemma \ref{lem:recu}), and then we use the continuity property implicit in \cite[Proposition 8.2]{Borot-Guionnet-Kozlowski13} to extend \eqref{eq:sigform} to analytic $\phi$.
        The novelty here lies in the form of double contour integral formula \eqref{eq:sigform}, as it facilitates the limit procedure of going from polynomials to analytic functions. (Indeed, the original form \eqref{eq:variance_poly}   as in \cite{Breuer-Duits13, Lambert18} is much less convenient for taking limits of  test functions.) We believe
         our approach here   can be useful for  studying other biorthogonal ensembles as well (see Remark \ref{rmk:congen} below).
	
	\subsection*{Acknowledgements}
	Dong Wang and Dong Yao are both corresponding authors of the paper.
		Dong Yao is supported by the National Key R\&D Program of China (No. 2023YFA1010101), the NSFC grant (No. 12571161), and the Basic Research Program of Jiangsu Province grant (No. BK20220677).
	Dong Wang is partially supported by the National Natural Science Foundation of China under grant numbers 12271502 and 11871425, and the University of Chinese Academy of Sciences start-up grant 118900M043. 
	Dong Wang thanks K.~A.~Muttalib for discussions in the early stage of this project and Tiefeng Jiang for assistance with the literature review.

	\section{Construction of the equilibrium measure} \label{sec:constr_eq_measure}
	
	In this section, we assume that $V$ is a potential function satisfying the conditions in Theorem \ref{thm:one-cut_rgular_w_hard_edge}, and we show that $V$ is one-cut regular with a hard edge at $0$ by explicitly constructing its equilibrium measure as in Theorem \ref{eq:find_b_psi}, which implies Theorem \ref{thm:one-cut_rgular_w_hard_edge}. Following \cite{Claeys-Wang11}, we first propose the support of the equilibrium measure as an ansatz, then compute the density within this support, and finally verify that the constructed measure satisfies the criteria for one-cut regularity, concluding that it is the unique equilibrium measure.
	
	At each step, we analyze the special case $V = x/\Mconst$ and derive explicit formulas for it.
	
	\subsection{Quantitative properties of the equilibrium measure} \label{subsubsec:J_trans}
	
	
	First, we collect some properties of the function $f$ defined in \eqref{eq:defn_f}.
	 We define $\rho$ to be the curve lying in $\compC_-$ with the formula
	\begin{equation} \label{eq:defn_rho}
		\rho = \left\{ \frac{t^2}{\pi^2} - \frac{\pi^2}{4} -it \mid t \in (0, +\infty) \right\}.
	\end{equation}
	Then $\rho \cup \{ 0 \} \cup \bar{\rho}$ forms a parabola. We define $\paraP$ as the region of $\compC$ to the right of this parabola, and $\interior \paraP$ as its interior. See Figure \ref{fig:paraP}.
	\begin{lem} \label{lem:f_prop} 
		$f(x)$ defined in \eqref{eq:defn_f} has a natural extension to an analytic function on $\compC$ and satisfies the following properties:
		\begin{itemize}
			\item
			$f(x) \in \realR_+$ for all $x \in \realR_+$, and $f(x)$ increases from $0$ to $+\infty$ as $x$ ranges from $0$ to $+\infty$.
			\item
			$f(x) \in (-1, 0)$ for $x \in (-\pi^2/4, 0)$, and $f(x)$ increases from $-1$ to $0$ as $x$ ranges from $-\pi^2/4$ to $0$.
			\item
			As $x \in \realR_+$ ranges from $0$ to $+\infty$, $f\left(\frac{1}{4}(x^2 - \pi^2) \pm \frac{1}{2} \pi xi\right) \in \realR_-$ and decreases from $-1$ to $-\infty$.
		\end{itemize}
	Consequently, the function $f: \interior \paraP \to \compC \setminus (-\infty, -1]$ is conformal, and it maps both $\rho$ and $\bar{\rho}$ to $(-\infty, -1]$.
	\end{lem}
        The proof of this lemma is straightforward and is omitted.
\begin{rmk}
	Alternatively, if we glue $\rho$ and $\bar{\rho}$ by identifying $z \in \rho$ with $\bar{z} \in \bar{\rho}$ and view $\paraP$ as a Riemann surface, then by Lemma \ref{lem:f_prop}, $f$ is a conformal mapping between $\paraP$ and $\compC$.
\end{rmk}
	
	\begin{figure}[htb]
		\centering
		\includegraphics{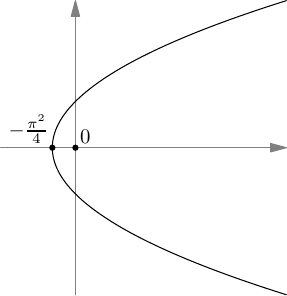}
		\caption{Shape of $\paraP$ (the region to the right of the parabola).}
		\label{fig:paraP}
	\end{figure}
	
	We define, for all $x \in \realR_+$, the transformation
	\begin{equation} \label{eq:defn_J_c}
		J_x(s) = x\sqrt{s} + \arcosh \left( \frac{s + 1}{s - 1} \right), \quad \text{and} \quad \J[x](s) = \frac{1}{4} (J_x(s))^2, \quad s \in \compC_+ \cup (1, \infty),
	\end{equation}
	such that $\sqrt{s}$ takes the principal branch on $\arg s \in (-\pi, \pi)$, and $\arcosh$ takes the branch that is a one-to-one mapping from $\compC \setminus (-\infty, 1]$ to $\{ s \in \compC \mid \Re s > 0 \text{ and } -\pi < \Im s < \pi \}$. Below, we extend the domain of $J_x(s)$ and $\J[x](s)$. Naturally, they extend to $s \in \compC_-$ by $J_x(\bar{s}) = \overline{J_x(s)}$ and $\J[x](\bar{s}) = \overline{\J[x](s)}$. $\J[x](s)$ also extends to $s \in (-\infty, 0)$ by analytic continuation. For $s \in [0, 1]$, we leave $\J[x](s)$ undefined, as it has a branch cut there. We also define ad hoc that $J_x(s) = \lim_{\epsilon \to 0_+} J_x(s + \epsilon i)$ for $s \in (-\infty, 1)$. In particular, $\Im J_x(s)=-\pi$ for $s\in (0,1)$. 
	
	\begin{lem} \label{lem:J_c}
		$J_x(s)$ satisfies the following properties:
		\begin{enumerate}
			\item \label{enu:lem:J_c:1}
			$J_x(s) \in \realR_+$ for $s \in (1, +\infty)$; $J_x(s)$ decreases from $+\infty$ to $J_x(s_2(x)) = \sqrt{(x + 1)^2 - 1} + \arcosh(x + 1)$ as $s$ ranges from $1$ to $s_2(x)$, which is defined by
			\begin{equation} \label{eq:defn_s_2_x}
				s_2(x) := 1 + 2/x,
			\end{equation}
			 and then increases from $J_x(s_2(x))$ to $+\infty$ as $s$ ranges from $s_2(x)$ to $\infty$.
			\item \label{enu:lem:J_c:2}
			$J_x(-s) \in i\realR$ for $s \in (0, +\infty)$. $\Im J_x(-s)$ increases monotonically from $-\pi$ to $+\infty$ as $s$ ranges from $0$ to $\infty$, and there exists a unique $s_1(x) \in (-\infty, 0)$ such that $J_x(s_1(x)) = 0$. Specifically, $s_1(x)$ is the unique solution on $(-\infty, 0)$ to
			\begin{equation} \label{eq:defn_s_1}
				x\sqrt{-s} = \arccos \left( \frac{-s - 1}{-s + 1} \right).
			\end{equation}
			\item \label{enu:lem:J_c:3}
			$J_x(s) + \pi i \in \realR_+$ for $s \in (0, 1)$. $\Re J_x(s)$ increases from $0$ to $+\infty$ as $s$ ranges from $0$ to $1$.
			\item \label{enu:lem:J_c:4}
			There exists a unique curve $\gamma_1(x) \subseteq \compC_+$ (see the left panel of Figure \ref{fig:J_c}) connecting $s_1(x)$ and $s_2(x)$, such that $J_x(z) \in \realR$ for $z \in \gamma_1(x)$, and $J_x(z)$ increases from $0$ to $J_x(s_2(x))$ as $z$ moves from $s_1(x)$ to $s_2(x)$ along $\gamma_1(x)$.
		\end{enumerate}
	\end{lem}
	
	Part \ref{enu:lem:J_c:4} of this lemma will be proved in Appendix \ref{sec:J_x_prop}. Parts \ref{enu:lem:J_c:1}, \ref{enu:lem:J_c:2}, and \ref{enu:lem:J_c:3} can be verified by direct computation, so we omit the details.
	
	Define
	\begin{equation} \label{eq:gamma_2}
		\gamma_2(x) = \{z \in \compC_- \mid \bar{z} \in \gamma_1(x) \}, 
	\end{equation}
	Then the union of $\gamma_1(x)$ and $\gamma_2(x)$ encloses a bounded region $D_x \subseteq \compC$. In this paper, we orient $\gamma_1(x)$ and $\gamma_2(x)$ from $s_1(x)$ to $s_2(x)$, unless stated otherwise. 
	For each $x \in (0, \infty)$, we let
	\begin{equation} \label{eq:defn_b(x)}
		\mathfrak{b}(x) := \J[x](s_2(x)) = \frac{1}{4}J_x(s_2(x))^2 = \frac{1}{4} \left( \sqrt{(x + 1)^2 - 1} + \arcosh(x + 1) \right)^2.
	\end{equation}
	It is clear that $\mathfrak{b}(x)$ is a continuous function of $x$ and increases monotonically from $0$ to $\infty$ as $x$ ranges from $0$ to $\infty$.
	We have the following lemma regarding the conformality of $\J[x]$. The proof will be
	postponed to Appendix \ref{sec:J_x_prop}. 
	\begin{lem} \label{lem:conformal_mapping}
		$\J[x]$ maps $\compC \setminus \overline{D_x}$ conformally to $\compC \setminus [0, \mathfrak{b}(x)]$ and maps $D_x \setminus [0, 1]$ conformally to $\paraP \setminus [0, \mathfrak{b}(x)]$.
	\end{lem}
	
	Let the functions $\Iinv_{x, 1}$ and $\Iinv_{x, 2}$ be the inverse functions of $\J[x]$, such that $\Iinv_{x, 1}$ is the inverse map of $\J[x]$ from $\compC \setminus [0, \mathfrak{b}(x)]$ to $\compC \setminus \overline{D_x}$, and $\Iinv_{x, 2}$ is the inverse map of $\J[x]$ from $\paraP \setminus [0, \mathfrak{b}(x)]$ to $D_x \setminus [0, 1]$:
	\begin{align}
		\Iinv_{x, 1}(\J[x](s)) = {}& s, && s \in \compC \setminus \overline{D_x}, \label{eq:defn_I1} \\
		\Iinv_{x, 2}(\J[x](s)) = {}& s, && s \in D_x \setminus [0, 1]. \label{eq:defn_I2}
	\end{align}
	We then denote for $u \in (0, \mathfrak{b}(x))$
	\begin{align}
		\Iinv_{x, +}(u) := {}& \lim_{\epsilon \to 0_+} \Iinv_{x, 1}(u + i\epsilon) = \lim_{\epsilon \to 0_+} \Iinv_{x, 2}(u - i\epsilon), \label{eq:Iinv_+} \\
		\Iinv_{x, -}(u) := {}& \lim_{\epsilon \to 0_+} \Iinv_{x, 1}(u - i\epsilon) = \lim_{\epsilon \to 0_+} \Iinv_{x, 2}(u + i\epsilon). \label{eq:Iinv_-}
	\end{align}
	We note that $\Iinv_{x, +}(x)$ lies in $\compC_+$, $\Iinv_{x, -}(x)$ lies in $\compC_-$, and their loci are the upper and lower boundaries of $D_x$ (i.e., $\gamma_1(x)$ and $\gamma_2(x)$, respectively). For later use, we define for $\xi \in (0, \mathfrak{b}(x))$
	\begin{equation} \label{eq:defn_F_x(uxi)}
		F_x(u; \xi) = \log \left\lvert \frac{(\sqrt{\Iinv_{x, +}(u)} + \sqrt{\Iinv_{x, +}(\xi)})(\sqrt{\Iinv_{x, -}(u)} - \sqrt{\Iinv_{x, +}(\xi)})}{(\sqrt{\Iinv_{x, +}(u)} - \sqrt{\Iinv_{x, +}(\xi)})(\sqrt{\Iinv_{x, -}(u)} + \sqrt{\Iinv_{x, +}(\xi)})} \right\rvert.
	\end{equation}
	
	The schematic illustration is given in Figures \ref{fig:J_c} and \ref{fig:J_on_D}. (To simplify the notation, we assume $x = c$ in Figure \ref{fig:J_on_D}; see \eqref{eq:simplified}.)

	\begin{figure}[htb]
		\centering
		\includegraphics{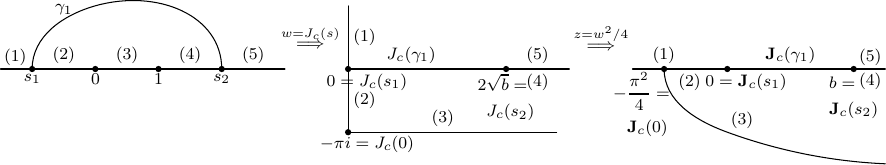}
		\caption[$J_c$ and $\J$.]{The schematic illustration of $J_c$ and $\J$ on $\compC_+$. (The definitions of $J_c$ and $\J$ are extended to $\compC_-$ naturally by complex conjugation.) If $c$ is replaced by a general $x \in (0, \infty)$, then $J_x(s_2(x)) = 2\sqrt{\mathfrak{b}(x)}$ and $\J[x](s_2(x)) = \mathfrak{b}(x)$ will change, while $J_x(s_1(x)) = \J[x](s_1(x)) = 0$, $J_x(0) = \pi i$, and $\J[x](0) = -\pi^2/4$ remain unchanged.}
		\label{fig:J_c}
	\end{figure}

	\begin{figure}[htb]
		\centering
		\includegraphics{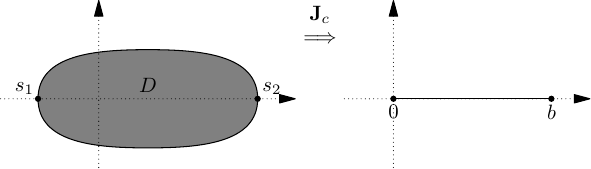} \\
		\vspace{1cm}
		\includegraphics{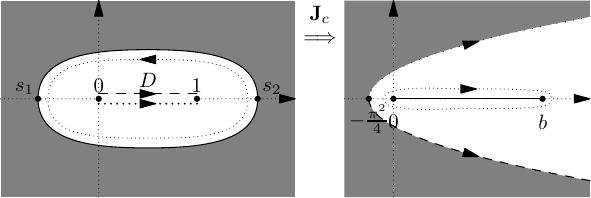}
		\caption[$\J$ transform]{$\J$ maps $\compC \setminus \overline{D}$ to $\compC \setminus [0, b]$ and maps $D \setminus [0, 1]$ to $\paraP \setminus [0, b]$.}
		\label{fig:J_on_D}
	\end{figure}
	
	The definitions of $J_x$ and $\J[x]$ are independent of $V$. Suppose $V$ is one-cut regular with a hard edge, such that the equilibrium measure $\mu_V$ associated with $V$ is supported on $[0, b]$ and has the density function $\psi(x)$. Given such $b>0$, we let $c=c(b) > 0$ be the unique solution to
	\begin{align} \label{eq:relation_b_c}
		\mathfrak{b}(c) = {}& b, && \text{or equivalently}, & \sqrt{(c + 1)^2 - 1} + \arcosh(c + 1) = {}& 2\sqrt{b}.
	\end{align}
	
	Throughout this paper, we denote
	\begin{align} \label{eq:simplified}
		s_i = {}& s_i(c), & \gamma_i = {}& \gamma_i(c), & D = {}& D_c, & \Iinv_i(u) = {}& \Iinv_{c, i}(u), & \Iinv_{\pm}(u) = {}& \Iinv_{c, \pm}(u), & F(u; \xi) = {}& F_c(u; \xi).
	\end{align}

        Below, we state the result on the constructive description of the equilibrium measure $\mu$ introduced in Proposition \ref{thm:potential_theory}, under the assumption \eqref{eq:one-cut_regular_cond} in Theorem \ref{thm:one-cut_rgular_w_hard_edge}.

\begin{thm} \label{eq:find_b_psi}
	Suppose $V$ satisfies all assumptions in Theorem \ref{thm:one-cut_rgular_w_hard_edge}.
	\begin{enumerate}
		\item \label{enu:eq:find_b_psi_1}
		The parameter $c$ in \eqref{eq:relation_b_c} is the unique solution to the following equation for $x \in (0, \infty)$:
		\begin{equation} \label{eq:determine_c}
			\frac{1}{2\pi i} \oint_{\gamma(x)} \frac{\J[x]'(\xi)V'(\J[x](\xi))}{\xi - s_2(x)} d\xi = \frac{1}{s_2(x) - 1},
		\end{equation}
		where $\gamma(x) = \gamma_1(x) \cup \gamma_2(x)$ is the boundary of $D_x$, with positive orientation. Hence, $b$---the right-end point of the support of the equilibrium measure $\mu$---is determined by $b = \mathfrak{b}(c)$.
		\item \label{enu:eq:find_b_psi_2}
		The density function $\psi(x)$ (defined in \eqref{eq:density_formula}) of the equilibrium measure $\mu$ is determined by
		\begin{equation}\label{eq:psiden}
			\psi(x) = \frac{1}{4\pi^2 \sqrt{x}} \int^b_0 U'(u) F(u; x) du,
		\end{equation}
		where $U$ is defined in \eqref{eq:one-cut_regular_cond} and $F(u; x)$ is defined by \eqref{eq:simplified} and \eqref{eq:defn_F_x(uxi)}.
		\item \label{enu:eq:find_b_psi_3}
		The (random) empirical measure $\frac{1}{n} \sum_{i=1}^n \delta_{\lambda_i}$ converges to $\mu$ in probability.
	\end{enumerate}
\end{thm}
Parts \ref{enu:eq:find_b_psi_1} and \ref{enu:eq:find_b_psi_2} of Theorem \ref{eq:find_b_psi} will be proved in the remaining part of Section \ref{sec:constr_eq_measure}. Part \ref{enu:eq:find_b_psi_3} follows from \cite[Theorem 2.2]{Borot-Guionnet-Kozlowski13} (see also Proposition \ref{thm:mean} and the paragraph above it).


\subsection{A technical lemma} \label{subsec:support_of_eq_measure}

Recall the mapping $\J[x]$ defined in \eqref{eq:defn_J_c} and $\gamma(x)$ defined in Part \ref{enu:eq:find_b_psi_1} of Theorem \ref{eq:find_b_psi}.

\begin{lem} \label{lem:determine_b}
	Suppose $V$ satisfies the conditions required in Theorem \ref{thm:one-cut_rgular_w_hard_edge}. There exists a unique $x \in (0, \infty)$ such that equation \eqref{eq:determine_c} (with $x$ as the unknown) holds.
\end{lem}

\begin{proof}
	Using the formulas \eqref{eq:defn_J_c} and \eqref{eq:defn_s_2_x}, we can rewrite \eqref{eq:determine_c} as
	\begin{equation} \label{eq:determine_b_alt}
		F(x) = 2, \quad \text{where} \quad F(x) = \frac{1}{2\pi i} \oint_{\gamma'} \frac{J_x(\xi) V'(\J[x](\xi))}{\sqrt{\xi}(\xi - 1)} d\xi.
	\end{equation}
	Here, $\gamma'$ can be $\gamma(x)$ or a slightly larger contour enclosing $\gamma(x)$, provided that $V'(\J(\xi))$ is well-defined therein. Let $U(z) = V'(z) \sqrt{z}$, where $\sqrt{z}$ takes the principal branch. By direct computation, we have
	\begin{equation}
		F'(x) = \frac{1}{4\pi i} \oint_{\gamma'} U'(\J[x](\xi)) J_x(\xi) \frac{d\xi}{\xi - 1}.
	\end{equation}
	Setting $\gamma' = \gamma(x)$ and changing variables to $y = \J[x](\xi)$, we obtain
	\begin{equation} \label{eq:f'(x)}
		\begin{split}
			F'(x) = {}& \frac{1}{x \pi i} \int^{\mathfrak{b}(x)}_0 U'(y) \left( \frac{\sqrt{\Iinv_{x, -}(y)}}{\Iinv_{x, -}(y) - s_2(x)} - \frac{\sqrt{\Iinv_{x, +}(y)}}{\Iinv_{x, +}(y) - s_2(x)} \right) dy \\
			= {}-& \frac{2}{x \pi} \int^{\mathfrak{b}(x)}_0 U'(y) \Im \frac{\sqrt{\Iinv_{x, +}(y)}}{\Iinv_{x, +}(y) - s_2(x)} dy.
		\end{split}
	\end{equation}
	For $y \in (0, \mathfrak{b}(x))$, note that $U'(y) = V''(y) x^{1/2} + \frac{1}{2} V'(y) x^{-1/2} > 1$. By the definition of $\Iinv_{x, +}$ (see \eqref{eq:Iinv_+}), $\Iinv_{x, +}(y) \in \gamma_1(x)$ for all $y \in (0, \mathfrak{b}(x))$. In the proof of Part \ref{enu:lem:J_c:4} of Lemma \ref{lem:J_c} (Appendix \ref{sec:J_x_prop}), we provide a parametrization of $\gamma_1(x)$ in \eqref{eq:gamma'_1x} (it is shown there that $\gamma'_1(x)$ in \eqref{eq:gamma'_1x} equals $\gamma_1(x)$). Clearly, $\arg \Iinv_{x, +}(y) \in (0, \pi)$ and $\arg \sqrt{\Iinv_{x, +}(y)} \in (0, \pi/2)$. Additionally, $\arg(\Iinv_{x, +}(y) - s_2(x)) \in (\pi/2, \pi)$. Thus, $\arg(\sqrt{\Iinv_{x, +}(y)}/(\Iinv_{x, +}(y) - s_2(x))) \in (-\pi, 0)$, and similarly $\arg(\sqrt{\Iinv_{x, +}(y)}/(\Iinv_{x, +}(y) - s_2(x))) \in (-\pi, 0)$. We conclude that the integrand on the right-hand side of \eqref{eq:f'(x)} is positive, so $F(x)$ is a monotonically increasing function.

	In the special case where $V'(\J(s)) = C > 0$ (a constant), which is equivalent to $U'(y) = (C/2) y^{-1/2}$, we have
	\begin{equation} \label{eq:linear_V_eq_measure}
		\left. F(x) \right\rvert_{V'(\J(s)) = C} = \frac{C}{2\pi i} \oint_{\gamma'} \frac{J_x(\xi)}{\sqrt{\xi}(\xi - 1)} d\xi = \frac{C}{2\pi i} \oint_{\gamma'} \frac{x}{\xi - 1} d\xi + \frac{C}{2\pi i} \oint_{\gamma'} \frac{\arcosh \frac{\xi + 1}{\xi - 1}}{\sqrt{\xi}(\xi - 1)} d\xi = Cx.
	\end{equation}
	The first integral equals $Cx$, and the second vanishes---this is easily verified by deforming $\gamma'$ into a very large circular contour. For general $V$ satisfying assumption \eqref{eq:one-cut_regular_cond}, a comparison argument between $U'(y)$ and $(C/2) y^{-1/2}$, combined with the integral formula \eqref{eq:f'(x)}, shows that $F(x) \to 0$ as $x \to 0$ and $F(x) \to \infty$ as $x \to \infty$. Thus, $F(x) = 2$ has a unique solution on $(0, \infty)$.
\end{proof}
We denote this solution as $c'$. Later, in Section \ref{subsec:density_of_eq_measure}, we will show that the unique solution to \eqref{eq:determine_c} equals $c$---the parameter in \eqref{eq:relation_b_c}.

In the special case where $V(x) = x/\Mconst$, we have $V'(\J(s)) = \Mconst^{-1}$. By \eqref{eq:linear_V_eq_measure}, the solution to \eqref{eq:determine_c} is
\begin{equation}
	c' = 2\Mconst.
\end{equation}
This also proves \eqref{eq:b_linear_case} in Theorem \ref{thm:eq_measure_linear}.

\subsection{The $\tilde{\gfn}$-functions and the density of the equilibrium measure} \label{subsec:density_of_eq_measure}

In this subsection, we first assume the existence of $b$ and $\psi$ (the density function of $\mu_V$ on $(0, b)$) and define the functions
\begin{align} \label{eq:defn_g_gtilde}
	\gfn(z) := {}& \int^b_0 \log(z - x) \psi(x) dx, & \tilde{\gfn}(z) := {}& \int^b_0 \log(f(z) - f(x)) \psi(x) dx,
\end{align}
with branch cuts for the logarithms taken where $z \in (-\infty, x)$ and $f(z) \in (-\infty, f(x))$, respectively. Using these, we derive a Riemann-Hilbert (RH) problem satisfied by $\gfn'(z)$ and $\tilde{\gfn}'(z)$. We then solve this RH problem via direct calculation, confirm the value of $b$, and express $\psi$ in a computable form. This completes the proofs of Theorems \ref{thm:one-cut_rgular_w_hard_edge} and \ref{eq:find_b_psi}.

Note that $\gfn(z)$ and $\tilde{\gfn}(z)$ (defined in \eqref{eq:defn_g_gtilde}) are analytic on $\compC \setminus (-\infty, b]$ and $\paraP \setminus (-\pi^2/4, b]$, respectively.

\begin{lem} \label{lem:regular_g}
	Suppose $V$ is one-cut regular with a hard edge. Then
	\begin{enumerate}
		\item \label{enu:lem:regular_g:1}
		For $x \in (-\infty, 0)$, $\gfn_{\pm}(x)$ are continuous, and
		\begin{equation}
			\gfn_+(x) = \gfn_-(x) + 2\pi i;
		\end{equation}
		for $x \in (-\pi^2/4, 0)$, $\tilde{\gfn}_{\pm}(x)$ are continuous, and
		\begin{equation}
			\tilde{\gfn}_+(x) = \tilde{\gfn}_-(x) + 2\pi i;
		\end{equation}
		and for $z \in \rho$ (the lower parabola defined in \eqref{eq:defn_rho}), $\tilde{\gfn}_{\pm}(z)$ are continuous, and
		\begin{equation}
			\tilde{\gfn}_-(\bar{z}) = \tilde{\gfn}_+(z) + 2\pi i,
		\end{equation}
		where $\bar{z}$ is the complex conjugate of $z$ (lying on the upper parabola), with both $\rho$ and $\bar{\rho}$ oriented from left to right.
		\item
		For $x \in (0, b)$, we have
		\begin{equation} \label{eq:psi_and_G}
			-\frac{1}{2\pi i}\left(\gfn_+'(x) - \gfn'_-(x)\right) = -\frac{1}{2\pi i}\left(\tilde{\gfn}'_+(x) - \tilde{\gfn}'_-(x)\right) > 0,
		\end{equation}
		and the left-hand side equals $\psi(x)$.
		\item \label{enu:regular_g:3} 
		As $z \to b$, the limits of $\gfn(z), \tilde{\gfn}(z), \gfn'(z)$ and $\tilde{\gfn}'(z)$ exist. As $x \to b$, $\gfn'(z) - \gfn'(b) = \bigO(|z - b|^{1/2})$, $\tilde{\gfn}'(z) - \tilde{\gfn}'(b) = \bigO(|z - b|^{1/2})$, and
		\begin{equation} \label{eq:limit_g_at_b}
			\lim_{x \to b_-} \frac{i\left(\gfn_+'(x) - \gfn'_-(x)\right)}{\sqrt{b - x}} = \lim_{x \to b_-} \frac{i\left(\tilde{\gfn}'_+(x) - \tilde{\gfn}'_-(x)\right)}{\sqrt{b - x}} \in (0, +\infty),
		\end{equation}
		As $z \to 0$ in $\compC_+$ or $\compC_-$, the limits of $\gfn(z)$ and $\tilde{\gfn}(z)$ exist. As $x \to 0$, $\gfn'(z) = \bigO(|z|^{-1/2})$, $\tilde{\gfn}'(z) = \bigO(|z|^{-1/2})$, and
		\begin{equation} \label{eq:limit_g_at_0}
			\lim_{x \to 0_+} i\left(\gfn_+'(x) - \gfn'_-(x)\right)\sqrt{x} = \lim_{x \to 0_+} i\left(\tilde{\gfn}'_+(x) - \tilde{\gfn}'_-(x)\right)\sqrt{x} \in (0, +\infty).
		\end{equation}
		The two limits in \eqref{eq:limit_g_at_b} and \eqref{eq:limit_g_at_0} are $2\pi \psi_b$ and $2\pi \psi_0$, respectively.
		\item
		As $z \to \infty$ in $\compC$, $\gfn'(z) = z^{-1} + \bigO(z^{-2})$. As $f(z) \to \infty$ (i.e., $\Re z \to \infty$), $\tilde{\gfn}'(z) = z^{-1/2}\left(1 + \bigO(f(z)^{-1})\right)$.
		\item \label{enu:regular_g:5}
		For $x \in [0, b]$, there exists a constant $\ell$ such that
		\begin{equation} \label{eq:Euler-Lagrange}
			\gfn_{\pm}(x) + \tilde{\gfn}_{\mp}(x) - V(x) - \ell = 0.
		\end{equation}
		For $x \in (b, \infty)$, we have
		\begin{equation} \label{eq:gfn_outer}
			\gfn_{\pm}(x) + \tilde{\gfn}_{\mp}(x) - V(x) - \ell < 0.
		\end{equation}
	\end{enumerate}
	Conversely, if functions $\gfn(x)$ and $\tilde{\gfn}(x)$ (analytic on $\compC \setminus (-\infty, b]$ and $\paraP \setminus (-\pi^2/4, b]$, respectively) satisfy all the properties above, then $V$ is one-cut regular with a hard edge, and its equilibrium measure is $d\mu(x) = \psi(x) dx$ (supported on $[0, b]$) with
	\begin{equation} \label{eq:psi_by_G}
		\psi(x) = -\frac{1}{2\pi i}\left(G_+(x) - G_-(x)\right) = -\frac{1}{2\pi i}\left(\tilde{G}_+(x) - \tilde{G}_-(x)\right),
	\end{equation}
	where
	\begin{align} \label{eq:defn_G}
		G(z) = {}& \gfn'(z) = \int^b_0 \frac{\psi(x) dx}{z - x}, & \tilde{G}(z) = {}& \tilde{\gfn}'(z) = \int^b_0 \frac{f'(z) \psi(x) dx}{f(z) - f(x)}.
	\end{align}
\end{lem}

The proof of Lemma \ref{lem:regular_g} is straightforward and is omitted here. Below, we construct $\gfn(z)$ and $\tilde{\gfn}(z)$ satisfying these properties, under the conditions on $V$ given in Theorem \ref{thm:one-cut_rgular_w_hard_edge}. Note that $b$ is unknown during the construction and must be determined.

To construct $\gfn(z)$ and $\tilde{\gfn}(z)$, it suffices to construct their derivatives $G(z)$ and $\tilde{G}(z)$. Recall the function $\mathfrak{b}(x)$ (defined in \eqref{eq:defn_b(x)}) and $s_1(x), s_2(x)$ (defined in Lemma \ref{lem:J_c}). We define two RH problems for $(H^{(x)}(z), \tilde{H}^{(x)}(z))$ and $N^{(x)}(s)$ with a parameter $x > 0$:

\begin{RHP} \hfill \label{RHP:H}
	\begin{enumerate}
		\item
		$H^{(x)}(z)$ is analytic in $\compC \setminus [0, \mathfrak{b}(x)]$, and $\tilde{H}^{(x)}(z)$ is analytic in $\paraP \setminus [0, \mathfrak{b}(x)]$.
		\item \label{enu:RHP:H_boundary}
		The boundary conditions are as follows: $H^{(x)}_{\pm}(z)$ and $\tilde{H}^{(x)}_{\pm}(z)$ are continuous on $(0, \mathfrak{b}(x))$ \footnote{In all RH problems in this paper, the boundary values of functions are continuous on both sides of the jump curves, unless stated otherwise. This continuity is omitted in later RH problems.} and
		\begin{align}
			H^{(x)}(z) = {}& \frac{1}{z} + \bigO(z^{-2}), && \text{as $z \to \infty$}, \label{eq:boundary_at_infty} \\
			\tilde{H}^{(x)}(z) = {}& z^{-\frac{1}{2}} \left(1 + \bigO(f(z)^{-1}\right), && \text{as $f(z) \to \infty$ (i.e., $\Re z \to +\infty$)}, \label{eq:boundary_at_infty_H} \\
			H^{(x)}(z) = {}& \bigO(1), \quad \tilde{H}^{(x)}(z) = \bigO(1), && \text{as $z \to \mathfrak{b}(x)$}, \\
			H^{(x)}(z) = {}& \bigO(z^{-\frac{1}{2}}), \quad \tilde{H}^{(x)}(z) = \bigO(z^{-\frac{1}{2}}), && \text{as $z \to 0$}.
		\end{align}
		\item \label{enu:RHP:H_jump}
		For $z \in (0, \mathfrak{b}(x))$, we have
		\begin{equation} \label{eq:H^t_jump}
			H^{(x)}_{\pm}(z) + \tilde{H}^{(x)}_{\mp}(z) - V'(z) = 0.
		\end{equation}
		\item \label{enu:RHP:H_4}
		$\tilde{H}^{(x)}(z)$ is continuous up to the boundary $\rho \cup \{ 0 \} \cup \bar{\rho}$ of $\paraP$. For $z \in \rho \subset \compC_-$ such that $z = \lim_{\epsilon \to 0_+} \J[x]'(y + \epsilon i)$ (with $y \in (0, 1)$), we have
		\begin{equation}
			\lim_{\epsilon \to 0_+} \tilde{H}^{(x)}(z + \epsilon i) \J[x]'(y + \epsilon i) = \lim_{\epsilon \to 0_+} \tilde{H}^{(x)}(\bar{z} - \epsilon i) \J[x]'(y - \epsilon i).
		\end{equation}
	\end{enumerate}
\end{RHP}

Our RH problem \ref{RHP:H} is motivated by the properties of $(G(z), \tilde{G}(z))$ (defined in \eqref{eq:defn_G}). Analogous to \eqref{eq:psi_by_G}, we define for $y \in (0, \mathfrak{b}(x))$
\begin{equation} \label{eq:psi^c'}
	\psi^{(x)}(y) = -\frac{1}{2\pi i}\left(H^{(x)}_+(y) - H^{(x)}_-(y)\right) = -\frac{1}{2\pi i}\left(\tilde{H}^{(x)}_+(y) - \tilde{H}^{(x)}_-(y)\right),
\end{equation}
and, analogous to \eqref{eq:defn_G}, we have
\begin{align} \label{eq:defn_H_by_psi}
	H^{(x)}(z) = {}& \int^{\mathfrak{b}(x)}_0 \frac{\psi^{(x)}(y) dy}{z - y}, & \tilde{H}^{(x)}(z) = {}& \int^{\mathfrak{b}(x)}_0 \frac{f'(z) \psi^{(x)}(y) dy}{f(z) - f(y)}.
\end{align}

\begin{RHP} \hfill \label{RHP:N}
	\begin{enumerate}
		\item \label{enu:RHP_N_1}
		$N^{(x)}(s)$ is analytic in $\compC \setminus \gamma(x)$, where $\gamma(x) = \gamma_1(x) \cup \gamma_2(x)$ (defined in Theorem \ref{eq:find_b_psi}).
		\item \label{enu:RHP_N_3}
		$N^{(x)}_{\pm}(s)$ is bounded on $\gamma_1(x)$ and $\gamma_2(x)$. $N^{(x)}(s)$ is bounded as $s \to s_1(x)$, $s \to s_2(x)$, or $s \to 0$, and has the behavior
		\begin{equation} \label{eq:limit_N(s)_at_infty_weak}
			N^{(x)}(s) = \bigO(s^{-1}), \quad \text{as} \quad s \to \infty.
		\end{equation}
		\item \label{enu:RHP_N_2}
		$N^{(x)}(s)$ satisfies the jump condition
		\begin{equation}
			N^{(x)}_+(s) + N^{(x)}_-(s) = \frac{s - 1}{s - s_2(x)} \J[x]'(s) V'(\J[x](s)), \quad s \in \gamma(x) \setminus \{ s_1(x), s_2(x) \},
		\end{equation}
		where $\J[x](s)$ is defined in \eqref{eq:defn_J_c} and $s_2(x)$ is defined in \eqref{eq:defn_s_2_x}.
	\end{enumerate}
\end{RHP}

If $(H^{(x)}(z), \tilde{H}^{(x)}(z))$ solves RH problem \ref{RHP:H}, then the function $\tilde{N}^{(x)}(s)$---defined by
\begin{equation} \label{eq:N_in_H}
	\tilde{N}^{(x)}(s)= \frac{s - 1}{s - s_2(x)} \times
	\begin{cases}
		\J[x]'(s) H^{(x)}(\J[x](s)), & s \in \compC \setminus \overline{D_{x}}, \\
		\J[x]'(s) \tilde{H}^{(x)}(\J[x](s)), & s \in D_{x} \setminus [0, 1],
	\end{cases}
\end{equation}
---solves RH problem \ref{RHP:N}. (Although $\tilde{N}^{(x)}(s)$ is undefined on $(0, 1)$ via \eqref{eq:N_in_H}, it extends naturally to $(0, 1)$ by analytic continuation, thanks to Item \ref{enu:RHP:H_4} of RH problem \ref{RHP:H}.) Moreover, $N^{(x)}(s)$ (given by \eqref{eq:N_in_H}) satisfies a stronger condition than \eqref{eq:limit_N(s)_at_infty_weak}:
\begin{align} \label{eq:limit_N(s)_at_infty}
	\tilde{N}^{(x)}(1) = {}& (s_2 - 1)^{-1}, &  \tilde{N}^{(x)}(s) = {}& s^{-1} + \bigO(s^{-2}), \quad s \to \infty.
\end{align}

For each $x > 0$, RH problem \ref{RHP:N} has at most one solution. Suppose $N^{(x), 1}(s)$ and $N^{(x), 2}(s)$ are both solutions. Define
\begin{equation}
	M^{(x)}(s) =
	\begin{cases}
		N^{(x), 1}(s) - N^{(x), 2}(s), & s \in \compC \setminus \bar{D}, \\
		-N^{(x), 1}(s) + N^{(x), 2}(s), & s \in D.
	\end{cases}
\end{equation}
After analytic continuation on $\gamma(x) \setminus \{ s_1(x), s_2(x) \}$, $M^{(x)}(s)$ is analytic in $\compC \setminus \{ s_1(x), s_2(x) \}$, bounded as $s \to s_1(x), s_2(x)$, and $M^{(x)}(s) \to 0$ as $s \to \infty$. By Liouville's theorem, $M^{(x)}(s) = 0$, so $\tilde{N}^{(x), 1}(s) = N^{(x), 2}(s)$.

The unique solution to RH problem \ref{RHP:N} has the explicit formula
\begin{equation} \label{eq:defn_N(s)}
	N^{(x)}(s) =
	\begin{cases}
		-\frac{1}{2\pi i} \int_{\gamma(x)} \frac{(\xi - 1) \J[x]'(\xi) V'(\J[x](\xi))}{(\xi - s_2(x))(\xi - s)} d\xi, & s \in \compC \setminus \overline{D_{x}}, \\
		\frac{1}{2\pi i} \int_{\gamma(x)} \frac{(\xi - 1) \J[x]'(\xi) V'(\J[x](\xi))}{(\xi - s_2(x))(\xi - s)} d\xi, & s \in D_{x} \setminus [0, 1].
	\end{cases}
\end{equation}
By direct computation,
\begin{align}
	N^{(x)}(1) = {}& (s_2(x) - 1)^{-1} C_x, & \lim_{s \to \infty} N^{(x)}(s) = {}& C_x s^{-1} + \bigO(s^{-2}), 
\end{align}
where
\begin{equation}
	\begin{split}
		C_x = {}& \frac{1}{2\pi i} \int_{\gamma(x)} \frac{(\xi - 1) \J[x]'(\xi) V'(\J[x](\xi))}{\xi - s_2(x)} d\xi 
		= {} \frac{s_2(x) - 1}{2\pi i} \int_{\gamma(x)} \frac{\J[x]'(\xi) V'(\J[x](\xi))}{\xi - s_2(x)} d\xi.
	\end{split}
\end{equation}
We conclude that RH problem \ref{RHP:H} has a solution (implying RH problem \ref{RHP:N} has a solution satisfying \eqref{eq:limit_N(s)_at_infty}) only if $C_x = 1$---which is equivalent to equation \eqref{eq:determine_c} in Theorem \ref{eq:find_b_psi}. Thus, RH problem \ref{RHP:N} has a solution only if $x = c'$ (the unique solution to \eqref{eq:determine_c}, as shown in Lemma \ref{lem:determine_b}). Under the conditions of Theorem \ref{thm:one-cut_rgular_w_hard_edge}, and assuming $V$ is one-cut regular with a hard edge, the only possible right-end point of the equilibrium measure's support is $b' = b(c')$. The only candidates for $G, \tilde{G}$ (defined in \eqref{eq:defn_g_gtilde} and \eqref{eq:defn_G}) are
\begin{align}
	H^{(c')}(z) = {}& \frac{2}{c'} \sqrt{\frac{\Iinv_{c', 1}(z)}{z}} N^{(c')}(\Iinv_{c', 1}(z)), &
	\tilde{H}^{(c')}(z) = {}& \frac{2}{c'} \sqrt{\frac{\Iinv_{c', 2}(z)}{z}} N^{(c')}(\Iinv_{c', 2}(z)),
\end{align}
and the only candidate for the density function is $\psi^{(c')}(x)$ (defined in \eqref{eq:psi^c'}).

The remaining step is to show that $\psi^{(c')}(x) dx$ on $[0, b']$ satisfies Requirement \ref{req:one-cut_reg}. Identities \eqref{eq:boundary_at_infty} and \eqref{eq:boundary_at_infty_H} (Item \ref{enu:RHP:H_boundary} of RH problem \ref{RHP:H}) imply the total mass of the (possibly signed) measure $\psi^{(c')}(x) dx$ is $1$. Identity \eqref{eq:defn_ell} (Part \ref{enu:req:one-cut_reg:4} of Requirement \ref{req:one-cut_reg}) follows from \eqref{eq:H^t_jump} (Item \ref{enu:RHP:H_jump} of RH problem \ref{RHP:H}). Thus, we only need to verify Parts \ref{enu:req:one-cut_reg:2}, \ref{enu:req:one-cut_reg:5}, and \ref{enu:req:one-cut_reg:3} of Requirement \ref{req:one-cut_reg}.

To do this, for $s \in \compC \setminus \overline{D}$, we express $N^{(c')}(s)$ in terms of $U(u)$ (defined in \eqref{eq:one-cut_regular_cond}):
\begin{equation}
	\begin{split}
		N^{(c')}(s) = {}& \frac{-1}{2\pi i} \frac{c'}{2} \oint_{\gamma(c')} \frac{V'(\J[c'](\xi)) \sqrt{\J[c'](\xi)}}{\sqrt{\xi}(\xi - s)} d\xi \\
		= {}& \frac{-1}{2\pi i} \frac{c'}{2} \int_{\gamma_1(c')} \frac{V'(\J[c'](\xi)) \sqrt{\J[c'](\xi)}}{\sqrt{\xi}(\xi - s)} d\xi + \frac{1}{2\pi i} \frac{c'}{2} \int_{\gamma_2(c')} \frac{V'(\J[c'](\xi)) \sqrt{\J[c'](\xi)}}{\sqrt{\xi}(\xi - s)} d\xi \\
		= {}& \frac{-1}{2\pi i} \frac{c'}{2} \left[ \int^{b'}_0 \frac{U(u) \Iinv'_{c', +}(u)}{\sqrt{\Iinv_{c', +}(u)} (\Iinv_{c', +}(u) - s)} du - \int^{b'}_0 \frac{U(u) \Iinv'_{c', -}(u)}{\sqrt{\Iinv_{c', -}(u)} (\Iinv_{c', -}(u) - s)} du\right] \\
		= {}& \frac{-1}{2\pi i} c' \left[ \int^{b'}_0 \frac{U(u) \sqrt{\Iinv_{c', +}(u)}'}{(\sqrt{\Iinv_{c', +}(u)} + \sqrt{s})(\sqrt{\Iinv_{c', +}(u)} - \sqrt{s})} du \right. \\
		& \phantom{\smash{\frac{-1}{2\pi i} c'}}
		- \left. \int^{b'}_0 \frac{U(u) \sqrt{\Iinv_{c', -}(u)}'}{(\sqrt{\Iinv_{c', -}(u)} + \sqrt{s})(\sqrt{\Iinv_{c', -}(u)} - \sqrt{s})} du \right] \\
		= {}& \frac{-1}{2\pi i} \frac{c'}{2\sqrt{s}} \int^{b'}_0 U(u) \frac{d}{du} \left( \log \frac{\sqrt{\Iinv_{c', +}(u)} - \sqrt{s}}{\sqrt{\Iinv_{c', +}(u)} + \sqrt{s}} - \log \frac{\sqrt{\Iinv_{c', -}(u)} - \sqrt{s}}{\sqrt{\Iinv_{c', -}(u)} + \sqrt{s}} \right) du \\
		= {}& \frac{1}{2\pi i} \frac{c'}{2\sqrt{s}} \int^{b'}_0 U'(u) \left( \log \frac{\sqrt{\Iinv_{c', +}(u)} - \sqrt{s}}{\sqrt{\Iinv_{c', +}(u)} + \sqrt{s}} - \log \frac{\sqrt{\Iinv_{c', -}(u)} - \sqrt{s}}{\sqrt{\Iinv_{c', -}(u)} + \sqrt{s}} \right) du,
	\end{split}
\end{equation}
where $\sqrt{s}$ takes the principal branch. 
For $x \in (0, b)$,
\begin{equation} \label{eq:integral_rep_psi}
	\begin{split}
		& \psi^{(c')}(x) \\
		= {}& -\frac{x^{-\frac{1}{2}}}{4\pi^2} \Re \lim_{\epsilon \to 0_+} \int^{b'}_0 U'(u) \log \frac{(\sqrt{\Iinv_{c', +}(u)} - \sqrt{\Iinv_{c', 1}(x + \epsilon i)})(\sqrt{\Iinv_{c', -}(u)} + \sqrt{\Iinv_{c', 1}(x + \epsilon i)})}{(\sqrt{\Iinv_{c', +}(u)} + \sqrt{\Iinv_{c', 1}(x + \epsilon i)})(\sqrt{\Iinv_{c', -}(u)} - \sqrt{\Iinv_{c', 1}(x + \epsilon i)})} du \\
		= {}& \frac{x^{-\frac{1}{2}}}{4\pi^2} \int^{b'}_0 U'(u) F_{c'}(u; x) du,
	\end{split}
\end{equation}
where $F_{c'}(u; x)$ is defined in \eqref{eq:defn_F_x(uxi)}. Since $F_{c'}(u; x)$ is continuous on $u \in [0, x) \cup (x, b')$ and behaves like $\bigO(\log(|u - x|))$ as $u \to x$, the integral \eqref{eq:integral_rep_psi} is well-defined.

\paragraph{Positivity and regularity in the bulk}

For all $y \in (0, b')$, $\Re \Iinv_{c', +}(y) = \Re \Iinv_{c', -}(y) > 0$ and $\Im \Iinv_{c', +}(y) = -\Im \Iinv_{c', -}(y) > 0$. Thus, for all $x, u \in (0, b')$,
\begin{align}
	\left| \frac{\sqrt{\Iinv_{c', +}(u)} + \sqrt{\Iinv_{c', +}(x)}}{\sqrt{\Iinv_{c', -}(u)} + \sqrt{\Iinv_{c', +}(x)}} \right| > {}& 1, & \left| \frac{\sqrt{\Iinv_{c', -}(u)} - \sqrt{\Iinv_{c', +}(x)}}{\sqrt{\Iinv_{c', +}(u)} - \sqrt{\Iinv_{c', +}(x)}} \right| > {}& 1.
\end{align}
Hence, as a function of $u$,
\begin{equation} \label{eq:positivity_F}
	F_{c'}(u; x) > 0, \quad \text{for all} \quad u \in (0, x) \cup (x, b').
\end{equation}
Combined with the assumption $U'(x) > 0$ (from \eqref{eq:one-cut_regular_cond}), this implies $\psi^{(c')}(x) > 0$ for all $x \in (0, b')$.

\paragraph{Regularity at the edges}

We need to show $\psi^{(c')}(x)$ satisfies $\lim_{x \to b'_-} \psi^{(c')}(x)/\sqrt{b' - x} = c_1 (1 + \bigO(x - b'))$ and $\lim_{x \to 0^+} \psi^{(c')}(x)\sqrt{x} = c_2 (1 + \bigO(x))$ for some $c_1, c_2 > 0$. By \eqref{eq:integral_rep_psi}, this is equivalent to showing the integral
\begin{equation} \label{eq:integral_for_psi}
	\int^{b'}_0 U'(u) F_{c'}(u; x) du
\end{equation}
decays like $\sqrt{b' - x}$ as $x \to b'$ and approaches a non-zero limit as $x \to 0$. This relies on the continuity and positivity properties of $U'(u)$ and $F_{c'}(u; x)$. They are both continuous functions as $u \in (0, b')$, and we have the positivity results \eqref{eq:one-cut_regular_cond} for $U'(u)$ on $(0, b')$ and \eqref{eq:positivity_F} for $F_{c'}(u; x)$ as $u, x \in (0, b')$.

Moreover, since $V(z)$ is analytic at $0$, $U(x)$ either converges to $0$ (if $V'(0) = 0$) or behaves like $x^{-1/2}$ as $u \to 0_+$ (if $V'(0) > 0$).

For $x \in (0, b')$ near $0$, properties of $\Iinv_{c', \pm}(x)$ (given in Appendix \ref{sec:J_x_prop}) show the function 
\begin{equation}
	\frac{F_{c'}(u; x)}{\sqrt{\log^2(|u - x|) + 1}}
\end{equation}
is uniformly bounded and converges uniformly to a non-vanishing limit as $x \to 0$. Thus, \eqref{eq:integral_for_psi} converges to a non-zero limit as $x \to 0$. For $x \in (0, b')$ near $b'$, we use $d_1(x)$ (defined in \eqref{eq:defn_d_1}) to write
\begin{equation}
	F_{c'}(u; x) = \log \left| \frac{2b' + d_1(c') (\sqrt{b' - u} + \sqrt{b' - x}) i}{2b' + d_1(c') (\sqrt{b' - u} - \sqrt{b' - x}) i} \frac{\sqrt{b' - u} + \sqrt{b' - x}}{\sqrt{b' - u} - \sqrt{b' - x}} \right| \left(1 + \sqrt{b' - u} f(u; x)\right),
\end{equation}
where $f(u; x)$ is continuous on $[0, b']$ and converges uniformly to a limit function as $x \to b'$. Hence, the product of the integral in \eqref{eq:integral_for_psi} and $(b' - x)^{-1/2}$ converges to a non-zero limit as $x \to b'_-$.

\paragraph{Regularity away from the support}

We aim to show that if $\psi(x)$ is defined as $\psi^{(c')}(x)$ in \eqref{eq:psi^c'} and $b = b'$, then Item \ref{enu:req:one-cut_reg:5} of Requirement \ref{req:one-cut_reg} holds. Define
\begin{equation}
	g(x) = H^{(c')}(x) + \tilde{H}^{(c')}(x) - V'(x), \quad x \in (b', +\infty).
\end{equation}
We need to show $g(x) < 0$ for all $x > b'$. By \eqref{eq:H^t_jump}, the continuity of $H^{(c')}(z)$, and the continuity of $\tilde{H}^{(c')}(z)$ at $z = b'$, $\lim_{x \to b_+} g(x) = 0$. Thus, it suffices to show $\frac{d}{dx}\left(g(x) \sqrt{x}\right) < 0$ on $(b', \infty)$. Using \eqref{eq:defn_H_by_psi}, we express $\frac{d}{dx}\left(g(x) \sqrt{x}\right)$ as
\begin{equation}
	-\int^{b'}_0 \left( \frac{x + t}{2\sqrt{x}(x - t)^2} + \frac{\sinh^2(\sqrt{x}) + \cosh(2\sqrt{x}) \sinh^2(\sqrt{t})}{2\sqrt{x}(\sinh^2(\sqrt{x}) - \sinh^2(\sqrt{t}))} \right) \psi^{(c')}(t) dt - U'(x),
\end{equation}
which is negative for all $x \in (b', +\infty)$ due to the positivity of $\psi^{(c')}(t)$ and $U'(x)$.

\begin{proof}[Proof of Theorems \ref{thm:one-cut_rgular_w_hard_edge} and \ref{eq:find_b_psi}]
	We only need to prove Theorem \ref{eq:find_b_psi}, which is a quantitative version of Theorem \ref{thm:one-cut_rgular_w_hard_edge}.

        By the computation above in this subsection, we find that with $c'$ being the unique solution of \eqref{eq:determine_c}, the explicitly constructed density function $\psi^{(c')}(x)$ and the measure it defines satisfy all the items in Requirement \ref{req:one-cut_reg}, so it is the desired equilibrium measure. (The uniqueness is guaranteed by Proposition \ref{thm:potential_theory}.) Thus, $c = c'$ and $b = b(c') = b'$.
\end{proof}

\begin{proof}[Proof of \eqref{eq:psi_linear_case} in Theorem \ref{thm:eq_measure_linear}]
Recall from the end of Section \ref{subsec:support_of_eq_measure} that $c = 2\Mconst$ for the special case $V(x) = \frac{x}{\Mconst}$. For $s \in D$, we have
\begin{subequations}
	\begin{align}
		N(s) = {}& \frac{1}{2\pi i} \oint_{\gamma} \frac{\Mconst^{-1} \frac{c}{4}\left(c\sqrt{\xi} + \arcosh \frac{\xi + 1}{\xi - 1}\right)}{\sqrt{\xi}(\xi - s)} d\xi \\
		= {}& \frac{c^2}{4\Mconst} \frac{1}{2\pi i} \oint_{\gamma} \frac{d \xi}{\xi - s} + \frac{c}{4\Mconst} \frac{1}{2\pi i} \oint_{\gamma} \frac{\arcosh \frac{\xi + 1}{\xi - 1}}{\sqrt{\xi}(\xi - s)} d\xi, \label{eq:N(s)_linear_inside}
	\end{align}
\end{subequations}
The first term in \eqref{eq:N(s)_linear_inside} equals $\frac{c^2}{4\Mconst} s$, and the second vanishes (verified by deforming $\gamma$ into a large circle). Similarly, evaluating $N(s)$ for $s \in \compC \setminus \overline{D}$ (and recalling $c = 2\Mconst$) gives
\begin{equation}
	N(s) =
	\begin{cases}
		\Mconst, & s \in D, \\
		\frac{1}{2\sqrt{s}} \arcosh \frac{s + 1}{s - 1}, & s \in \compC \setminus \overline{D}.
	\end{cases}
\end{equation}
Thus,
\begin{align} \label{eq:G_linear_V}
	G(z) = {}& \frac{1}{\Mconst} - \sqrt{\frac{\Iinv_1(z)}{z}}, & \tilde{G}(z) = {}& \sqrt{\frac{\Iinv_2(z)}{z}},
\end{align}
and \eqref{eq:psi_linear_case} follows from \eqref{eq:psi_by_G}.
\end{proof} 

\section{Asymptotic analysis for $p^{(n)}_{n + k}(x)$} \label{sec:asy_p}

\subsection{RH problem of the polynomials}

Consider the following modified Cauchy transform of $p_j$:
\begin{equation}
	C p_j(z) := \frac{1}{2\pi i} \int_{\realR_+} \frac{p_j(x)}{f(x) - f(z)} W^{(n)}_{\alpha}(x) dx,
\end{equation}
which is well defined for $z \in \paraP \setminus \realR_+$. Since $W^{(n)}_{\alpha}(x)$ is real analytic and vanishes rapidly as $x \to +\infty$, we have the following asymptotic expansion for $C p_j(z)$ as $z \in \paraP \setminus \realR_+$ and $\Re z \to +\infty$:
\begin{equation}
	\begin{split}
		C p_j(z) = {}& \frac{-1}{2\pi i f(z)} \int_{\realR_+} \frac{p_j(x)}{1 - f(x)/f(z)} W^{(n)}_{\alpha}(x) dx \\
		= {}& \frac{-1}{2\pi i f(z)} \sum^M_{k = 0} \left( \int_{\realR_+} p_j(x) f^k(x) W^{(n)}_{\alpha}(x) dx \right) f^{-(k + 1)}(z) + \bigO( f^{-(M + 2)}(z)),
	\end{split}
\end{equation}
for any $M \in \natN$ and uniformly in $\Im z$. Thus, due to orthogonality,
\begin{equation}
	C p_j(z) = \frac{-h^{(n)}_j}{2\pi i} f^{-(j + 1)}(z) + \bigO(f^{-(j + 2)}(z)),
\end{equation}
where $h^{(n)}_j$ is given in \eqref{eq:defn_h}. 

Hence, we conclude that if we define the array
\begin{equation} \label{eq:defn_Y}
	Y(z) = Y^{(j, n)}(z) := (p_j(z), C p_j(z)),
\end{equation}
then they satisfy the following conditions:

\begin{RHP} \hfill \label{RHP:Y}
	\begin{enumerate}
		\item
		$Y = (Y_1, Y_2)$, where $Y_1$ is analytic on $\compC$, and $Y_2$ is analytic on $\paraP \setminus \realR_+$.
		\item
		With the standard orientation of $\realR_+$,
		\begin{equation}
			Y_+(x) = Y_-(x)
			\begin{pmatrix}
				1 & W^{(n)}_{\alpha}(x)/f'(x) \\
				0 & 1
			\end{pmatrix},
			\quad \text{for $x \in \realR_+$}.
		\end{equation}
		\item \label{enu:RHP:Y_infty}
		As $z \to \infty$ in $\compC$, $Y_1(z) = z^j + \bigO(z^{j - 1})$.
		\item \label{enu:RHP:Y_infty_P}
		As $f(z) \to \infty$ in $\paraP$ (i.e., $\Re z \to +\infty$), $Y_2(z) = \bigO(f^{-(j + 1)}(z))$. 
		\item \label{enu:RHP:Y_zero} \setcounter{favoriteitem}{\value{enumi}}
		As $z \to 0$ in $\compC$ or $\paraP$,
		\begin{align} \label{eq:Y_bd_0}
			Y_1(z) = {}& \bigO(1), &  Y_2(z) = {}&
			\begin{cases}
				\bigO(1), & \alpha > 0, \\
				\bigO(\log z), & \alpha = 0, \\
				\bigO(z^{\alpha}), & \alpha \in (-1, 0).
			\end{cases}
		\end{align}
		\item \label{enu:RHP:Y_bd_rho}
		At $z \in \rho \cup \{ -\pi^2/4 \} \cup \bar{\rho}$, the limit $Y_2(z) := \lim_{w \to z \text{ in } \paraP} Y_2(w)$ exists and is continuous, and 
		\begin{equation} \label{eq:Y_bd_cond}
			Y_2(z) = Y_2(\bar{z}).
		\end{equation}
	\end{enumerate}
\end{RHP}

Below, we take $j = n + k$ where $k$ is a constant integer, and our goal is to obtain the asymptotics for $Y = Y^{(n + k, n)}$ as $n \to \infty$.

\subsubsection{Uniqueness of RH problem \ref{RHP:Y}} \label{subsubsec:uniqueness}

For later use in the proof of Lemma \ref{lem:uniqueness_of_R}, we consider the uniqueness of a weaker form of RH problem \ref{RHP:Y} where Item \ref{enu:RHP:Y_zero} of RH problem \ref{RHP:Y} is replaced by
\begin{em}
	\begin{enumerate} [start=\value{favoriteitem}, label=\arabic{*}'.]
		\item \label{enu:RHP:Y_zero_alt}
		As $z \to 0$ in $\compC$ or $\paraP$, 
		\begin{align} 
			Y_1(z) = {}&
			\begin{cases}
				\bigO(z^{-\alpha}), & \alpha > 0 \text{ and $\arg z \in [-\frac{\pi}{3}, \frac{\pi}{3}]$}, \\
				\bigO(\log z), & \alpha = 0, \\
				\bigO(1), & \alpha \in (-1, 0) \text{ or $\alpha > 0$ and $\arg (-z) \in (-\frac{2\pi}{3}, \frac{2\pi}{3})$},
			\end{cases} \\
			Y_2(z) = {}&
			\begin{cases}
				\bigO(1), & \alpha > 0, \\
				\bigO(\log z), & \alpha = 0, \\
				\bigO(z^{\alpha}), & \alpha \in (-1, 0).
			\end{cases}
		\end{align}
		and $Y_1(z)$ and $Y_2(z)$ are allowed to have a mild blowup at $b$, such that
		\begin{align}
			Y_1(z) = {}& \bigO((z - b)^{-\frac{1}{2}}), & Y_2(z) = {}& \bigO((z - b)^{-\frac{1}{2}}), & \text{as $z \to b$},
		\end{align}
	\end{enumerate}
\end{em}

\begin{proof}[Proof of the uniqueness of RH problem \ref{RHP:Y} with Item \ref{enu:RHP:Y_zero} weakened to \ref{enu:RHP:Y_zero_alt}]
  Despite nominally $Y_1$ may have blowups at $0$ and $b$, since $Y_1$ has pole singularities at $0, b$, and it may only blow up at $b$ like an inverse square root and may only blow up at $0$ like an inverse logarithm at $0$ in the sector $\arg (-z) \in (-\frac{2\pi}{3}, \frac{2\pi}{3})$, we find that $Y_1(z)$ actually has no singularities at $0, b$. Hence, by Item \ref{enu:RHP:Y_infty} of RH problem \ref{RHP:Y}, we find that $Y_1$ is a polynomial of degree $j$. Next, define
	\begin{equation}
		Z_2(z) = Y_2(z) - \frac{1}{2\pi i} \int_{\realR_+} \frac{p_j(x)}{f(x) - f(z)} W^{(n)}_{\alpha}(x) dx.
	\end{equation}
	Then $Z_2(z)$ has only a trivial jump on $\realR_+$, so it can be defined analytically on $\paraP \setminus \{ 0, b \}$. By an argument similar to that applied to $Y_1$ above, $b$ is not a singular point of $Z_2(z)$, and $Z_2(z)$ can be defined analytically on $\paraP \setminus \{ 0 \}$. Now consider the function $Z_2(f^{-1}(z))$ where $z \in \compC \setminus (-\infty, -1]$. By Item \ref{enu:RHP:Y_bd_rho} of RH problem \ref{RHP:Y}, $Z_2(f^{-1}(z))$ can be extended analytically to $\compC \setminus \{ 0 \}$. By \eqref{eq:Y_bd_0}, $Z_2(f^{-1}(z)) = \mathfrak{o}(z^{-1})$ as $z \to 0$. Hence, $Z_2(f^{-1}(z))$ is analytic on $\compC$, and we conclude $Z_2(z) = 0$ by Item \ref{enu:RHP:Y_infty_P} of RH problem \ref{RHP:Y}. Finally, Item \ref{enu:RHP:Y_infty_P} of RH problem \ref{RHP:Y} implies the orthogonality of $Y_1(z)$, so the uniqueness of biorthogonal polynomials given in Proposition \ref{prop:char} completes the proof.
\end{proof}

\subsection{First transformation $Y \mapsto T$}

Recall $\gfn(z)$ and $\gfntilde(z)$ defined in \eqref{eq:defn_g_gtilde} on $\compC \setminus [0, b]$ and $\paraP \setminus [0, b]$, respectively. Denote $Y = Y^{(n + k, n)}$ and define $T$ as
\begin{equation} \label{eq:trans_Y_to_T}
	T(z) = e^{-\frac{n \ell}{2}} Y(z)
	\begin{pmatrix}
		e^{-n\gfn(z)} & 0 \\
		0 & e^{n\gfntilde(z)}
	\end{pmatrix}
	e^{\frac{n\ell}{2} \sigma_3},
\end{equation}
where $\ell$ is the constant appearing in \eqref{eq:defn_ell}, and $\sigma_3 = (\begin{smallmatrix} 1 & 0 \\ 0 & -1 \end{smallmatrix})$. Then $T$ satisfies an RH problem with the same domain of analyticity as $Y$, but with different asymptotic behavior and a different jump relation.

\begin{RHP} \hfill \label{RHP:T}
	\begin{enumerate}
		\item
		$T = (T_1, T_2)$, where $T_1$ is analytic in $\compC \setminus \realR_+$, and $T_2$ is analytic in $\paraP \setminus \realR_+$.
		\item
		$T$ satisfies the jump relation
		\begin{equation}
			T_+(x) = T_-(x) J_T(x), \quad \text{for $x \in \realR_+$},
		\end{equation}
		where
		\begin{equation}
			J_T(x) =
			\begin{pmatrix}
				e^{n(\gfn_-(x) - \gfn_+(x))} & \frac{x^{\alpha} h(x)}{f'(x)} e^{n(\gfn_-(x) + \gfntilde_+(x) - V(x) - \ell)} \\
				0 & e^{n(\tilde{\gfn}_+(x) - \tilde{\gfn}_-(x))} 
			\end{pmatrix}.
		\end{equation}
		\item 
		\begin{align} \label{eq:T_at_infty}
			T_1(z) = {}& z^k + \bigO(z^{k - 1}) \quad \text{as $z \to \infty$ in $\compC$}, & T_2(z) = {}& \bigO(f^{-(k + 1)}(z)) \quad \text{as $f(z) \to \infty$ in $\paraP$}.
		\end{align}
		\item
		As $z \to 0$ in $\compC$ or $\paraP$, $T(z)$ has the same limit behavior as $Y(z)$ in \eqref{eq:Y_bd_0}.
		\item 
		\begin{align} \label{eq:T_at_b}
			T_1(z) = {}& \bigO(1), & T_2(z) = {}& \bigO(1), & \text{as $z \to b$}.
		\end{align}
		\item
		At $z \in \rho \cup \{ -\pi^2/4 \} \cup \bar{\rho}$, $T_2(z)$ satisfies the same boundary condition as $Y_2(z)$ in \eqref{eq:Y_bd_cond}.
	\end{enumerate}
\end{RHP}

\subsection{Second transformation $T \mapsto S$} \label{sec:T_to_S}

Let
\begin{equation} \label{eq:defn_phi}
	\phi(z) = \gfn(z) + \tilde{\gfn}(z) - V(z) - \ell
\end{equation}
for $z \in \paraP \setminus (-\infty, b)$, where $\ell$ is a constant chosen such that $\phi(0) = \phi(b) = 0$. (See \eqref{eq:defn_ell} and \eqref{eq:defn_ell:2}.) For $x \in (b, +\infty)$, it follows from the analyticity of $\gfn(z)$ and $\tilde{\gfn}(z)$ there and \eqref{eq:gfn_outer} that the jump matrix $J_T(x)$ tends to the identity matrix exponentially fast as $n \to \infty$. For $x \in (0, b)$, we decompose the jump matrix $J_T(x)$ as
\begin{equation} \label{eq:J_T_decomposition}
	\begin{pmatrix}
		1 & 0 \\
		\frac{f'(x)}{x^{\alpha} h(x)} e^{-n \phi_-(x)} & 1
	\end{pmatrix}
	\begin{pmatrix}
		0 & \frac{e^{n(\gfn_-(x) + \gfntilde_+(x) - V(x) - \ell)}}{f'(x) x^{-\alpha} h(x)^{-1}} \\
		-\frac{f'(x) x^{-\alpha} h(x)^{-1}}{e^{n(\gfn_-(x) + \gfntilde_+(x) - V(x) - \ell)}} & 0
	\end{pmatrix}
	\begin{pmatrix}
		1 & 0 \\
		\frac{f'(x)}{x^{\alpha} h(x)} e^{-n \phi_+(x)} & 1
	\end{pmatrix}.
\end{equation}
The function $\phi(z)$ has discontinuities on $\realR_-$ and $(0, b)$, such that
\begin{align}
	\phi_+(x) = {}& \phi_-(x) + 4\pi i, & x < {}& 0, \\
	\phi_+(x) = {}& -\phi_-(x), & x \in {}& (0, b).
\end{align}

Then we ``open the lens'', where the lens $\Sigma$ is a contour consisting of $\realR_+$ and two arcs from $0$ to $b$. We assume one of the two arcs lies in the upper half-plane and denote it by $\Sigma_1$, and the other lies in the lower half-plane and denote it by $\Sigma_2$ (see Figure \ref{fig:Sigma_S}). (We may take $\Sigma_1$ and $\Sigma_2$ symmetric about the real axis.) We do not fix the shape of $\Sigma$ at this stage, but only require that $\Sigma$ is in $\paraP$ and $V$ is analytic in a simply connected region containing $\Sigma$. The exact shape of $\Sigma_1$ and $\Sigma_2$ will be given in Sections \ref{subsec:Airy_para}, \ref{subsec:Bes_para}, and \ref{subsec:final_trans}.

\begin{figure}[htb]
	\centering
	\includegraphics{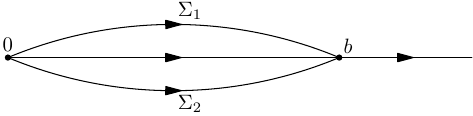}
	\caption{The lens $\Sigma$.}
	\label{fig:Sigma_S}
\end{figure}

Define
\begin{equation} \label{eq:trans_S_in_T}
	S(z) :=
	\begin{cases}
		T(z), & \text{outside the lens}, \\
		T(z)
		\begin{pmatrix}
			1 & 0 \\
			z^{-\alpha} h(z)^{-1} f'(z) e^{-n \phi(z)} & 1
		\end{pmatrix},
		& \text{in the lower part of the lens}, \\
		T(z)
		\begin{pmatrix}
			1 & 0 \\
			-z^{-\alpha} h(z)^{-1} f'(z) e^{-n \phi(z)} & 1
		\end{pmatrix},
		& \text{in the upper part of the lens}.
	\end{cases}
\end{equation}
From the definition of $S$, identity \eqref{eq:Euler-Lagrange}, and decomposition of $J_T(x)$ in \eqref{eq:J_T_decomposition}, we have that $S$ satisfies the following:

\begin{RHP} \hfill \label{RHP:S}
	\begin{enumerate}
		\item
		$S = (S_1, S_2)$, where $S_1$ is analytic in $\compC \setminus \Sigma$, and $S_2$ is analytic in $\paraP \setminus \Sigma$.
		\item
		We have
		\begin{equation}
			S_+(z) = S_-(z) J_S(z), \quad \text{for $z \in \Sigma$},
		\end{equation}
		where
		\begin{equation}
			J_S(z) =
			\begin{cases}
				\begin{pmatrix}
					1 & 0 \\
					z^{-\alpha} h(z)^{-1} f'(z) e^{-n \phi(z)} & 1
				\end{pmatrix},
				& \text{for $z \in \Sigma_1 \cup \Sigma_2$}, \\
				\begin{pmatrix}
					0 & z^{\alpha} h(z) f'(z)^{-1} \\
					-z^{-\alpha} h(z)^{-1} f'(z) & 0
				\end{pmatrix},
				& \text{for $z \in (0, b)$}, \\
				\begin{pmatrix}
					1 & z^{\alpha} h(z) f'(z)^{-1} e^{n \phi(z)} \\
					0 & 1
				\end{pmatrix},
				& \text{for $z \in (b, \infty)$}.
			\end{cases}
		\end{equation}
		\item
		As $z \to \infty$ in $\compC$ or $\paraP$, $S(z)$ has the same limit behavior as $T(z)$ in \eqref{eq:T_at_infty}.
		\item
		As $z \to 0$ in $\compC \setminus \Sigma$, we have
		\begin{equation}
			S_1(z) =
			\begin{cases}
				\bigO(z^{-\alpha}), & \text{$\alpha > 0$ and $z$ inside the lens}, \\
				\bigO(\log z), & \text{$\alpha = 0$ and $z$ inside the lens}, \\
				\bigO(1), & \text{$z$ outside the lens or $-1 < \alpha < 0$}.
			\end{cases}
		\end{equation}
		\item 
		As $z \to 0$ in $\paraP$, $S_2$ has the same limit behavior as $Y_2(z)$ in \eqref{eq:Y_bd_0}.
		\item
		As $z \to b$, $S(z)$ has the same limit behavior as $T(z)$ in \eqref{eq:T_at_b}.
		\item
		At $z \in \rho \cup \{ -\pi^2/4 \} \cup \bar{\rho}$, $S_2(z)$ satisfies the same boundary condition as $Y_2(z)$ in \eqref{eq:Y_bd_cond}.
	\end{enumerate}
\end{RHP}

By \eqref{eq:Euler-Lagrange}, for $x \in (0, b)$, we have
\begin{equation}
	\phi'_{\pm}(x) = \gfn'_{\pm}(x) + \gfntilde'_{\pm}(x) - V'(x) = \gfn'_{\pm}(x) - \gfntilde'_{\mp}(x) = \mp 2\pi i \psi(x).
\end{equation}
Since $\psi(x) > 0$ for all $x \in (0, b)$, we have, by the Cauchy-Riemann condition, $\Re \phi(z) > 0$ on both the upper arc $\Sigma_1$ and lower arc $\Sigma_2$, provided these arcs are sufficiently close to $(0, b)$. As a consequence, the jump matrix for $S$ on the lenses tends to the identity matrix as $n \to \infty$. Uniform convergence breaks down when $x$ approaches the endpoints $0$ and $b$, so we need to use special local parametrices near these points.

\begin{rmk} \label{rmk:contour_deform}
	Here and in subsequent RH problems, we may deform the jump contour $(b, +\infty)$ locally, as long as $\Re \phi(z) < 0$ there, so that the entry $z^{\alpha} h(z) f'(z)^{-1} e^{n \phi(z)}$ remains exponentially small.
\end{rmk}

\subsection{Construction of the global parametrix}

Since $J_S(z)$ converges to $I$ on $\Sigma_1 \cup \Sigma_2 \cup (b, \infty)$, we construct the following:

\begin{RHP} \hfill \label{RHP:Pinfinity}
	\begin{enumerate}
		\item
		$P^{(\infty)} = (P^{(\infty)}_1, P^{(\infty)}_2)$, where $P^{(\infty)}_1$ is analytic in $\compC \setminus [0, b]$, and $P^{(\infty)}_2$ is analytic in $\paraP \setminus [0, b]$.
		\item
		For $x \in (0, b)$, we have
		\begin{equation}
			P^{(\infty)}_+(x) = P^{(\infty)}_-(x)
			\begin{pmatrix}
				0 & x^{\alpha} h(x) f'(x)^{-1} \\
				-x^{-\alpha} h(x)^{-1} f'(x) & 0
			\end{pmatrix}.
		\end{equation}
		\item
		As $z \to \infty$ in $\compC$ or $\paraP$, $P^{(\infty)}(z)$ has the same limit behavior as $T(z)$ in \eqref{eq:T_at_infty}.
		\item \label{enu:RHP:Pinfinity_5}
		At $z \in \rho \cup \{ -\pi^2/4 \} \cup \bar{\rho}$, $P^{(\infty)}_2(z)$ satisfies the same boundary condition as $Y_2(z)$ in \eqref{eq:Y_bd_cond}.
	\end{enumerate}
\end{RHP}

To construct a solution to the above RH problem, we follow the idea in \cite{Claeys-Wang11} to map the RH problem for $P^{(\infty)}$ to a scalar RH problem that can be solved explicitly. More precisely, using the function $\J(s)$ defined in \eqref{eq:defn_J_c}, we set
\begin{equation} \label{eq:P_scalar_defn}
	\P(s) :=
	\begin{cases}
		P^{(\infty)}_1(\J(s)), & s \in \compC \setminus \overline{D}, \\
		P^{(\infty)}_2(\J(s)), & s \in D \setminus [0, 1],
	\end{cases}
\end{equation}
where $D$ is the region bounded by the curves $\gamma_1$ and $\gamma_2$, as shown in Figure \ref{fig:J_c}. Due to Item \ref{enu:RHP:Pinfinity_5} of RH Problem \ref{RHP:Pinfinity}, the function $\P$ is then well defined on $[0, 1)$ by analytic continuation. It is straightforward to check that $\P$ satisfies the following:

\begin{RHP} \hfill \label{RHP:P_scalar}
	\begin{enumerate}
		\item
		$\P$ is analytic in $\compC \setminus ( \gamma_1 \cup \gamma_2 \cup \{ 1 \} )$.
		\item
		For $s \in \gamma_1 \cup \gamma_2$, $\P_+(s) = \P_-(s) J_{\P}(s)$, where, with $\gamma_1$ and $\gamma_2$ oriented from $s_1$ to $s_2$,
		\begin{equation}
			J_{\P}(s) =
			\begin{cases}
				-\frac{\sinh(J_c(s))}{\J^{\alpha}(s) h(\J(s)) J_c(s)}, & s \in \gamma_1, \\
				\frac{\J^{\alpha}(s) h(\J(s)) J_c(s)}{\sinh(J_c(s))}, & s \in \gamma_2.
			\end{cases}
		\end{equation}
		
		\item
		As $s \to \infty$, $\P(s) = (\frac{c^2}{4})^k s^k + \bigO(s^{k - 1})$.
		\item
		As $s \to 1$, $\P(s) = \bigO((s - 1)^{k + 1})$.
		
	\end{enumerate}
\end{RHP}

The solution to RH problem \ref{RHP:P_scalar} may not be unique. One solution is
\begin{equation} \label{eq:defn_P_scalar}
	\P(s) = 
	\begin{cases}
		G_k(s), & s \in \compC \setminus \bar{D}, \\
		2 \frac{(\frac{c^2}{4})^{\alpha + \frac{1}{2} + k} (s - s_1)^{\alpha + 1} s^{\frac{1}{2}} (s - 1)^k}{\sinh(J_c(s)) \sqrt{(s - s_1)(s - s_2)}} \tilde{D}(s)^{-1}, & s \in D.
	\end{cases}
\end{equation}
where $\tilde{D}(s)$ is defined in \eqref{eq:defn_Dtilde(s)}, the power function $(s - s_1)^{\alpha + 1}$ takes the principal branch, and $\sqrt{(s - s_1)(s - s_2)} \sim s$ in $\compC \setminus \gamma_1$. Later in this paper, we take \eqref{eq:defn_P_scalar} as the definition of $\P(s)$. We note that
\begin{align}
	\P(s) = {}&
	\begin{cases}
		\bigO((s - s_1)^{-\alpha - \frac{1}{2}}), & s \to s_1 \text{ in } \compC \setminus \bar{D}, \\
		\bigO((s - s_1)^{\alpha - \frac{1}{2}}), & s \to s_1 \text{ in } D,
	\end{cases}
	& \P(s) = {}& \bigO((s - s_2)^{-\frac{1}{2}}), \quad s \to s_2.
\end{align}

Based on this solution, we construct the solution to RH problem \ref{RHP:Pinfinity} as
\begin{align}
	P^{(\infty)}_1(z) = {}& \P(\Iinv_1(z)) = G_k(\Iinv_1(z)), & z \in {}& \compC \setminus [0, b], \label{eq:Pinfty_1_in_scalar} \\
	P^{(\infty)}_2(z) = {}& \P(\Iinv_2(z)), & z \in {}& \paraP \setminus [0, b]. \label{eq:Pinfty_2_in_scalar}
\end{align}
By direct calculation in Section \ref{sec:J_x_prop}, we have
\begin{align}
	P^{(\infty)}_1(z) = {}& \bigO(z^{-\frac{\alpha}{2} - \frac{1}{4}}), & P^{(\infty)}_2(z) = {}& \bigO(z^{\frac{\alpha}{2} - \frac{1}{4}}), & \text{as } z \to 0, \label{eq:Pinfty_at_0} \\
	P^{(\infty)}_1(z) = {}& \bigO(z^{-\frac{1}{4}}), & P^{(\infty)}_2(z) = {}& \bigO(z^{-\frac{1}{4}}), & \text{as } z \to b. \label{eq:Pinfty_at_b}
\end{align}

\subsection{Third transformation $S \mapsto Q$}

Noting that $P^{(\infty)}_1(z) \neq 0$ for all $z \in \compC \setminus [0, b]$ and $P^{(\infty)}_2(z) \neq 0$ for all $z \in \paraP \setminus [0, b]$, we define the third transformation by
\begin{equation} \label{eq:defn_Q}
	Q(z) = (Q_1(z), Q_2(z)) = \left( \frac{S_1(z)}{P^{(\infty)}_1(z)}, \frac{S_2(z)}{P^{(\infty)}_2(z)} \right).
\end{equation}
In view of the RH problems \ref{RHP:S} and \ref{RHP:Pinfinity}, and the properties of $P^{(\infty)}(z)$ in \eqref{eq:Pinfty_at_0} and \eqref{eq:Pinfty_at_b}, it is easy to see that $Q$ satisfies the following RH problem (the contour $\Sigma$ is given in Figure \ref{fig:Sigma_S}):

\begin{RHP} \hfill \label{RHP:Q}
	\begin{enumerate}
		\item
		$Q = (Q_1, Q_2)$, where $Q_1$ is analytic in $\compC \setminus \Sigma$, and $Q_2$ is analytic in $\paraP \setminus \Sigma$.
		\item
		For $z \in \Sigma$, we have
		\begin{equation} \label{eq:jump_Q}
			Q_+(z) = Q_-(z) J_Q(z), 
		\end{equation}
		where
		\begin{equation} \label{eq:defn_J_Q}
			J_Q(z) =
			\begin{cases}
				\begin{pmatrix}
					1 & 0 \\
					z^{-\alpha} h(z)^{-1} f'(z) \frac{P^{(\infty)}_2(z)}{P^{(\infty)}_1(z)} e^{-n \phi(z)} & 1
				\end{pmatrix},
				& z \in \Sigma_1 \cup \Sigma_2, \\
				\begin{pmatrix}
					0 & 1 \\
					1 & 0
				\end{pmatrix},
				& z \in (0, b), \\
				\begin{pmatrix}
					1 & z^{\alpha} h(z) f'(z)^{-1} \frac{P^{(\infty)}_1(z)}{P^{(\infty)}_2(z)} e^{-n \phi(z)} \\
					0 & 1 \\
				\end{pmatrix},
				& z \in (b, \infty).
			\end{cases}
		\end{equation}
		\item
		\begin{align}
			Q_1(z) = {}& 1 + \bigO(z^{-1}), \quad \text{as $z \to \infty$ in $\compC$}, & Q_2(z) = {}& \bigO(f^{-1}(z)), \quad \text{as $f(z) \to \infty$ in $\paraP$}.
		\end{align}
		\item \label{enu:RHP:Q_zero}
		As $z \to 0$ in $\compC \setminus \Sigma$, we have
		\begin{equation} \label{eq:RHP:Q_zero}
			Q_1(z) =
			\begin{cases}
				\bigO(z^{-\frac{\alpha}{2} + \frac{1}{4}}), & \text{$\alpha > 0$ and $z$ inside the lens}, \\
				\bigO(z^{\frac{1}{4}} \log z), & \text{$\alpha = 0$ and $z$ inside the lens}, \\
				\bigO(z^{\frac{\alpha}{2} + \frac{1}{4}}), & \text{$z$ outside the lens or $-1 < \alpha < 0$}.
			\end{cases}
		\end{equation}
		\item \label{enu:RHP:Q_zero_2} 
		As $z \to 0$ in $\paraP$, we have
		\begin{equation} \label{eq:RHP:Q_zero_2}
			Q_2(z) =
			\begin{cases}
				\bigO(z^{-\frac{\alpha}{2} + \frac{1}{4}}), & \alpha > 0, \\
				\bigO(z^{\frac{1}{4}} \log z), & \alpha = 0, \\
				\bigO(z^{\frac{\alpha}{2} + \frac{1}{4}}), & \alpha \in (-1, 0).
			\end{cases}
		\end{equation}
		\item \label{enu:RHP:Q_b} 
		\begin{align} \label{eq:RHP:Q_b} 
			Q_1(z) = {}& \bigO((z - b)^{\frac{1}{4}}), & Q_2(z) = {}& \bigO((z - b)^{\frac{1}{4}}), & \text{as $z \to b$}.
		\end{align}
		\item
		At $z \in \rho \cup \{ -\pi^2/4 \} \cup \bar{\rho}$, $Q_2(z)$ satisfies the same boundary condition as $Y_2(z)$ in \eqref{eq:Y_bd_cond}.
	\end{enumerate}
\end{RHP}

\subsection{Construction of local parametrix near $b$} \label{subsec:Airy_para}

First, we consider the local parametrix near $b$. Using Part \ref{enu:regular_g:5} of Lemma \ref{lem:regular_g}, we have $\lim_{z \to b} \phi(z) = 0$, where $\phi$ is defined in \eqref{eq:defn_phi}. By Part \ref{enu:regular_g:3} of Lemma \ref{lem:regular_g} and Part \ref{enu:req:one-cut_reg:3} of Requirement \ref{req:one-cut_reg}, we obtain the local behavior of $\phi$ in the vicinity of $b$ (where $\psi_b$ is defined in \eqref{eq:reg:one-cut_reg:3}):
\begin{equation}
	\phi(z) = -\frac{4\pi}{3} \psi_b (z - b)^{\frac{3}{2}} + \bigO(\lvert z - b \rvert^{\frac{5}{2}}),
\end{equation}
$\phi(z)/(z - b)^{3/2}$ is analytic at $b$, and thus we have
\begin{equation} \label{eq:f_b_prop}
	f_b(z) := \left( -\frac{3}{4} \phi(z) \right)^{\frac{2}{3}}, \quad \text{with $f_b(b) = 0$ and $f'_b(b) = (\pi \psi_b)^{\frac{2}{3}} > 0$}
\end{equation}
is a well-defined analytic function in a small neighborhood of $b$.
Moreover, we also choose the shape of the contour $\Sigma$ such that the image of $\Sigma \cap D(b, \epsilon)$ under the mapping $f_b$ coincides with the jump contour $\Gamma_{\Ai}$ defined in \eqref{def:AiryContour}, which is the jump contour of the RH problem \ref{rhp:Ai} for the Airy parametrix.

Let
\begin{align} \label{eq:defn_g^(b)_i}
	g^{(b)}_1(z) = {}& \frac{f'(z)/h(z)}{P^{(\infty)}_1(z)}, & g^{(b)}_2(z) = {}& \frac{z^{\alpha}}{P^{(\infty)}_2(z)},
\end{align}
and define
\begin{equation} \label{eq:defn_Pmodel}
	\Pmodel^{(b)}(z) := \Psi^{(\Ai)}(n^{\frac{2}{3}} f_b(z))
	\begin{pmatrix}
		e^{-\frac{n}{2} \phi(z)} g^{(b)}_1(z) & 0 \\
		0 & e^{\frac{n}{2} \phi(z)} g^{(b)}_2(z)
	\end{pmatrix},
	\quad z \in D(b, \epsilon) \setminus \Sigma.
\end{equation}
From \eqref{eq:defn_g^(b)_i} and the RH problem \ref{RHP:Pinfinity} satisfied by $P^{(\infty)}(z)$, we have
\begin{align}
	g^{(b)}_{1, +}(x) = {}& -g^{(b)}_{2, -}(x), & g^{(b)}_{2, +}(x) = {}& g^{(b)}_{1, -}(x), && \text{for $x \in (b - \epsilon, b)$}, \label{eq:gb_jump} \\
	g^{(b)}_1(z) = {}& \bigO((z - b)^{\frac{1}{4}}), & g^{(b)}_2(z) = {}& \bigO((z - b)^{\frac{1}{4}}), && \text{as $z \to b$}. \label{eq:gb_bd}
\end{align}
Then $\Pmodel^{(b)}(z)$ satisfies the following RH problem:

\begin{RHP} \hfill \label{RHP:Pmodel^b}
	\begin{enumerate}
		\item
		$\Pmodel^{(b)}(z)$ is a $2 \times 2$ matrix-valued function analytic for $z \in D(b, \epsilon) \setminus \Sigma$.
		\item
		For $z \in \Sigma \cap D(b, \epsilon)$, we have
		\begin{equation} \label{eq:jump_Pmodel}
			\Pmodel^{(b)}_+(z) = \Pmodel^{(b)}_-(z) J_Q(z),
		\end{equation}
		where $J_Q(z)$ is defined in \eqref{eq:defn_J_Q}.
		\item
		\begin{align} \label{eq:Pmodel^b_asy}
			(\Pmodel^{(b)})_{ij}(z) = {}& \bigO((z - b)^{\frac{1}{4}}), & ((\Pmodel^{(b)})^{-1})_{ij}(z) = {}& \bigO((z - b)^{-\frac{1}{4}}), & \text{as $z \to b$}, \quad i, j = 1, 2.
		\end{align}
		\item
		For $z \in \partial D(b, \epsilon)$, we have, as $n \to \infty$,
		\begin{equation}
			E^{(b)}(z) \Pmodel^{(b)}(z) = I + \bigO(n^{-1}),
		\end{equation}
		where
		\begin{equation} \label{eq:defn_E^b}
			E^{(b)}(z) = 
			\frac{1}{\sqrt{2}}
			\begin{pmatrix}
				g^{(b)}_1(z) & 0 \\
				0 & g^{(b)}_2(z)
			\end{pmatrix}^{-1}
			e^{\frac{\pi i}{4} \sigma_3}
			\begin{pmatrix}
				1 & -1 \\
				1 & 1
			\end{pmatrix}
			\begin{pmatrix}
				n^{\frac{1}{6}} f_b(z)^{\frac{1}{4}} & 0 \\
				0 & n^{-\frac{1}{6}} f_b(z)^{-\frac{1}{4}}
			\end{pmatrix}.
		\end{equation}
	\end{enumerate}
\end{RHP}

It is straightforward to see that $E^{(b)}(z)$ defined in \eqref{eq:defn_E^b} is analytic on $D(b, \epsilon) \setminus (b - \epsilon, b]$, and for $x \in (b - \epsilon, b)$
\begin{equation}
	E^{(b)}_+(x) E^{(b)}_-(x)^{-1} =
	\begin{pmatrix}
		0 & 1 \\
		1 & 0
	\end{pmatrix},
\end{equation}
and as $z \to b$,
\begin{align}
	E^{(b)}(z) = {}&
	\begin{pmatrix}
		\bigO(1) & \bigO(z^{-\frac{1}{2}}) \\
		\bigO(1) & \bigO(z^{-\frac{1}{2}})
	\end{pmatrix}, &
	E^{(b)}(z)^{-1} = {}&
	\begin{pmatrix}
		\bigO(1) & \bigO(1) \\
		\bigO(z^{\frac{1}{2}}) & \bigO(z^{\frac{1}{2}})
	\end{pmatrix}.
\end{align}
Then we define a $2 \times 2$ matrix-valued function
\begin{equation} \label{eq:P^b_in_E}
	P^{(b)}(z) = E^{(b)}(z) \Pmodel^{(b)}(z), \quad z \in D(b, \epsilon) \setminus \Sigma,
\end{equation}
where $\Pmodel^{(b)}$ is given in \eqref{eq:defn_Pmodel}. Then $P^{(b)}(z)$ satisfies the following RH problem:

\begin{RHP} \hfill \label{RHP:P^b_actual}
	\begin{enumerate}
		\item
		$P^{(b)}(z)$ is analytic in $D(b, \epsilon) \setminus \Sigma$.
		\item
		For $z \in \Sigma \cap D(b, \epsilon)$, we have
		\begin{equation}
			P^{(b)}_+(z) =
			\begin{cases}
				P^{(b)}_-(z) J_Q(z), & z \in \Sigma \cap D(b, \epsilon) \setminus (b - \epsilon, b], \\
				\begin{pmatrix}
					0 & 1 \\
					1 & 0
				\end{pmatrix}
				P^{(b)}_-(z) J_Q(z), & z \in (b - \epsilon, b).
			\end{cases}
		\end{equation}
		\item
		\begin{align} \label{eq:P^b_asy}
			(P^{(b)})_{ij}(z) = {}& \bigO((z - b)^{-\frac{1}{4}}), & ((P^{(b)})^{-1})_{ij}(z) = {}& \bigO((z - b)^{-\frac{1}{4}}), & \text{as $z \to b$}, \quad i, j = 1, 2.
		\end{align}
		\item 
		For $z$ on the boundary $\partial D(b, \epsilon)$, we have, as $n \to \infty$, $P^{(b)}(z) = I + \bigO(n^{-1})$.
	\end{enumerate}
\end{RHP}

Finally, we define a vector-valued function $V^{(b)}$ by
\begin{equation} \label{eq:V^b_in_Q_P}
	V^{(b)}(z) = Q(z) P^{(b)}(z)^{-1}, \quad z \in D(b, \epsilon) \setminus \Sigma,
\end{equation}
where $Q(z)$ is defined in \eqref{eq:defn_Q}. We find that $V^{(b)}(z)$ has only a trivial jump on $(\Sigma_1 \cup \Sigma_2 \cup [b, b + \epsilon)) \cap D(b, \epsilon)$, so $V^{(b)}(z)$ can be defined by continuation on $D(b, \epsilon) \setminus (b - \epsilon, b]$. It satisfies the following RH problem:

\begin{RHP} \hfill \label{RHP:Vb}
	\begin{enumerate}
		\item
		$V^{(b)} = (V^{(b)}_1, V^{(b)}_2)$ is analytic in $D(b, \epsilon) \setminus (b - \epsilon, b]$.
		\item
		For $x \in (b - \epsilon, b)$, we have
		\begin{equation}
			V^{(b)}_+(x) = V^{(b)}_-(x)
			\begin{pmatrix}
				0 & 1 \\
				1 & 0
			\end{pmatrix}.
		\end{equation}
		\item
		\begin{align}
			V^{(b)}_1(z) = {}& \bigO(1), & V^{(b)}_2(z) = {}& \bigO(1), & \text{as $z \to b$}.
		\end{align}
		\item
		For $z \in \partial D(b, \epsilon)$, we have, as $n \to \infty$, $V^{(b)}(z) = Q(z) (I + \bigO(n^{-1}))$.
	\end{enumerate}
\end{RHP}

\subsection{Construction of local parametrix near $0$} \label{subsec:Bes_para}

Using Part \ref{enu:lem:regular_g:1} of Lemma \ref{lem:regular_g}, we have $\lim_{z \to 0 \text{ in } \compC_{\pm}} \phi(z) = \pm \pi i$. We define in the vicinity of $0$:
\begin{equation}
	\phi^L(z) =
	\begin{cases}
		\phi(z) - \pi i, & z \in \compC_+, \\
		\phi(z) + \pi i, & z \in \compC_-,
	\end{cases}
\end{equation}
By Part \ref{enu:regular_g:3} of Lemma \ref{lem:regular_g} and Part \ref{enu:req:one-cut_reg:3} of Requirement \ref{req:one-cut_reg}, we obtain the local behavior of $\phi$ (or equivalently $\phi^L$) at $0$ (where $\psi_0$ is the positive constant defined in Part \ref{enu:req:one-cut_reg:3} of Requirement \ref{req:one-cut_reg}):
\begin{equation}
	\phi^L(z) = 4\pi \psi_0 (-z)^{1/2} + \bigO((-z)^{3/2}),
\end{equation}
and $\phi^L(z)/(-z)^{1/2}$ is analytic at $0$. In addition,  similarly to \eqref{eq:f_b_prop}, 
\begin{equation} \label{eq:f_0_prop}
	f_0(z) := \frac{1}{16} \phi^L(z)^2= \frac{1}{16} (\phi(z) \mp \pi i)^2, \quad \text{with $f_0(0) = 0$ and $f'_0(0) = -(\pi\psi_0)^2 < 0$}
\end{equation}
(where the sign is $-$ in $\compC_+$ and $+$ in $\compC_-$), is a well-defined analytic function in a certain neighborhood of $0$. In fact, 	$f_0(z)$ 
is a conformal mapping in a neighborhood $D(0, \epsilon)$ around $0$ for small enough $\epsilon>0$. Moreover, we also choose the shape of the contour $\Sigma$ such that the image of $\Sigma \cap D(0, \epsilon)$ under the mapping $f_0$ coincides with the jump contour $\Gamma_{\Be}$ defined in \eqref{def:BesselContour}, which is the jump contour of the RH problem \ref{RHP:Bessel} for the Bessel parametrix.

Let
\begin{align} \label{eq:defn_g^(0)_i}
	g^{(0)}_1(z) = {}& \frac{(-z)^{-\alpha/2} f'(z)/h(z)}{P^{(\infty)}_1(z)}, & g^{(0)}_2(z) = {}& \frac{(-z)^{\alpha/2}}{P^{(\infty)}_2(z)},
\end{align}
where $(-z)^{\pm \alpha/2}$ takes the principal branch, and define (where $\Psi^{(\Be)}_{\alpha}$ is defined in \eqref{eq:psibeal}),
\begin{equation} \label{eq:defn_Pmodel_0}
	\Pmodel^{(0)}(z) = \Psi^{(\Be)}_{\alpha}(n^2 f_0(z))
	\begin{pmatrix}
		e^{-\frac{n}{2} \phi(z)} g^{(0)}_1(z) & 0 \\
		0 & e^{\frac{n}{2} \phi(z)} g^{(0)}_2(z)
	\end{pmatrix},
	\quad z \in D(0, \epsilon) \setminus \Sigma.
\end{equation}
From \eqref{eq:defn_g^(0)_i} and the RH problem \ref{RHP:Pinfinity} satisfied by $P^{(\infty)}(z)$, we have
\begin{align}
	g^{(0)}_{1, +}(x) = {}& -g^{(0)}_{2, -}(x), & g^{(0)}_{2, +}(x) = {}& g^{(0)}_{1, -}(x), && \text{for $x \in (0, \epsilon)$}, \label{eq:g0_jump} \\
	g^{(0)}_1(z) = {}& \bigO(z^{\frac{1}{4}}), & g^{(0)}_2(z) = {}& \bigO(z^{\frac{1}{4}}), && \text{as $z \to 0$}. \label{eq:g0_bd}
\end{align}
Then $\Pmodel^{(0)}(z)$ satisfies the following RH problem:
\begin{RHP} \hfill \label{RHP:Pmodel}
	\begin{enumerate}
		\item
		$\Pmodel^{(0)}(z)$ is a $2 \times 2$ matrix-valued function analytic for $z \in D(0, \epsilon) \setminus \Sigma$.
		\item
		For $z \in \Sigma \cap D(0, \epsilon)$, we have
		\begin{equation}
			\Pmodel^{(0)}_+(z) = \Pmodel^{(0)}_-(z) J_Q(z),
		\end{equation}
		where $J_Q(z)$ is defined in \eqref{eq:defn_J_Q}.
		\item
		As $z \in \partial D(0, \epsilon)$, we have
		\begin{equation}
			E^{(0)}(z) \Pmodel^{(0)}(z) = (I + \bigO(n^{-1})),
		\end{equation}
		where
		\begin{equation} \label{eq:defn_E^0}
			E^{(0)}(z) = \frac{1}{\sqrt{2}}
			\begin{pmatrix}
				g^{(0)}_1(z) & 0 \\
				0 & g^{(0)}_2(z)
			\end{pmatrix}^{-1}
			\begin{pmatrix}
				1 & i \\
				i & 1
			\end{pmatrix}
			\begin{pmatrix}
				n^{\frac{1}{2}} f_0(z)^{\frac{1}{4}} & 0 \\
				0 & n^{-\frac{1}{2}} f_0(z)^{-\frac{1}{4}}
			\end{pmatrix}
			(2\pi)^{\frac{1}{2} \sigma_3}.
		\end{equation}
		\item \label{enu:RHP:Pmodel_zero}
		As $z \to 0$, if $\alpha \in (-1, 0)$, then
		\begin{align} \label{eq:Pmodel^0_asy_1}
			\Pmodel^{(0)}(z) = {}&
			\begin{pmatrix}
				\bigO(z^{\frac{\alpha}{2} + \frac{1}{4}}) & \bigO(z^{\frac{\alpha}{2} + \frac{1}{4}}) \\
				\bigO(z^{\frac{\alpha}{2} + \frac{1}{4}}) & \bigO(z^{\frac{\alpha}{2} + \frac{1}{4}})
			\end{pmatrix}, &
			\Pmodel^{(0)}(z)^{-1} = {}&
			\begin{pmatrix}
				\bigO(z^{\frac{\alpha}{2} - \frac{1}{4}}) & \bigO(z^{\frac{\alpha}{2} - \frac{1}{4}}) \\
				\bigO(z^{\frac{\alpha}{2} - \frac{1}{4}}) & \bigO(z^{\frac{\alpha}{2} - \frac{1}{4}})
			\end{pmatrix},
		\end{align}
		if $\alpha = 0$, then
		\begin{align} \label{eq:Pmodel^0_asy_2}
			\Pmodel^{(0)}(z) = {}&
			\begin{pmatrix}
				\bigO(z^{\frac{1}{4}} \log z) & \bigO(z^{\frac{1}{4}} \log z) \\
				\bigO(z^{\frac{1}{4}} \log z) & \bigO(z^{\frac{1}{4}} \log z)
			\end{pmatrix}, &
			\Pmodel^{(0)}(z)^{-1} = {}&
			\begin{pmatrix}
				\bigO(z^{-\frac{1}{4}} \log z) & \bigO(z^{-\frac{1}{4}} \log z) \\
				\bigO(z^{-\frac{1}{4}} \log z) & \bigO(z^{-\frac{1}{4}} \log z)
			\end{pmatrix},
		\end{align}
		and if $\alpha > 0$, then outside the lens
		\begin{align} \label{eq:Pmodel^0_asy_3}
			\Pmodel^{(0)}(z) = {}&
			\begin{pmatrix}
				\bigO(z^{\frac{\alpha}{2} + \frac{1}{4}}) & \bigO(z^{-\frac{\alpha}{2} + \frac{1}{4}}) \\
				\bigO(z^{\frac{\alpha}{2} + \frac{1}{4}}) & \bigO(z^{-\frac{\alpha}{2} + \frac{1}{4}})
			\end{pmatrix}, &
			\Pmodel^{(0)}(z)^{-1} = {}&
			\begin{pmatrix}
				\bigO(z^{-\frac{\alpha}{2} - \frac{1}{4}}) & \bigO(z^{-\frac{\alpha}{2} - \frac{1}{4}}) \\
				\bigO(z^{\frac{\alpha}{2} - \frac{1}{4}}) & \bigO(z^{\frac{\alpha}{2} - \frac{1}{4}})
			\end{pmatrix},
		\end{align}
		and inside the lens
		\begin{align} \label{eq:Pmodel^0_asy_4}
			\Pmodel^{(0)}(z) = {}&
			\begin{pmatrix}
				\bigO(z^{-\frac{\alpha}{2} + \frac{1}{4}}) & \bigO(z^{-\frac{\alpha}{2} + \frac{1}{4}}) \\
				\bigO(z^{-\frac{\alpha}{2} + \frac{1}{4}}) & \bigO(z^{-\frac{\alpha}{2} + \frac{1}{4}})
			\end{pmatrix}, &
			\Pmodel^{(0)}(z)^{-1} = {}&
			\begin{pmatrix}
				\bigO(z^{-\frac{\alpha}{2} - \frac{1}{4}}) & \bigO(z^{-\frac{\alpha}{2} - \frac{1}{4}}) \\
				\bigO(z^{-\frac{\alpha}{2} - \frac{1}{4}}) & \bigO(z^{-\frac{\alpha}{2} - \frac{1}{4}})
			\end{pmatrix}, &
		\end{align}
	\end{enumerate}
\end{RHP}

It is straightforward to see that $E^{(0)}(z)$ is analytic on $D(0, \epsilon) \setminus [0, \epsilon)$, and for $x \in (0, \epsilon)$,
\begin{equation} \label{eq:E0_jump}
	E^{(0)}_+(x) E^{(0)}_-(x)^{-1} =
	\begin{pmatrix}
		0 & 1 \\
		1 & 0
	\end{pmatrix},
\end{equation}
and as $z \to 0$,
\begin{align} \label{eq:E^0_defn}
	E^{(0)}(z) = {}&
	\begin{pmatrix}
		\bigO(1) & \bigO(z^{-\frac{1}{2}}) \\
		\bigO(1) & \bigO(z^{-\frac{1}{2}})
	\end{pmatrix}, &
	E^{(0)}(z)^{-1} = {}&
	\begin{pmatrix}
		\bigO(1) & \bigO(1) \\
		\bigO(z^{\frac{1}{2}}) & \bigO(z^{\frac{1}{2}})
	\end{pmatrix}.
\end{align}
Then we define a $2 \times 2$ matrix-valued function
\begin{equation} \label{eq:P^0_in_E}
	P^{(0)}(z) = E^{(0)}(z) \Pmodel^{(0)}(z), \quad z \in D(0, \epsilon) \setminus \Sigma,
\end{equation}
where $\Pmodel^{(0)}$ is given in \eqref{eq:defn_Pmodel_0}. Hence, $P^{(0)}(z)$ satisfies the following RH problem:

\begin{RHP} \hfill
	\begin{enumerate}
		\item
		$P^{(0)}(z)$ is analytic in $D(0, \epsilon) \setminus \Sigma$.
		\item
		For $z \in \Sigma \cap D(0, \epsilon)$, we have
		\begin{equation}
			P^{(0)}_+(z) =
			\begin{cases}
				P^{(0)}_-(z) J_Q(z), & z \in \Sigma \cap D(0, \epsilon) \setminus [0, \epsilon), \\
				\begin{pmatrix}
					0 & 1 \\
					1 & 0
				\end{pmatrix}
				P^{(0)}_-(z) J_Q(z), & z \in (0, \epsilon).
			\end{cases}
		\end{equation}
		\item 
		As $z \to 0$, if $\alpha \in (-1, 0)$, then
		\begin{align} \label{eq:P^0_asy_1}
			P^{(0)}(z) = {}&
			\begin{pmatrix}
				\bigO(z^{\frac{\alpha}{2} - \frac{1}{4}}) & \bigO(z^{\frac{\alpha}{2} - \frac{1}{4}}) \\
				\bigO(z^{\frac{\alpha}{2} - \frac{1}{4}}) & \bigO(z^{\frac{\alpha}{2} - \frac{1}{4}})
			\end{pmatrix}, &
			P^{(0)}(z)^{-1} = {}&
			\begin{pmatrix}
				\bigO(z^{\frac{\alpha}{2} - \frac{1}{4}}) & \bigO(z^{\frac{\alpha}{2} - \frac{1}{4}}) \\
				\bigO(z^{\frac{\alpha}{2} - \frac{1}{4}}) & \bigO(z^{\frac{\alpha}{2} - \frac{1}{4}})
			\end{pmatrix},
		\end{align}
		if $\alpha = 0$, then
		\begin{align} \label{eq:P^0_asy_2}
			P^{(0)}(z) = {}&
			\begin{pmatrix}
				\bigO(z^{-\frac{1}{4}} \log z) & \bigO(z^{-\frac{1}{4}} \log z) \\
				\bigO(z^{-\frac{1}{4}} \log z) & \bigO(z^{-\frac{1}{4}} \log z)
			\end{pmatrix}, &
			P^{(0)}(z)^{-1} = {}&
			\begin{pmatrix}
				\bigO(z^{-\frac{1}{4}} \log z) & \bigO(z^{-\frac{1}{4}} \log z) \\
				\bigO(z^{-\frac{1}{4}} \log z) & \bigO(z^{-\frac{1}{4}} \log z)
			\end{pmatrix},
		\end{align}
		and if $\alpha > 0$, then outside the lens
		\begin{align} \label{eq:P^0_asy_3}
			P^{(0)}(z) = {}&
			\begin{pmatrix}
				\bigO(z^{\frac{\alpha}{2} - \frac{1}{4}}) & \bigO(z^{-\frac{\alpha}{2} - \frac{1}{4}}) \\
				\bigO(z^{\frac{\alpha}{2} - \frac{1}{4}}) & \bigO(z^{-\frac{\alpha}{2} - \frac{1}{4}})
			\end{pmatrix}, &
			P^{(0)}(z)^{-1} = {}&
			\begin{pmatrix}
				\bigO(z^{-\frac{\alpha}{2} - \frac{1}{4}}) & \bigO(z^{-\frac{\alpha}{2} - \frac{1}{4}}) \\
				\bigO(z^{\frac{\alpha}{2} - \frac{1}{4}}) & \bigO(z^{\frac{\alpha}{2} - \frac{1}{4}})
			\end{pmatrix}.
		\end{align}
		and inside the lens
		\begin{align} \label{eq:P^0_asy_4}
			P^{(0)}(z) = {}&
			\begin{pmatrix}
				\bigO(z^{-\frac{\alpha}{2} - \frac{1}{4}}) & \bigO(z^{-\frac{\alpha}{2} - \frac{1}{4}}) \\
				\bigO(z^{-\frac{\alpha}{2} - \frac{1}{4}}) & \bigO(z^{-\frac{\alpha}{2} - \frac{1}{4}})
			\end{pmatrix}, &
			P^{(0)}(z)^{-1} = {}&
			\begin{pmatrix}
				\bigO(z^{-\frac{\alpha}{2} - \frac{1}{4}}) & \bigO(z^{-\frac{\alpha}{2} - \frac{1}{4}}) \\
				\bigO(z^{-\frac{\alpha}{2} - \frac{1}{4}}) & \bigO(z^{-\frac{\alpha}{2} - \frac{1}{4}})
			\end{pmatrix}, &
		\end{align}
		\item 
		For $z$ on the boundary $\partial D(0, \epsilon)$, we have, as $n \to \infty$, $P^{(0)}(z) = I + \bigO(n^{-1})$.
	\end{enumerate}
\end{RHP}

Consider the vector-valued function
\begin{equation} \label{eq:U_from_Q}
	U(z) = (U_1(z), U_2(z)) := Q(z) \Pmodel^{(0)}(z)^{-1}, \quad z \in D(0, \epsilon) \setminus \Sigma,
\end{equation}
where $Q(z)$ is defined in \eqref{eq:defn_Q}, and then define the vector-valued function $V^{(0)}$ on $D(0, \epsilon) \setminus \Sigma$ by
\begin{equation} \label{eq:V_in_U}
	V^{(0)}(z) = (V^{(0)}_1(z), V^{(0)}_2(z)) := Q(z) P^{(0)}(z)^{-1} = U(z) E^{(0)}(z)^{-1}.
\end{equation}
Due to the jump conditions \eqref{eq:jump_Q} and \eqref{eq:jump_Pmodel}, $U(z)$ can be extended analytically to $D(0, \epsilon) \setminus \{ 0 \}$. Furthermore, as $z \to 0$, from Part \ref{enu:RHP:Q_zero} of the RH problem \ref{RHP:Q} satisfied by $Q(z)$ and Part \ref{enu:RHP:Pmodel_zero} of the RH problem \ref{RHP:Pmodel} satisfied by $\Pmodel^{(0)}(z)$, we have
\begin{equation}
	(U_1(z), U_2(z)) =
	\begin{cases}
		(\bigO(1), \bigO(1)), & \alpha > 0 \text{ and $z$ is outside the lens}, \\
		(\bigO(z^{-\alpha}), \bigO(z^{-\alpha})), & \alpha > 0 \text{ and $z$ is inside the lens}, \\
		(\bigO((\log z)^2), \bigO((\log z)^2)), & \alpha = 0, \\
		(\bigO(z^{\alpha}), \bigO(z^{\alpha})), & \alpha \in (-1, 0).
	\end{cases}
\end{equation}
Since $0$ is an isolated singular point of $U_1(z)$ and $U_2(z)$, the estimates above imply that $0$ is a removable singular point of $U_1(z)$ and $U_2(z)$—in other words, these two functions can be extended analytically to $D(0, \epsilon)$. Hence, $V^{(0)}(z)$ satisfies the following RH problem:

\begin{RHP} \hfill \label{RHP:V^0}
	\begin{enumerate}
		\item
		$V^{(0)}(z) = (V^{(0)}_1(z), V^{(0)}_2(z))$ is analytic in $D(0, \epsilon) \setminus [0, \epsilon)$.
		\item
		For $x \in (0, \epsilon)$, we have
		\begin{equation}
			V^{(0)}_+(x) = V^{(0)}_-(x)
			\begin{pmatrix}
				0 & 1 \\
				1 & 0
			\end{pmatrix}.
		\end{equation}
		\item \label{enu:RHP:V^0_zero}
		As $z \to 0$, we have $V^{(0)}(z) = (\bigO(1), \bigO(1))$.
		\item
		As $z \in \partial D(0, \epsilon)$, $V^{(0)}(z) = Q(z)(I + \bigO(n^{-1}))$.
	\end{enumerate}
\end{RHP}

	\subsection{Final transformation} \label{subsec:final_trans}

We define $R(z) = (R_1(z), R_2(z))$ as
\begin{align}
	R_1(z) = {}&
	\begin{cases}
		V^{(b)}_1(z), & z \in D(b, \epsilon) \setminus (b - \epsilon, b], \\
		V^{(0)}_1(z), & z \in D(0, \epsilon) \setminus [0, \epsilon), \\
		Q_1(z), & z \in \compC \setminus (\overline{D(b, \epsilon)} \cup \overline{D(0, \epsilon)} \cup \Sigma),
	\end{cases} \label{eq:R_1_defn} \\
	R_2(z) = {}&
	\begin{cases}
		V^{(b)}_2(z), & z \in D(b, \epsilon) \setminus (b - \epsilon, b], \\
		V^{(0)}_2(z), & z \in D(0, \epsilon) \setminus [0, \epsilon), \\
		Q_2(z), & z \in \paraP \setminus (\overline{D(b, \epsilon)} \cup \overline{D(0, \epsilon)} \cup \Sigma).
	\end{cases} \label{eq:R_2_defn}
\end{align}
We set
\begin{equation}
	\Sigma^R := [0, b] \cup [b + \epsilon, \infty) \cup \partial D(0, \epsilon) \cup \partial D(b, \epsilon) \cup \Sigma^R_1 \cup \Sigma^R_2, 
\end{equation}
where
\begin{equation} \label{eq:shape_Sigma^R_i}
	\Sigma^R_i := \Sigma_i \setminus \{ D(0, \epsilon) \cup D(b, \epsilon) \}, \quad i = 1, 2.
\end{equation}
See Figure \ref{fig:Sigma_R} for an illustration and the orientation of the arcs. Here we can fix the shape of $\Sigma^R_i$ (and so finally fix the shape of $\Sigma_i$) by letting $\Sigma_i$ be a continuous arc and $\Re \phi(z) < 0$ on $\Sigma^R_i$. It is straightforward to check that $R$ satisfies the following RH problem:

\begin{figure}[htb]
	\centering
	\includegraphics{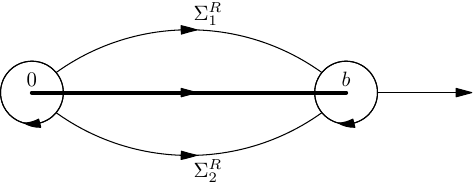}
	\caption{Contour $\Sigma^R$.}
	\label{fig:Sigma_R}
\end{figure}

\begin{RHP} \hfill \label{RHP:R}
	\begin{enumerate}
		\item
		$R(z) = (R_1(z), R_2(z))$, where $R_1(z)$ is analytic in $\compC \setminus \Sigma^R$, and $R_2(z)$ is analytic in $\paraP \setminus \Sigma^R$.
		\item
		$R(z)$ satisfies the following jump conditions:
		\begin{equation}
			R_+(z) = R_-(z)
			\begin{cases}
				J_Q(z), & z \in \Sigma^R_1 \cup \Sigma^R_2 \cup (b + \epsilon, +\infty), \\
				P^{(b)}(z), & z \in \partial D(b, \epsilon), \\
				P^{(0)}(z), & z \in \partial D(0, \epsilon), \\
				\begin{pmatrix}
					0 & 1 \\
					1 & 0
				\end{pmatrix},
				& z \in (0, b) \setminus \{ \epsilon, b - \epsilon \}.
			\end{cases}
		\end{equation}
		\item
		\begin{align}
			R_1(z) = {}& 1 + \bigO(z^{-1}) \quad \text{as $z \to \infty$ in $\compC$}, & R_2(z) = {}& \bigO(1) \quad \text{as $f(z) \to \infty$ in $\paraP$}.
		\end{align}
		\item
		\begin{align}
			R_1(z) = {}& \bigO(1), & R_2(z) = {}& \bigO(1), & \text{as $z \to 0$}, \\
			R_1(z) = {}& \bigO(1), & R_2(z) = {}& \bigO(1), & \text{as $z \to b$}.
		\end{align}
		\item
		At $z \in \rho \cup \{ -\pi^2/4 \} \cup \bar{\rho}$, $R_2(z)$ satisfies the same boundary condition as $Y_2(z)$ in \eqref{eq:Y_bd_cond}.
	\end{enumerate}
\end{RHP}

Similar to the idea used in the construction of the global parametrix, to estimate $R$ for large $n$, we now transform the RH problem for $R$ to a scalar one on the complex plane by defining
\begin{equation} \label{eq:defn_R_by_IJ}
	\R(s) =
	\begin{cases}
		R_1(\J(s)), & s \in \compC \setminus \overline{D} \text{ and } s \notin \Iinv_1(\Sigma^R), \\
		R_2(\J(s)), & s \in D \setminus [0, 1] \text{ and } s \notin \Iinv_2(\Sigma^R),
	\end{cases}
\end{equation}
where we recall that $D$ is the region bounded by the curves $\gamma_1$ and $\gamma_2$, $\Iinv_1: \compC \setminus [0, b] \to \compC \setminus \overline{D}$ and $\Iinv_2: \paraP \setminus [0, b]$ are defined in \eqref{eq:defn_I1} and \eqref{eq:defn_I2}, respectively.

We are now at the stage of describing the RH problem for $\R$. For this purpose, we define
\begin{equation}
	\begin{aligned}
		\SigmaR^{(1)} := {}& \Iinv_1(\Sigma^R_1 \cup \Sigma^R_2) \subseteq \compC \setminus \overline{D}, & \SigmaR^{(1')} := {}& \Iinv_2(\Sigma^R_1 \cup \Sigma^R_2) \subseteq D, \\
		\SigmaR^{(2)} := {}& \Iinv_1((b + \epsilon, +\infty)) \subseteq \compC \setminus \overline{D}, & \SigmaR^{(2')} := {}& \Iinv_2((b + \epsilon, +\infty)) \subseteq D, \\
		\SigmaR^{(3)} := {}& \Iinv_1(\partial D(b, \epsilon)) \subseteq \compC \setminus \overline{D}, & \SigmaR^{(3')} := {}& \Iinv_2(\partial D(b, \epsilon)) \subseteq D, \\
		\SigmaR^{(4)} := {}& \Iinv_1(\partial D(0, \epsilon)) \subseteq \compC \setminus \overline{D}, & \SigmaR^{(4')} := {}& \Iinv_2(\partial D(0, \epsilon)) \subseteq D,
	\end{aligned}
\end{equation}
and set
\begin{equation} \label{eq:defn_SigmaR}
	\SigmaR := \SigmaR^{(1)} \cup \SigmaR^{(1')} \cup \SigmaR^{(2)} \cup \SigmaR^{(2')} \cup \SigmaR^{(3)} \cup \SigmaR^{(3')} \cup \SigmaR^{(4)} \cup \SigmaR^{(4')}.
\end{equation}
See Figure \ref{fig:jump_R_scalar} for an illustration. We also define the following functions on each curve constituting $\SigmaR$: (Below $z = \J(s)$)

\begin{figure}[htb]
	\centering
	\includegraphics{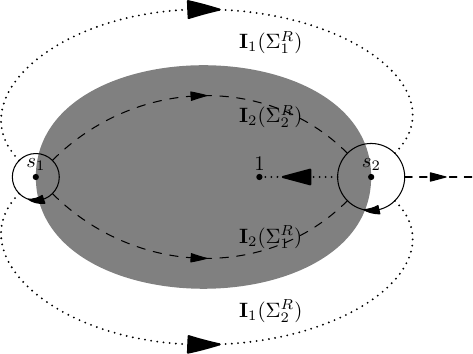}
	\caption{Contour $\SigmaR$. (It is also $\omega(\tilde{\SigmaR})$.) The solid and the dotted curves are the non-trivial jump contour for the RH problem for $\R$. (The solid and dashed curves, upon the mapping $\omega$, are the non-trivial jump contour for the RH problem for $\tilde{\R}$.)}
	\label{fig:jump_R_scalar}
\end{figure}
\begin{align}
	J_{\SigmaR^{(1)}}(s) = {}& (J_Q)_{21}(z), & s \in {}& \SigmaR^{(1)}, \label{eq:J(s)_1} \\
	J_{\SigmaR^{(2')}}(s) = {}& (J_Q)_{12}(z), & s \in {}& \SigmaR^{(2')}, \label{eq:J(s)_2} \\
	J^1_{\SigmaR^{(3)}}(s) = {}& (P^{(b)})_{11}(z) - 1, \quad J^2_{\SigmaR^{(3)}}(s) = (P^{(b)})_{21}(z), & s \in {}& \SigmaR^{(3)}, \label{eq:J(s)_3} \\
	J^1_{\SigmaR^{(3')}}(s) = {}& (P^{(b)})_{22}(z) - 1, \quad J^2_{\SigmaR^{(3')}}(s) = (P^{(b)})_{12}(z), & s \in {}& \SigmaR^{(3')}, \label{eq:J(s)_3_2} \\
	J^1_{\SigmaR^{(4)}}(s) = {}& (P^{(0)})_{11}(z) - 1, \quad J^2_{\SigmaR^{(4)}}(s) = (P^{(0)})_{21}(z), & s \in {}& \SigmaR^{(4)}, \label{eq:J(s)_4} \\
	J^1_{\SigmaR^{(4')}}(s) = {}& (P^{(0)})_{22}(z) - 1, \quad J^2_{\SigmaR^{(4')}}(s) = (P^{(0)})_{12}(z), & s \in {}& \SigmaR^{(4')}, \label{eq:J(s)_4_2}
\end{align}
where $z = \J(s)$ is in $\Sigma^R_1 \cup \Sigma^R_2$ in \eqref{eq:J(s)_1}; in $(b + \epsilon, +\infty)$ in \eqref{eq:J(s)_2}; in $\partial D(b, \epsilon)$ in \eqref{eq:J(s)_3} and \eqref{eq:J(s)_3_2}; in $\partial D(0, \epsilon)$ in \eqref{eq:J(s)_4} and \eqref{eq:J(s)_4_2}. With the aid of these functions, we further define an operator $\Delta_{\SigmaR}$ such that for any complex-valued function $f(s)$ defined on $\SigmaR$, $\Delta_{\SigmaR}$ transforms it linearly into a function $\Delta_{\SigmaR}f$ that is also a complex-valued function defined on $\SigmaR$, with expression
\begin{equation} \label{eq:defn_Delta_SigmaR}
	(\Delta_{\SigmaR} f)(s) =
	\begin{cases}
		J_{\SigmaR^{(1)}}(s) f(\tilde{s}), & s \in \SigmaR^{(1)} \text{ and } \tilde{s} = \Iinv_2(\J(s)) \in \SigmaR^{(1')}, \\
		J_{\SigmaR^{(2')}}(s) f(\tilde{s}), & s \in \SigmaR^{(2')} \text{ and } \tilde{s} = \Iinv_1(\J(s)) \in \SigmaR^{(2)}, \\
		0, & s \in \SigmaR^{(1')} \cup \SigmaR^{(2)}, \\
		J^1_{\SigmaR^{(3)}}(s) f(s) + J^2_{\SigmaR^{(3)}}(s) f(\tilde{s}), & s \in \SigmaR^{(3)} \text{ and } \tilde{s} = \Iinv_2(\J(s)) \in \SigmaR^{(3')}, \\
		J^1_{\SigmaR^{(3')}}(s) f(s) + J^2_{\SigmaR^{(3')}}(s) f(\tilde{s}), & s \in \SigmaR^{(3')} \text{ and } \tilde{s} = \Iinv_1(\J(s)) \in \SigmaR^{(3)}, \\
		J^1_{\SigmaR^{(4)}}(s) f(s) + J^2_{\SigmaR^{(4)}}(s) f(\tilde{s}), & s \in \SigmaR^{(4)} \text{ and } \tilde{s} = \Iinv_2(\J(s)) \in \SigmaR^{(4')}, \\
		J^1_{\SigmaR^{(4')}}(s) f(s) + J^2_{\SigmaR^{(4')}}(s) f(\tilde{s}), & s \in \SigmaR^{(4')} \text{ and } \tilde{s} = \Iinv_1(\J(s)) \in \SigmaR^{(4)}.
	\end{cases}
\end{equation}
We note that all the functions $J_{\SigmaR^{(1)}}(s)$, \dots, $J^2_{\SigmaR^{(4')}}(s)$ that define $\Delta_{\SigmaR}$ in \eqref{eq:defn_Delta_SigmaR} are uniformly $\bigO(n^{-1})$. If we view $\Delta_{\SigmaR}$ as an operator from $L^2(\SigmaR)$ to $L^2(\SigmaR)$, then we have the estimate that for all large enough $n$, there is a constant $M_{\SigmaR}>0$ such that
\begin{equation}\label{eq:estOperator}
	\lVert \Delta_{\SigmaR} \rVert_{L^2(\SigmaR)} \leq M_{\SigmaR} n^{-1}.
\end{equation}

RH problem \ref{RHP:R} entails a scalar shifted RH problem for $\R$:

\begin{RHP} \hfill \label{RHP:R_scalar}
	\begin{enumerate}
		\item
		$\R(s)$ is analytic in $\compC \setminus \SigmaR$, where the contour $\SigmaR$ is defined in \eqref{eq:defn_SigmaR}.
		\item
		For $s \in \SigmaR$, we have
		\begin{equation}
			\R_+(s) - \R_-(s) = (\Delta_{\SigmaR} \R_-)(s),
		\end{equation}
		where $\Delta_{\SigmaR}$ is the operator defined in \eqref{eq:defn_Delta_SigmaR}.
		\item
		As $s \to \infty$, we have
		\begin{equation}
			\R(s) = 1 + \bigO(s^{-1}).
		\end{equation}
		\item 
		As $s \to 0$, we have $\R(s) = \bigO(1)$.
	\end{enumerate}
\end{RHP}

We have the following uniqueness result about the solution of the above RH problem.

\begin{lem} \label{lem:uniqueness_of_R}
	The function $\R(s)$ defined in \eqref{eq:defn_R_by_IJ} is the unique solution of RH problem \ref{RHP:R_scalar}.
\end{lem}

\begin{proof}
	Suppose $\R^{\sol}(s)$ is one solution to RH problem \ref{RHP:R_scalar}. Then, using \eqref{eq:defn_R_by_IJ} backwardly, we have a solution $R^{\sol}(z) = (R^{\sol}_1(z), R^{\sol}_2(z))$ to RH problem \ref{RHP:R}. From $R^{\sol}(s)$, we define the vector-valued function $U^{\sol}(z)$ on $D(0, \epsilon) \setminus [0, \epsilon)$ as (using \eqref{eq:V_in_U} backwardly)
	\begin{equation}
		(U^{\sol}_1(z), U^{\sol}_2(z)) = (R^{\sol}_1(z), R^{\sol}_2(z)) E^{(0)}(z),
	\end{equation}
	where $E^{(0)}(z)$ is defined in \eqref{eq:E^0_defn}. Then $U^{\sol}_1(z)$ and $U^{\sol}_2(z)$ can be defined analytically in $D(0, \epsilon) \setminus \{ 0 \}$, and at the isolated singularity $0$ they may only blow up like an inverse square root. Hence, $U^{\sol}_1(z)$ and $U^{\sol}_2(z)$ are actually analytic in $D(0, \epsilon)$. Next, we define  $Q^{\sol}(z) = (Q^{\sol}_1(z), Q^{\sol}_2(z))$ by (using \eqref{eq:R_1_defn}, \eqref{eq:R_2_defn} and \eqref{eq:U_from_Q} backwardly)
	\begin{align}
		Q^{\sol}_1(z) = {}&
		\begin{cases}
			R^{\sol}_1(z)(P^{(b)})_{11}(z) + R^{\sol}_2(z)(P^{(b)})_{21}(z), & z \in D(b, \epsilon) \setminus \Sigma, \\
			U^{\sol}_1(z)(\Pmodel^{(0)})_{11}(z) + U^{\sol}_2(z)(\Pmodel^{(0)})_{21}(z), & z \in D(0, \epsilon) \setminus \Sigma, \\
			R^{\sol}_1(z), & z \in \compC \setminus (\overline{D(b, \epsilon)} \cup \overline{D(0, \epsilon)} \cup \Sigma),
		\end{cases} \\
		Q^{\sol}_2(z) = {}&
		\begin{cases}
			R^{\sol}_1(z)(P^{(b)})_{12}(z) + R^{\sol}_2(z)(P^{(b)})_{22}(z), & z \in D(b, \epsilon) \setminus \Sigma, \\
			U^{\sol}_1(z)(\Pmodel^{(0)})_{12}(z) + U^{\sol}_2(z)(\Pmodel^{(0)})_{22}(z), & z \in D(0, \epsilon) \setminus \Sigma, \\
			R^{\sol}_2(z), & z \in \paraP \setminus (\overline{D(b, \epsilon)} \cup \overline{D(0, \epsilon)} \cup \Sigma),
		\end{cases}
	\end{align}
	We find that $Q^{\sol}_1(z)$ and $Q^{\sol}_2(z)$ can be defined analytically on $\compC \setminus \Sigma$ and $\paraP \setminus \Sigma$, respectively, and find that $Q^{\sol}(z)$ satisfies a variation of RH problem \ref{RHP:Q}, such that in Item \ref{enu:RHP:Q_zero}, the limit behaviour of $Q_1(z)$ as $z \to 0$ from outside of the lens is changed to $Q_1(z) = \bigO(z^{1/4} \log z)$, and in Item \ref{enu:RHP:Q_b}, the occurrences of $(z - b)^{\frac{1}{4}}$ in \eqref{eq:RHP:Q_b} are replaced by those of $(z - b)^{-\frac{1}{4}}$.
	
	Furthermore, we do the transforms $Y \to T \to S \to Q$ backwardly, and find that from $Q^{\sol}$ we can construct $Y^{\sol}(z) = (Y^{\sol}_1(z), Y^{\sol}_2(z))$ that satisfies a variation of RH problem \ref{RHP:Y} such that Item \ref{enu:RHP:Y_zero} is replaced by Item \ref{enu:RHP:Y_zero_alt} in Section \ref{subsubsec:uniqueness}. By the argument in Section \ref{subsubsec:uniqueness}, we have that $Y^{\sol}(z)$ is unique, and then $\R^{\sol}(s)$ is unique.
\end{proof}

Finally, we show that
\begin{lem}\label{lem:tRest}
	For all $s \in \compC \setminus \SigmaR$, we have the uniform convergence
	\begin{equation}\label{eq:esttR}
		\R(s)=1+\bigO(n^{-1}).
	\end{equation}
\end{lem}
This lemma immediately yields that
\begin{align} \label{eq:P_1P_2_bounded}
	R_1(z) = {}& 1 + \bigO(n^{-1}) \quad \text{uniformly in $\compC \setminus \Sigma^R$}, & R_2(z) = {}& 1 + \bigO(n^{-1}) \quad \text{uniformly in $\paraP \setminus \Sigma^R$}.
\end{align}

\begin{proof}[Proof of Lemma \ref{lem:tRest}]
	We use the strategy proposed in \cite{Claeys-Wang11}, and start with the claim that $\R$ satisfies the integral equation
	\begin{equation} \label{eq:constr_R}
		\R(s) = 1 + \mathcal{C}(\Delta_{\SigmaR} \R_-)(s),
	\end{equation}
	where $\mathcal{C}$ is the Cauchy transform on $\SigmaR$, such that for any $g(s)$ defined on $\SigmaR$,
	\begin{equation}
		\mathcal{C}g(s) = \frac{1}{2\pi i} \int_{\SigmaR} \frac{g(\xi)}{\xi - s} d\xi, \quad s\in\compC \setminus \SigmaR.
	\end{equation}
	To verify \eqref{eq:constr_R}, by the uniqueness of RH problem \ref{RHP:R_scalar}, it suffices to show that the right-hand side of \eqref{eq:constr_R} satisfies RH problem \ref{RHP:R_scalar}, and it is straightforward.
	
	\eqref{eq:constr_R} can be written as
	\begin{equation}\label{eq:tRs-1split}
		\R(s) - 1 = \frac{1}{2\pi i} \int_{\SigmaR} \frac{\Delta_{\SigmaR}(\R_- - 1)(\xi)}{\xi - s} d\xi + \frac{1}{2\pi i} \int_{\SigmaR} \frac{\Delta_{\SigmaR}(1)(\xi)}{\xi - s} d\xi,\quad s\in\compC \setminus \SigmaR.
	\end{equation}
	Below we estimate the two terms on the right-hand side of the above formula.
	
	By taking the limit where $s$ approaches the minus side of $\SigmaR$, we obtain from \eqref{eq:tRs-1split} that
	\begin{equation}\label{eq:tR-}
		\R_-(s) - 1 = \mathcal{C}_{\Delta_{\SigmaR}}(\R_- - 1)(s) + \mathcal{C}_-(\Delta_{\SigmaR}(1))(s),
	\end{equation}
	where $\mathcal{C}_-$ is the Hilbert-like transform defined on $\SigmaR$,
	\begin{align}
		\mathcal{C}_-g(s) = {}& \frac{1}{2\pi i} \lim_{s' \to s_-} \int_{\SigmaR} \frac{g(\xi)}{\xi - s'} d\xi, && \text{and} && \mathcal{C}_{\Delta_{\SigmaR}}f(s) = \mathcal{C}_-(\Delta_{\SigmaR}(f))(s), 
	\end{align}
	such that the limit $s' \to s_-$ is taken when approaching the contour from the minus side. Since the Hilbert-like operator $\mathcal{C}_-$ is bounded on $L^2(\SigmaR)$, we see from estimate \eqref{eq:estOperator} that the operator norm of $\mathcal{C}_{\Delta_{\SigmaR}}$ is also uniformly $\bigO(n^{-1})$ as $n\to \infty$. Hence, if $n$ is large enough, the operator $1-\mathcal{C}_{\Delta_{\SigmaR}}$ is invertible, and we could rewrite \eqref{eq:tR-} as
	\begin{equation}
		\R_-(s) - 1 = (1 - \mathcal{C}_{\Delta_{\SigmaR}})^{-1}( \mathcal{C}_-(\Delta_{\SigmaR}(1)))(s).
	\end{equation}
	As one can check directly that
	\begin{equation}\label{eq:DeltaR1norm}
		\lVert \Delta_{\SigmaR}(1) \rVert_{L^2(\SigmaR)} = \bigO(n^{-1}),
	\end{equation}
	combining the above two formulas gives us
	\begin{equation}\label{eq:tR-1norm}
		\lVert \R_- - 1 \rVert_{L^2(\SigmaR)} = \bigO(n^{-1}).
	\end{equation}
	By \eqref{eq:tRs-1split}, we have that for any fixed $\delta > 0$, if $\dist(s, \SigmaR) > \delta$, then
	\begin{equation}
		\begin{split}
			\lvert \R(s) - 1 \rvert & \leq  \frac{1}{2\pi} \left( \lVert \Delta_{\SigmaR}(\R_- - 1) \rVert_{L^2(\SigmaR)} + \lVert \Delta_{\SigmaR}(1) \rVert_{L^2(\SigmaR)} \right) \cdot \lVert \frac{1}{\xi - s} \rVert_{L^2(\SigmaR)} = \bigO(n^{-1}).
		\end{split}
	\end{equation}
	As a consequence, we conclude \eqref{eq:esttR} holds uniformly in $\{ s \in \compC : \dist(s, \SigmaR) > \delta \}$. Since we can deform the contour $\Sigma$ outside a neighbourhood of $1$, say $D(1, \epsilon')$, by varying the value of $\epsilon$ in \eqref{eq:R_1_defn} and \eqref{eq:R_2_defn}, choosing different shapes of $\Sigma^R_1$ and $\Sigma^R_2$ in \eqref{eq:shape_Sigma^R_i}, and deforming the jump contour $(b, \infty)$ as in Remark \ref{rmk:contour_deform}, we can then show that \eqref{eq:esttR} holds uniformly in $\{ s \in \compC : s \notin \SigmaR \text{ and } \lvert s - 1 \rvert \geq \epsilon' \}$. (We cannot freely deform $\SigmaR$ around $1$, because $\SigmaR$ needs to connect to $1$ as a vertex.) At last, in $D(1, \epsilon')$, $\R(s)$ satisfies a simple RH problem: its value on $\partial D(1, \epsilon')$ is uniformly $1 + \bigO(n^{-1})$, its limit at $1$ is $1$, and it has a jump along $[1, 1 + \epsilon')$, where the jump is given by $J_{\SigmaR^{(2')}}(s)$ that is exponentially small. Hence we also conclude that \eqref{eq:esttR} holds uniformly in $\{ s \in D(1, \epsilon') : s \notin \SigmaR \}$.
\end{proof}

\section{Asymptotic analysis for $q^{(n)}_{n + k}(f(z))$} \label{sec:asy_q}

In this section, we analyze $q^{(n)}_{n + k}(f(z))$ in the same way as we analyzed $p^{(n)}_{n + k}(z)$ in Section \ref{sec:asy_p}. Since the method is parallel, we omit some details. It is worth noting that the jump contours in this section can be taken to be the same as those in Section \ref{sec:asy_p}.

\subsection{RH problem for the polynomials}

Consider the following Cauchy transform of $q_j$:
\begin{equation}\label{eq:cauchyqj}
	\tilde{C} q_j(z) := \frac{1}{2\pi i} \int_{\realR_+} \frac{q_j(f(x))}{x - z} W^{(n)}_{\alpha}(x) dx,
\end{equation}
which is well-defined for $z \in \paraP \setminus \realR_+$. Since $W^{(n)}_{\alpha}(x)$ is real analytic and vanishes rapidly as $x \to +\infty$, we have the following asymptotic expansion for $\C q_j(z)$ as $z \in \paraP \setminus \realR_+$ and $\Re z \to +\infty$:
\begin{equation}
	\begin{split}
		\C q_j(z) = {}& \frac{-1}{2\pi i z} \int_{\realR_+} \frac{q_j(f(x))}{1 - x/z} W^{(n)}_{\alpha}(x) dx \\
		= {}& \frac{-1}{2\pi i z} \sum^M_{k = 0} \left( \int_{\realR_+} q_j(f(x)) x^k W^{(n)}_{\alpha}(x) dx \right) z^{-(k + 1)} + \bigO( z^{-(M + 2)}),
	\end{split}
\end{equation}
for any $M \in \natN$, uniformly in $\Im z$. Thus, due to the orthogonality,
\begin{equation}
	\C q_j(z) = \frac{-h^{(n)}_j}{2\pi i} z^{-(j + 1)} + \bigO(z^{-(j + 2)}),
\end{equation}
where $h^{(n)}_j$ is given in \eqref{eq:defn_h}.

Hence, we conclude that if we define the array
\begin{equation} \label{eq:defn_Y_tilde}
	\tilde{Y}(z) = \tilde{Y}^{(j, n)}(z) := (q_j(f(z)), C q_j(z)),
\end{equation}
then it satisfies the following conditions:

\begin{RHP} \hfill \label{RHP:Y_tilde}
	\begin{enumerate}
		\item
		$\tilde{Y} = (\tilde{Y}_1, \tilde{Y}_2)$, where $\tilde{Y}_1$ is analytic on $\paraP$, and $\tilde{Y}_2$ is analytic on $\compC \setminus \realR_+$.
		\item
		With the standard orientation of $\realR_+$,
		\begin{equation}
			\tilde{Y}_+(x) = \tilde{Y}_-(x)
			\begin{pmatrix}
				1 & W^{(n)}_{\alpha}(x) \\
				0 & 1
			\end{pmatrix},
			\quad \text{for $x \in \realR_+$}.
		\end{equation}
		\item
		As $f(z) \to \infty$ in $\paraP$ (i.e., $\Re z \to +\infty$), $\tilde{Y}_1(z) = f(z)^j + \bigO(f(z)^{j - 1})$.
		\item
		As $z \to \infty$ in $\compC$, $\tilde{Y}_2(z) = \bigO(z^{-(j + 1)})$.
		\item
		As $z \to 0$ in $\paraP$ or $\compC$,
		\begin{align} \label{eq:Y_bd_0_tilde}
			\tilde{Y}_1(z) = {}& \bigO(1), & \tilde{Y}_2(z) = {}&
			\begin{cases}
				\bigO(1), & \alpha > 0, \\
				\bigO(\log z), & \alpha = 0, \\
				\bigO(z^{\alpha}), & \alpha \in (-1, 0).
			\end{cases}
		\end{align}
		\item
		At $z \in \rho \cup \{ -\pi^2/4 \} \cup \bar{\rho}$, the limit $\tilde{Y}_1(z) := \lim_{w \to z \text{ in } \paraP} \tilde{Y}_1(w)$ exists and is continuous, and
		\begin{equation} \label{eq:tilde_Y_bd_cond_tilde}
			\tilde{Y}_1(z) = \tilde{Y}_1(\bar{z}).
		\end{equation}
	\end{enumerate}
\end{RHP}

Conversely, the RH problem for $\tilde{Y}$ has a unique solution given by \eqref{eq:defn_Y}. We omit the proof since it is analogous to that of RH problem \ref{RHP:Y} given in Section \ref{subsubsec:uniqueness}.

Below, we take $j = n + k$, where $k$ is a fixed integer, and our goal is to obtain the asymptotics for $\tilde{Y} = \tilde{Y}^{(n + k, n)}$ as $n \to \infty$.

\subsection{First transformation $\tilde{Y} \mapsto \tilde{T}$}

Analogous to \eqref{eq:trans_Y_to_T}, we denote $\tilde{Y} = \tilde{Y}^{(n + k, n)}$ and define $\tilde{T}$ as
\begin{equation}
	\tilde{T}(z) = e^{-\frac{n \ell}{2}} \tilde{Y}(z)
	\begin{pmatrix}
		e^{-n\tilde{\gfn}(z)} & 0 \\
		0 & e^{n\gfn(z)}
	\end{pmatrix}
	e^{\frac{n\ell}{2} \sigma_3}.
\end{equation}
Then $\tilde{T}$ satisfies an RH problem with the same domain of analyticity as $\tilde{Y}$, but with different asymptotic behavior and a different jump relation.

\begin{RHP} \hfill \label{RHP:T_tilde}
	\begin{enumerate}
		\item
		$\tilde{T} = (\tilde{T}_1, \tilde{T}_2)$, where $\tilde{T}_1$ is analytic in $\paraP \setminus \realR_+$, and $\tilde{T}_2$ is analytic in $\compC \setminus \realR_+$.
		\item
		$\tilde{T}$ satisfies the jump relation
		\begin{equation}
			\tilde{T}_+(x) = \tilde{T}_-(x) J_{\tilde{T}}(x), \quad \text{for $x \in \realR_+$},
		\end{equation}
		where
		\begin{equation}
			J_{\tilde{T}}(x) =
			\begin{pmatrix}
				e^{n(\tilde{\gfn}_-(x) - \tilde{\gfn}_+(x))} & x^{\alpha} h(x) e^{n(\gfn_-(x) + \gfntilde_+(x) - V(x) - \ell)} \\
				0 & e^{n(\gfn_+(x) - \gfn_-(x))}
			\end{pmatrix}.
		\end{equation}
		\item
		\begin{align} \label{eq:T_at_infty_tilde}
			\tilde{T}_1(z) = {}& f(z)^k + \bigO(f(z)^{k - 1}) \quad \text{as $f(z) \to \infty$ in $\paraP$}, & \tilde{T}_2(z) = {}& \bigO(z^{-(k + 1)}) \quad \text{as $z \to \infty$ in $\compC$}.
		\end{align}
		\item
		As $z \to 0$ in $\paraP$ or $\compC$, $\tilde{T}(z)$ has the same limiting behavior as $\tilde{Y}(z)$ in \eqref{eq:Y_bd_0_tilde}.
		\item
		\begin{align} \label{eq:T_at_b_tilde}
			\tilde{T}_1(z) = {}& \bigO(1), & \tilde{T}_2(z) = {}& \bigO(1), & \text{as $z \to b$}.
		\end{align}
		\item
		At $z \in \rho \cup \{ -\pi^2/4 \} \cup \bar{\rho}$, $\tilde{T}_1(z)$ satisfies the same boundary condition as $\tilde{Y}_1(z)$ in \eqref{eq:tilde_Y_bd_cond_tilde}.
	\end{enumerate}
\end{RHP}

\subsection{Second transformation $\tilde{T} \mapsto \tilde{S}$} \label{sec:T_to_S_tilde}

Analogous to \eqref{eq:J_T_decomposition}, for $x \in (0, b)$, we decompose the jump matrix $J_{\tilde{T}}(x)$ as
\begin{multline} \label{eq:J_T_decomposition_tilde}
	\begin{pmatrix}
		1 & 0 \\
		\frac{1}{x^{\alpha} h(x)} e^{-n \phi_-(x)} & 1
	\end{pmatrix}
	\begin{pmatrix}
		0 & x^{\alpha} h(x) e^{n(\gfntilde_-(x) + \gfn_+(x) - V(x) - \ell)} \\
		-\dfrac{1}{x^{\alpha} h(x) e^{n(\gfntilde_-(x) + \gfn_+(x) - V(x) - \ell)}} & 0
	\end{pmatrix} \\
	\times
	\begin{pmatrix}
		1 & 0 \\
		\frac{1}{x^{\alpha} h(x)} e^{-n \phi_+(x)} & 1
	\end{pmatrix}.
\end{multline}

Then, analogous to \eqref{eq:trans_S_in_T}, define
\begin{equation}
	\tilde{S}(z) :=
	\begin{cases}
		\tilde{T}(z), & \text{outside the lens}, \\
		\tilde{T}(z)
		\begin{pmatrix}
			1 & 0 \\
			z^{-\alpha} h(z)^{-1} e^{-n \phi(z)} & 1
		\end{pmatrix},
		& \text{in the lower part of the lens}, \\
		\tilde{T}(z)
		\begin{pmatrix}
			1 & 0 \\
			-z^{-\alpha} h(z)^{-1} e^{-n \phi(z)} & 1
		\end{pmatrix},
		& \text{in the upper part of the lens}.
	\end{cases}
\end{equation}
From the definition of $\tilde{S}$ and the decomposition of $J_{\tilde{T}}(x)$ in \eqref{eq:J_T_decomposition}, we have, analogous to RH problem \ref{RHP:S}, that

\begin{RHP} \hfill \label{RHP:S_tilde}
	\begin{enumerate}
		\item
		$\tilde{S} = (\tilde{S}_1, \tilde{S}_2)$, where $\tilde{S}_1$ is analytic in $\paraP \setminus \Sigma$, and $\tilde{S}_2$ is analytic in $\compC \setminus \Sigma$.
		\item
		We have
		\begin{equation}
			\tilde{S}_+(z) = \tilde{S}_-(z) J_{\tilde{S}}(z), \quad \text{for $z \in \Sigma$},
		\end{equation}
		where
		\begin{equation}
			J_{\tilde{S}}(z) =
			\begin{cases}
				\begin{pmatrix}
					1 & 0 \\
					z^{-\alpha} h(z)^{-1} e^{-n \phi(z)} & 1
				\end{pmatrix},
				& \text{for $z \in \Sigma_1 \cup \Sigma_2$}, \\
				\begin{pmatrix}
					0 & z^{\alpha} h(z) \\
					-z^{-\alpha} h(z)^{-1} & 0
				\end{pmatrix},
				& \text{for $z \in (0, b)$}, \\
				\begin{pmatrix}
					1 & z^{\alpha} h(z) e^{n \phi(z)} \\
					0 & 1
				\end{pmatrix},
				& \text{for $z \in (b, \infty)$}.
			\end{cases}
		\end{equation}
		\item
		As $z \to \infty$ in $\paraP$ or $\compC$, $\tilde{S}(z)$ has the same limiting behavior as $\tilde{T}(z)$ in \eqref{eq:T_at_infty_tilde}.
		\item
		As $z \to 0$ in $\paraP \setminus \Sigma$, we have
		\begin{equation}
			\tilde{S}_1(z) =
			\begin{cases}
				\bigO(z^{-\alpha}), & \text{$\alpha > 0$ and $z$ inside the lens}, \\
				\bigO(\log z), & \text{$\alpha = 0$ and $z$ inside the lens}, \\
				\bigO(1), & \text{$z$ outside the lens or $-1 < \alpha < 0$}.
			\end{cases}
		\end{equation}
		\item
		As $z \to 0$ in $\compC$, $\tilde{S}_2$ has the same behavior as $\tilde{Y}_2(z)$ in \eqref{eq:Y_bd_0_tilde}.
		\item
		As $z \to b$, $\tilde{S}(z)$ has the same limiting behavior as $\tilde{T}(z)$ in \eqref{eq:T_at_b_tilde}.
		\item
		At $z \in \rho \cup \{ -\pi^2/4 \} \cup \bar{\rho}$, $\tilde{S}_1(z)$ satisfies the same boundary condition as $\tilde{Y}_1(z)$ in \eqref{eq:tilde_Y_bd_cond_tilde}.
	\end{enumerate}
\end{RHP}

\subsection{Construction of the global parametrix}

Analogous to RH problem \ref{RHP:Pinfinity}, we construct the following:
\begin{RHP} \hfill \label{RHP:Pinfinity_tilde}
	\begin{enumerate}
		\item
		$\tilde{P}^{(\infty)} = (\tilde{P}^{(\infty)}_1, \tilde{P}^{(\infty)}_2)$, where $\tilde{P}^{(\infty)}_1$ is analytic in $\paraP \setminus [0, b]$, and $\tilde{P}^{(\infty)}_2$ is analytic in $\compC \setminus [0, b]$.
		\item
		For $x \in (0, b)$, we have
		\begin{equation}
			\tilde{P}^{(\infty)}_+(x) = \tilde{P}^{(\infty)}_-(x)
			\begin{pmatrix}
				0 & x^{\alpha} h(x) \\
				-x^{-\alpha} h(x)^{-1} & 0
			\end{pmatrix}.
		\end{equation}
		\item
		As $z \to \infty$ in $\paraP$ or $\compC$, $\tilde{P}^{(\infty)}(z)$ has the same limiting behavior as $\tilde{T}(z)$ in \eqref{eq:T_at_infty_tilde}.
		\item
		At $z \in \rho \cup \{ -\pi^2/4 \} \cup \bar{\rho}$, $\tilde{P}^{(\infty)}_1(z)$ satisfies the same boundary condition as $\tilde{Y}_1(z)$ in \eqref{eq:tilde_Y_bd_cond_tilde}.
	\end{enumerate}
\end{RHP}

To construct a solution to the above RH problem, we set, analogous to \eqref{eq:P_scalar_defn},
\begin{equation}
	\tilde{\P}(s) :=
	\begin{cases}
		\tilde{P}^{(\infty)}_2(\J(s)), & s \in \compC \setminus \overline{D}, \\
		\tilde{P}^{(\infty)}_1(\J(s)), & s \in D \setminus [0, 1].
	\end{cases}
\end{equation}
Like $\P(s)$ in \eqref{eq:P_scalar_defn}, $\tilde{\P}$ is well-defined on $[0, 1)$ by continuation. Analogous to RH problem \ref{RHP:P_scalar} for $\P$, $\tilde{\P}$ satisfies the following:

\begin{RHP} \hfill
	\begin{enumerate}
		\item
		$\tilde{\P}$ is analytic in $\compC \setminus (\gamma_1 \cup \gamma_2 \cup \{ 1 \})$.
		\item
		For $s \in \gamma_1 \cup \gamma_2$, $\tilde{\P}_+(s) = \tilde{\P}_-(s) J_{\tilde{\P}}(s)$, where
		\begin{equation}
			J_{\tilde{\P}}(s) =
			\begin{cases}
				\J^{\alpha}(s) h(\J(s)), & s \in \gamma_1, \\
				-\J^{-\alpha}(s) h(\J(s))^{-1}, & s \in \gamma_2.
			\end{cases}
		\end{equation}

		\item
		As $s \to \infty$, $\tilde{\P}(s) = \bigO(s^{-(k + 1)})$.
		\item
		As $s \to 1$, $\tilde{\P}(s) = e^{kc} (s - 1)^{-k} + \bigO((s - 1)^{-(k + 1)})$.
	\end{enumerate}
\end{RHP}

A solution $\tilde{\P}$ to the above RH problem is explicitly given by
\begin{equation} \label{eq:defn_P_scalar_tilde}
	\tilde{\P}(s) =
	\begin{cases}
		\tilde{G}_k(s), & s \in D \setminus \{ 1 \}, \\
		\frac{(1 - s_1)^{\alpha + \frac{1}{2}} \sqrt{s_2 - 1} i \tilde{D}(1)^{-1} e^{kc} \J(s)^{\alpha}}{(s - s_1)^{\alpha} (s - 1)^k \sqrt{(s - s_1)(s - s_2)}} D(s)^{-1}, & s \in \compC \setminus \overline{D},
	\end{cases}
\end{equation}
where $D(s)$ is defined in \eqref{eq:defn_D(s)}, the square root is taken in $\compC \setminus \gamma_2$ with $\sqrt{(s - s_1)(s - s_2)} \sim s$ as $s \to \infty$. Later in this paper, we take \eqref{eq:defn_P_scalar_tilde} as the definition of $\tilde{\P}(s)$.

Based on this solution, we construct the solution to RH problem \ref{RHP:Pinfinity_tilde} as
\begin{align}
	\tilde{P}^{(\infty)}_2(z) = {}& \tilde{\P}(\Iinv_1(z)), & z \in {}& \compC \setminus [0, b], \label{eq:Pinfty_1_tilde_in_scalar} \\
	\tilde{P}^{(\infty)}_1(z) = {}& \tilde{\P}(\Iinv_2(z)) = \tilde{G}_k(\Iinv_2(z)), & z \in {}& \paraP \setminus [0, b]. \label{eq:Pinfty_2_tilde_in_scalar}
\end{align}

By direct calculation in Section \ref{sec:J_x_prop}, we have, analogous to \eqref{eq:Pinfty_at_0} and \eqref{eq:Pinfty_at_b},
\begin{align}
	\tilde{P}^{(\infty)}_1(z) = {}& \bigO(z^{-\frac{\alpha}{2} - \frac{1}{4}}), & \tilde{P}^{(\infty)}_2(z) = {}& \bigO(z^{\frac{\alpha}{2} - \frac{1}{4}}), & \text{as } z \to 0, \\
	\tilde{P}^{(\infty)}_1(z) = {}& \bigO(z^{-\frac{1}{4}}), & \tilde{P}^{(\infty)}_2(z) = {}& \bigO(z^{-\frac{1}{4}}), & \text{as } z \to b.
\end{align}

\subsection{Third transformation $\tilde{S} \mapsto \tilde{Q}$}

Noting that $\tilde{P}^{(\infty)}_1(z) \neq 0$ for all $z \in \paraP \setminus [0, b]$ and $\tilde{P}^{(\infty)}_2(z) \neq 0$ for all $z \in \compC \setminus [0, b]$, we define, analogous to \eqref{eq:defn_Q},
\begin{equation} \label{eq:defn_Q_tilde}
	\tilde{Q}(z) = (\tilde{Q}_1(z), \tilde{Q}_2(z)) = \left( \frac{\tilde{S}_1(z)}{\tilde{P}^{(\infty)}_1(z)}, \frac{\tilde{S}_2(z)}{\tilde{P}^{(\infty)}_2(z)} \right).
\end{equation}
Analogous to RH problem \ref{RHP:Q}, $\tilde{Q}$ satisfies the following RH problem:

\begin{RHP} \hfill \label{RHP:Q_tilde}
	\begin{enumerate}
		\item
		$\tilde{Q} = (\tilde{Q}_1, \tilde{Q}_2)$, where $\tilde{Q}_1$ is analytic in $\paraP \setminus \Sigma$, and $\tilde{Q}_2$ is analytic in $\compC \setminus \Sigma$.
		\item
		For $z \in \Sigma$, we have
		\begin{equation} \label{eq:jump_Q_tilde}
			\tilde{Q}_+(z) = \tilde{Q}_-(z) J_{\tilde{Q}}(z),
		\end{equation}
		where
		\begin{equation} \label{eq:defn_J_Q_tilde}
			J_{\tilde{Q}}(z) =
			\begin{cases}
				\begin{pmatrix}
					1 & 0 \\
					z^{-\alpha} h(z)^{-1} \frac{\tilde{P}^{(\infty)}_2(z)}{\tilde{P}^{(\infty)}_1(z)} e^{-n \phi(z)} & 1
				\end{pmatrix},
				& z \in \Sigma_1 \cup \Sigma_2, \\
				\begin{pmatrix}
					0 & 1 \\
					1 & 0
				\end{pmatrix},
				& z \in (0, b), \\
				\begin{pmatrix}
					1 & z^{\alpha} h(z) \frac{\tilde{P}^{(\infty)}_1(z)}{\tilde{P}^{(\infty)}_2(z)} e^{-n \phi(z)} \\
					0 & 1
				\end{pmatrix},
				& z \in (b, \infty).
			\end{cases}
		\end{equation}
		\item
		\begin{align}
			\tilde{Q}_1(z) = {}& 1 + \bigO(f(z)^{-1}), \quad \text{as $f(z) \to \infty$ in $\paraP$}, & \tilde{Q}_2(z) = {}& \bigO(z^{-1}), \quad \text{as $z \to \infty$ in $\compC$}.
		\end{align}
		\item \label{enu:RHP:Q_zero_tilde}
		As $z \to 0$ in $\paraP \setminus \Sigma$, we have
		\begin{equation}
			\tilde{Q}_1(z) =
			\begin{cases}
				\bigO(z^{-\frac{\alpha}{2} + \frac{1}{4}}), & \text{$\alpha > 0$ and $z$ inside the lens}, \\
				\bigO(z^{\frac{1}{4}} \log z), & \text{$\alpha = 0$ and $z$ inside the lens}, \\
				\bigO(z^{\frac{\alpha}{2} + \frac{1}{4}}), & \text{$z$ outside the lens or $-1 < \alpha < 0$}.
			\end{cases}
		\end{equation}
		\item
		As $z \to 0$ in $\compC$, we have
		\begin{equation}
			\tilde{Q}_2(z) =
			\begin{cases}
				\bigO(z^{-\frac{\alpha}{2} + \frac{1}{4}}), & \alpha > 0, \\
				\bigO(z^{\frac{1}{4}} \log z), & \alpha = 0, \\
				\bigO(z^{\frac{\alpha}{2} + \frac{1}{4}}), & \alpha \in (-1, 0).
			\end{cases}
		\end{equation}
		\item
		\begin{align}
			\tilde{Q}_1(z) = {}& \bigO((z - b)^{\frac{1}{4}}), & \tilde{Q}_2(z) = {}& \bigO((z - b)^{\frac{1}{4}}), & \text{as $z \to b$}.
		\end{align}
		\item
		At $z \in \rho \cup \{ -\pi^2/4 \} \cup \bar{\rho}$, $\tilde{Q}_1(z)$ satisfies the same boundary condition as $\tilde{Y}_1(z)$ in \eqref{eq:tilde_Y_bd_cond_tilde}.
	\end{enumerate}
\end{RHP}

\subsection{Construction of local parametrix near $b$} \label{subsec:Airy_para:tilde}

Let, analogous to \eqref{eq:defn_g^(b)_i},
\begin{align} \label{eq:defn_g^(b)_i_tilde}
	\tilde{g}^{(b)}_1(z) = {}& \frac{h(z)^{-1}}{\tilde{P}^{(\infty)}_1(z)}, & \tilde{g}^{(b)}_2(z) = {}& \frac{z^{\alpha}}{\tilde{P}^{(\infty)}_2(z)},
\end{align}
and define, analogous to \eqref{eq:defn_Pmodel},
\begin{equation} \label{eq:defn_Pmodel_tilde}
	\tilde{\Pmodel}^{(b)}(z) := \Psi^{(\Ai)}(n^{\frac{2}{3}} f_b(z))
	\begin{pmatrix}
		e^{-\frac{n}{2} \phi(z)} \tilde{g}^{(b)}_1(z) & 0 \\
		0 & e^{\frac{n}{2} \phi(z)} \tilde{g}^{(b)}_2(z)
	\end{pmatrix},
	\quad z \in D(b, \epsilon) \setminus \Sigma.
\end{equation}
From \eqref{eq:defn_g^(b)_i_tilde} and the RH problem \ref{RHP:Pinfinity_tilde} satisfied by $\tilde{P}^{(\infty)}(z)$, we have, analogous to \eqref{eq:gb_jump} and \eqref{eq:gb_bd},
\begin{align}
	\tilde{g}^{(b)}_{1, +}(x) = {}& -\tilde{g}^{(b)}_{2, -}(x), & \tilde{g}^{(b)}_{2, +}(x) = {}& \tilde{g}^{(b)}_{1, -}(x), && \text{for $x \in (b - \epsilon, b)$}, \\
	\tilde{g}^{(b)}_1(z) = {}& \bigO((z - b)^{\frac{1}{4}}), & \tilde{g}^{(b)}_2(z) = {}& \bigO((z - b)^{\frac{1}{4}}), && \text{as $z \to b$}.
\end{align}
Then, analogous to RH problem \ref{RHP:Pmodel^b}, we have the following RH problem satisfied by $\tilde{\Pmodel}^{(b)}(z)$:
\begin{RHP} \hfill
	\begin{enumerate}
		\item
		$\tilde{\Pmodel}^{(b)}(z)$ is a $2 \times 2$ matrix-valued function analytic for $z \in D(b, \epsilon) \setminus \Sigma$.
		\item
		For $z \in \Sigma \cap D(b, \epsilon)$, we have
		\begin{equation} \label{eq:jump_Pmodel_tilde}
			\tilde{\Pmodel}^{(b)}_+(z) = \tilde{\Pmodel}^{(b)}_-(z) J_{\tilde{Q}} (z),
		\end{equation}
		where $J_{\tilde{Q}}(z)$ is defined in \eqref{eq:defn_J_Q_tilde}.
		\item
		As $z \to b$, the limiting behavior of $\tilde{\Pmodel}^{(b)}(z)$ and $(\tilde{\Pmodel}^{(b)})^{-1}(z)$ is the same as that of $\Pmodel^{(b)}(z)$ and $(\Pmodel^{(b)})^{-1}(z)$ in \eqref{eq:Pmodel^b_asy}.
		\item
		For $z \in \partial D(b, \epsilon)$, we have, as $n \to \infty$,
		\begin{equation}
			\tilde{E}^{(b)}(z) \tilde{\Pmodel}^{(b)}(z) = I + \bigO(n^{-1}),
		\end{equation}
		where
		\begin{equation} \label{eq:defn_E^b_tilde}
			\begin{split}
				\tilde{E}^{(b)}(z) = {}& \frac{1}{\sqrt{2}}
				\begin{pmatrix}
					\tilde{g}^{(b)}_1(z) & 0 \\
					0 & \tilde{g}^{(b)}_2(z)
				\end{pmatrix}^{-1}
				e^{\frac{\pi i}{4} \sigma_3}
				\begin{pmatrix}
					1 & -1 \\
					1 & 1
				\end{pmatrix}
				\begin{pmatrix}
					n^{\frac{1}{6}} f_b(z)^{\frac{1}{4}} & 0 \\
					0 & n^{-\frac{1}{6}} f_b(z)^{-\frac{1}{4}}
				\end{pmatrix}
				\tilde{\Pmodel}^{(b)}(z) \\
				= {}& I + \bigO(n^{-1}).
			\end{split}
		\end{equation}
	\end{enumerate}
\end{RHP}

We now define a $2 \times 2$ matrix-valued function
\begin{equation} \label{eq:P^b_in_E_tilde}
	\tilde{P}^{(b)}(z) = \tilde{E}^{(b)}(z) \tilde{\Pmodel}^{(b)}(z), \quad z \in D(b, \epsilon) \setminus \Sigma,
\end{equation}
where $\tilde{\Pmodel}^{(b)}$ is given in \eqref{eq:defn_Pmodel_tilde} and $\tilde{E}^{(b)}(z)$ is defined in \eqref{eq:defn_E^b_tilde}. Like $E^{(b)}(z)$ defined in \eqref{eq:defn_E^b}, it is straightforward to see that $\tilde{E}^{(b)}(z)$ is analytic on $D(b, \epsilon) \setminus (b - \epsilon, b]$, and for $x \in (b - \epsilon, b)$
\begin{equation}
	\tilde{E}^{(b)}_+(x) \tilde{E}^{(b)}_-(x)^{-1} =
	\begin{pmatrix}
		0 & 1 \\
		1 & 0
	\end{pmatrix},
\end{equation}
and as $z \to b$,
\begin{align}
	\tilde{E}^{(b)}(z) = {}&
	\begin{pmatrix}
		\bigO(1) & \bigO(z^{-\frac{1}{2}}) \\
		\bigO(1) & \bigO(z^{-\frac{1}{2}})
	\end{pmatrix}, &
	\tilde{E}^{(b)}(z)^{-1} = {}&
	\begin{pmatrix}
		\bigO(1) & \bigO(1) \\
		\bigO(z^{\frac{1}{2}}) & \bigO(z^{\frac{1}{2}})
	\end{pmatrix}.
\end{align}
Hence, analogous to RH problem \ref{RHP:P^b_actual}, we have the following RH problem satisfied by $\tilde{P}^{(b)}(z)$:
\begin{RHP} \hfill
	\begin{enumerate}
		\item
		$\tilde{P}^{(b)}(z)$ is analytic in $D(b, \epsilon) \setminus \Sigma$.
		\item
		For $z \in \Sigma \cap D(b, \epsilon)$, we have
		\begin{equation}
			\tilde{P}^{(b)}_+(z) =
			\begin{cases}
				\tilde{P}^{(b)}_-(z) J_{\tilde{Q}}(z), & z \in \Sigma \cap D(b, \epsilon) \setminus (b - \epsilon, b], \\
				\begin{pmatrix}
					0 & 1 \\
					1 & 0
				\end{pmatrix}
				\tilde{P}^{(b)}_-(z) J_{\tilde{Q}}(z), & z \in (b - \epsilon, b).
			\end{cases}
		\end{equation}
		\item
		As $z \to b$, the limiting behavior of $\tilde{P}^{(b)}(z)$ and $(\tilde{P}^{(b)})^{-1}(z)$ is the same as that of $P^{(b)}(z)$ and $(P^{(b)})^{-1}(z)$ in \eqref{eq:P^b_asy}.
		\item
		For $z$ on the boundary $\partial D(b, \epsilon)$, we have, as $n \to \infty$, $\tilde{P}^{(b)}(z) = I + \bigO(n^{-1})$.
	\end{enumerate}
\end{RHP}

Finally, analogous to \eqref{eq:V^b_in_Q_P}, we define a vector-valued function $\tilde{V}^{(b)}$ by
\begin{equation} \label{eq:V^b_in_Q_P_tilde}
	\tilde{V}^{(b)}(z) = \tilde{Q}(z) \tilde{P}^{(b)}(z)^{-1}, \quad z \in D(b, \epsilon) \setminus \Sigma,
\end{equation}
where $\tilde{Q}(z)$ is defined in \eqref{eq:defn_Q_tilde}. It satisfies the following RH problem, which is analogous to RH problem \ref{RHP:Vb}:
\begin{RHP} \hfill
	\begin{enumerate}
		\item
		$\tilde{V}^{(b)} = (\tilde{V}^{(b)}_1, \tilde{V}^{(b)}_2)$ is analytic in $D(b, \epsilon) \setminus (b - \epsilon, b]$.
		\item
		For $x \in (b - \epsilon, b)$, we have
		\begin{equation}
			\tilde{V}^{(b)}_+(x) = \tilde{V}^{(b)}_-(x)
			\begin{pmatrix}
				0 & 1 \\
				1 & 0
			\end{pmatrix}.
		\end{equation}
		\item
		\begin{align}
			\tilde{V}^{(b)}_1(z) = {}& \bigO(1), & \tilde{V}^{(b)}_2(z) = {}& \bigO(1), & \text{as $z \to b$}.
		\end{align}
		\item
		For $z \in \partial D(b, \epsilon)$, we have, as $n \to \infty$, $\tilde{V}^{(b)}(z) = \tilde{Q}(z) (I + \bigO(n^{-1}))$.
	\end{enumerate}
\end{RHP}

\subsection{Construction of local parametrix near $0$}

Let, analogous to \eqref{eq:defn_g^(0)_i},
\begin{align}
	\tilde{g}^{(0)}_1(z) = {}& \frac{(-z)^{-\alpha/2}/h(z)}{\tilde{P}^{(\infty)}_1(z)}, & \tilde{g}^{(0)}_2(z) = {}& \frac{(-z)^{\alpha/2}}{\tilde{P}^{(\infty)}_2(z)},
\end{align}
where $(-z)^{\pm \alpha/2}$ takes the principal branch, and define, analogous to \eqref{eq:defn_Pmodel_0},
\begin{equation} \label{eq:defn_Pmodel_0_tilde}
	\tilde{\Pmodel}^{(0)}(z) = \Psi^{(\Be)}_{\alpha}(n^2 f_0(z))
	\begin{pmatrix}
		e^{-\frac{n}{2} \phi(z)} \tilde{g}^{(0)}_1(z) & 0 \\
		0 & e^{\frac{n}{2} \phi(z)} \tilde{g}^{(0)}_2(z)
	\end{pmatrix},
	\quad z \in D(0, \epsilon) \setminus \Sigma.
\end{equation}
Analogous to \eqref{eq:g0_jump} and \eqref{eq:g0_bd}, we have
\begin{align}
	\tilde{g}^{(0)}_{1, +}(x) = {}& -\tilde{g}^{(0)}_{2, -}(x), & \tilde{g}^{(0)}_{2, +}(x) = {}& \tilde{g}^{(0)}_{1, -}(x), && \text{for $x \in (0, \epsilon)$}, \\
	\tilde{g}^{(0)}_1(z) = {}& \bigO(z^{\frac{1}{4}}), & \tilde{g}^{(0)}_2(z) = {}& \bigO(z^{\frac{1}{4}}), && \text{as $z \to 0$}.
\end{align}
Analogous to RH problem \ref{RHP:Pmodel}, we have the following RH problem satisfied by $\tilde{\Pmodel}^{(0)}(z)$:
\begin{RHP} \hfill
	\begin{enumerate}
		\item
		$\tilde{\Pmodel}^{(0)}(z)$ is a $2 \times 2$ matrix-valued function analytic for $z \in D(0, \epsilon) \setminus \Sigma$.
		\item
		For $z \in \Sigma \cap D(0, \epsilon)$, we have
		\begin{equation}
			\tilde{\Pmodel}^{(0)}_+(z) = \tilde{\Pmodel}^{(0)}_-(z) J_{\tilde{Q}}(z),
		\end{equation}
		where $J_{\tilde{Q}}(z)$ is defined in \eqref{eq:defn_J_Q_tilde}.
		\item
		As $z \in \partial D(0, \epsilon)$, we have
		\begin{equation}
			\tilde{E}^{(0)}(z) \tilde{\Pmodel}^{(0)}(z) = (I + \bigO(n^{-1})),
		\end{equation}
		where
		\begin{equation}
			\tilde{E}^{(0)}(z) = \frac{1}{\sqrt{2}}
			\begin{pmatrix}
				\tilde{g}^{(0)}_1(z) & 0 \\
				0 & \tilde{g}^{(0)}_2(z)
			\end{pmatrix}^{-1}
			\begin{pmatrix}
				1 & i \\
				i & 1
			\end{pmatrix}
			\begin{pmatrix}
				n^{\frac{1}{2}} f_0(z)^{\frac{1}{4}} & 0 \\
				0 & n^{-\frac{1}{2}} f_0(z)^{-\frac{1}{4}}
			\end{pmatrix}
			(2\pi)^{\frac{1}{2} \sigma_3}.
		\end{equation}
		\item
		As $z \to 0$, the limiting behavior of $\tilde{\Pmodel}^{(0)}(z)$ and $(\tilde{\Pmodel}^{(0)})^{-1}(z)$ is the same as that of $\Pmodel^{(0)}(z)$ and $(\Pmodel^{(0)})^{-1}(z)$ in \eqref{eq:Pmodel^0_asy_1}--\eqref{eq:Pmodel^0_asy_4}.
	\end{enumerate}
\end{RHP}
Like $E^{(0)}(z)$, it is straightforward to see that $\tilde{E}^{(0)}(z)$ is analytic on $D(0, \epsilon) \setminus [0, \epsilon)$, and for $x \in (0, \epsilon)$, like \eqref{eq:E0_jump},
\begin{equation}
	\tilde{E}^{(0)}_+(x) \tilde{E}^{(0)}_-(x)^{-1} =
	\begin{pmatrix}
		0 & 1 \\
		1 & 0
	\end{pmatrix},
\end{equation}
and as $z \to 0$, like \eqref{eq:E^0_defn},
\begin{align}
	\tilde{E}^{(0)}(z) = {}&
	\begin{pmatrix}
		\bigO(1) & \bigO(z^{-\frac{1}{2}}) \\
		\bigO(1) & \bigO(z^{-\frac{1}{2}})
	\end{pmatrix}, &
	\tilde{E}^{(0)}(z)^{-1} = {}&
	\begin{pmatrix}
		\bigO(1) & \bigO(1) \\
		\bigO(z^{\frac{1}{2}}) & \bigO(z^{\frac{1}{2}})
	\end{pmatrix}.
\end{align}
Then we define a $2 \times 2$ matrix-valued function
\begin{equation}
	\tilde{P}^{(0)}(z) = \tilde{E}^{(0)}(z) \tilde{\Pmodel}^{(0)}(z), \quad z \in D(0, \epsilon) \setminus \Sigma,
\end{equation}
where $\tilde{\Pmodel}^{(0)}$ is given in \eqref{eq:defn_Pmodel_0_tilde}. Hence, we have the following RH problem satisfied by $\tilde{P}^{(0)}(z)$:
\begin{RHP} \hfill
	\begin{enumerate}
		\item
		$\tilde{P}^{(0)}(z)$ is analytic in $D(0, \epsilon) \setminus \Sigma$.
		\item
		For $z \in \Sigma \cap D(0, \epsilon)$, we have
		\begin{equation}
			\tilde{P}^{(0)}_+(z) =
			\begin{cases}
				\tilde{P}^{(0)}_-(z) J_{\tilde{Q}}(z), & z \in \Sigma \cap D(0, \epsilon) \setminus [0, \epsilon), \\
				\begin{pmatrix}
					0 & 1 \\
					1 & 0
				\end{pmatrix}
				\tilde{P}^{(0)}_-(z) J_{\tilde{Q}}(z), & z \in (0, \epsilon).
			\end{cases}
		\end{equation}
		\item
		As $z \to 0$, the limiting behavior of $\tilde{P}^{(0)}(z)$ and $(\tilde{P}^{(0)})^{-1}(z)$ is the same as that of $P^{(0)}(z)$ and $(P^{(0)})^{-1}(z)$ in \eqref{eq:P^0_asy_1}--\eqref{eq:P^0_asy_4}.
		\item
		For $z$ on the boundary $\partial D(0, \epsilon)$, we have, as $n \to \infty$, $\tilde{P}^{(0)}(z) = I + \bigO(n^{-1})$.
	\end{enumerate}
\end{RHP}

Consider the vector-valued function
\begin{equation}
	\tilde{U}(z) = (\tilde{U}_1(z), \tilde{U}_2(z)) := \tilde{Q}(z) (\tilde{\Pmodel}^{(0)}(z))^{-1}, \quad z \in D(0, \epsilon) \setminus \Sigma,
\end{equation}
where $\tilde{Q}(z)$ is defined in \eqref{eq:defn_Q_tilde}, and then define the vector-valued function $\tilde{V}^{(0)}$ on $D(0, \epsilon) \setminus \Sigma$ by
\begin{equation}
	\tilde{V}^{(0)}(z) = (\tilde{V}^{(0)}_1(z), \tilde{V}^{(0)}_2(z)) := \tilde{Q}(z) (\tilde{P}^{(0)}(z))^{-1} = \tilde{U}(z) (\tilde{E}^{(0)}(z))^{-1}.
\end{equation}
Like $U(z)$ defined in \eqref{eq:U_from_Q}, we can show that $0$ is a removable singular point of $\tilde{U}_1(z)$ and $\tilde{U}_2(z)$. Hence, like RH problem \ref{RHP:V^0}, $\tilde{V}^{(0)}(z)$ satisfies the following RH problem:
\begin{RHP} \hfill
	\begin{enumerate}
		\item
		$\tilde{V}^{(0)}(z) = (\tilde{V}^{(0)}_1(z), \tilde{V}^{(0)}_2(z))$ is analytic in $D(0, \epsilon) \setminus [0, \epsilon)$.
		\item
		For $x \in (0, \epsilon)$, we have
		\begin{equation}
			\tilde{V}^{(0)}_+(x) = \tilde{V}^{(0)}_-(x)
			\begin{pmatrix}
				0 & 1 \\
				1 & 0
			\end{pmatrix}.
		\end{equation}
		\item
		As $z \to 0$, we have $\tilde{V}^{(0)}(z) = (\bigO(1), \bigO(1))$.
		\item
		As $z \in \partial D(0, \epsilon)$, $\tilde{V}^{(0)}(z) = \tilde{Q}(z)(I + \bigO(n^{-1}))$.
	\end{enumerate}
\end{RHP}

\subsection{Final transformation} \label{subsec:final_trans_tilde}

Analogous to \eqref{eq:R_1_defn} and \eqref{eq:R_2_defn}, we define $\tilde{R}(z) = (\tilde{R}_1(z), \tilde{R}_2(z))$ as
\begin{align}
	\tilde{R}_1(z) = {}&
	\begin{cases}
		\tilde{V}^{(b)}_1(z), & z \in D(b, \epsilon) \setminus (b - \epsilon, b], \\
		\tilde{V}^{(0)}_1(z), & z \in D(0, \epsilon) \setminus [0, \epsilon), \\
		\tilde{Q}_1(z), & z \in \paraP \setminus (\overline{D(b, \epsilon)} \cup \overline{D(0, \epsilon)} \cup \Sigma),
	\end{cases} \label{eq:R_1_defn_tilde} \\
	\tilde{R}_2(z) = {}&
	\begin{cases}
		\tilde{V}^{(b)}_2(z), & z \in D(b, \epsilon) \setminus (b - \epsilon, b], \\
		\tilde{V}^{(0)}_2(z), & z \in D(0, \epsilon) \setminus [0, \epsilon), \\
		\tilde{Q}_2(z), & z \in \compC \setminus (\overline{D(b, \epsilon)} \cup \overline{D(0, \epsilon)} \cup \Sigma).
	\end{cases} \label{eq:R_2_defn_tilde}
\end{align}
Then $\tilde{R}$ satisfies the following RH problem, which is analogous to RH problem \ref{RHP:R}:
\begin{RHP} \hfill \label{RHP:R_tilde}
	\begin{enumerate}
		\item
		$\tilde{R}(z) = (\tilde{R}_1(z), \tilde{R}_2(z))$, where $\tilde{R}_1(z)$ is analytic in $\paraP \setminus \Sigma^R$, and $\tilde{R}_2(z)$ is analytic in $\compC \setminus \Sigma^R$.
		\item
		$\tilde{R}(z)$ satisfies the following jump conditions:
		\begin{equation}
			\tilde{R}_+(z) = \tilde{R}_-(z)
			\begin{cases}
				J_{\tilde{Q}}(z), & z \in \Sigma^R_1 \cup \Sigma^R_2 \cup (b + \epsilon, +\infty), \\
				\tilde{P}^{(b)}(z), & z \in \partial D(b, \epsilon), \\
				\tilde{P}^{(0)}(z), & z \in \partial D(0, \epsilon), \\
				\begin{pmatrix}
					0 & 1 \\
					1 & 0
				\end{pmatrix},
				& z \in (0, b) \setminus \{ \epsilon, b - \epsilon \}.
			\end{cases}
		\end{equation}
		\item
		\begin{align}
			\tilde{R}_1(z) = {}& 1 + \bigO(f(z)^{-1}) \quad \text{as $f(z) \to \infty$ in $\paraP$}, & \tilde{R}_2(z) = {}& \bigO(1) \quad \text{as $z \to \infty$ in $\compC$}.
		\end{align}
		\item
		\begin{align}
			\tilde{R}_1(z) = {}& \bigO(1), & \tilde{R}_2(z) = {}& \bigO(1), & \text{as $z \to 0$}, \\
			\tilde{R}_1(z) = {}& \bigO(1), & \tilde{R}_2(z) = {}& \bigO(1), & \text{as $z \to b$}.
		\end{align}
		\item
		At $z \in \rho \cup \{ -\pi^2/4 \} \cup \bar{\rho}$, $\tilde{R}_1(z)$ satisfies the same boundary condition as $\tilde{Y}_1(z)$ in \eqref{eq:tilde_Y_bd_cond_tilde}.
	\end{enumerate}
\end{RHP}

Similar to the idea used in the construction of the global parametrix, to estimate $\tilde{R}$ for large $n$, we now transform the RH problem for $\tilde{R}$ to a scalar one on the complex $s$-plane. Let $\omega: z \to \omega(z)$ be a mapping from $\compC \setminus \{ 1 \}$ to $\compC \setminus \{ 1 \}$ defined by $\omega(z) = 1 + 1/(z - 1)$. Then let $\tildeJ(s)$ be the mapping
\begin{equation}
	\tildeJ(s) = \J(\omega(s)).
\end{equation}
By defining
\begin{equation} \label{eq:defn_R_by_IJ_tilde}
	\tilde{\R}(s) =
	\begin{cases}
		\tilde{R}_1(\tildeJ(s)), & s \in \omega(\compC \setminus \overline{D}) \text{ and } s \notin \omega(\Iinv_1(\Sigma^R)), \\
		\tilde{R}_2(\tildeJ(s)), & s \in \omega(D \setminus [0, 1]) \text{ and } s \notin \omega(\Iinv_2(\Sigma^R)),
	\end{cases}
\end{equation}
where we recall that $D$ is the region bounded by the curves $\gamma_1$ and $\gamma_2$, $\Iinv_1: \compC \setminus [0, b] \to \compC \setminus \overline{D}$ and $\Iinv_2: \paraP \setminus [0, b]$ are defined in \eqref{eq:defn_I1} and \eqref{eq:defn_I2}, respectively.

We are now at the stage of describing the RH problem for $\tilde{\R}$. For this purpose, we define, analogous to \eqref{eq:defn_SigmaR}, and set
\begin{equation} \label{eq:defn_SigmaR_tilde}
	\tilde{\SigmaR} := \tilde{\SigmaR}^{(1)} \cup \tilde{\SigmaR}^{(1')} \cup \tilde{\SigmaR}^{(2)} \cup \tilde{\SigmaR}^{(2')} \cup \tilde{\SigmaR}^{(3)} \cup \tilde{\SigmaR}^{(3')} \cup \tilde{\SigmaR}^{(4)} \cup \tilde{\SigmaR}^{(4')}, \quad \text{where} \quad \tilde{\SigmaR}^{(*)} = \omega(\SigmaR^{(*)}),
\end{equation}
for $* = 1, 1', 2, 2', 3, 3', 4, 4'$ respectively.

We also define the following functions on each curve constituting $\tilde{\SigmaR}$: (Below $z = \tildeJ(s)$)

\begin{align}
	J_{\tilde{\SigmaR}^{(1')}}(s) = {}& (J_{\tilde{Q}})_{21}(z), & s \in {}& \tilde{\SigmaR}^{(1')}, \label{eq:J(s)_1_tilde} \\
	J_{\tilde{\SigmaR}^{(2)}}(s) = {}& (J_{\tilde{Q}})_{12}(z), & s \in {}& \tilde{\SigmaR}^{(2)}, \label{eq:J(s)_2_tilde} \\
	J^1_{\tilde{\SigmaR}^{(3)}}(s) = {}& (\tilde{P}^{(b)})_{22}(z) - 1, \quad J^2_{\tilde{\SigmaR}^{(3)}}(s) = (\tilde{P}^{(b)})_{12}(z), & s \in {}& \tilde{\SigmaR}^{(3)}, \label{eq:J(s)_3_tilde} \\
	J^1_{\tilde{\SigmaR}^{(3')}}(s) = {}& (\tilde{P}^{(b)})_{11}(z) - 1, \quad J^2_{\tilde{\SigmaR}^{(3')}}(s) = (\tilde{P}^{(b)})_{21}(z), & s \in {}& \tilde{\SigmaR}^{(3')}, \label{eq:J(s)_3_2_tilde} \\
	J^1_{\tilde{\SigmaR}^{(4)}}(s) = {}& (\tilde{P}^{(0)})_{22}(z) - 1, \quad J^2_{\tilde{\SigmaR}^{(4)}}(s) = (\tilde{P}^{(0)})_{12}(z), & s \in {}& \tilde{\SigmaR}^{(4)}, \label{eq:J(s)_4_tilde} \\
	J^1_{\tilde{\SigmaR}^{(4')}}(s) = {}& (\tilde{P}^{(0)})_{11}(z) - 1, \quad J^2_{\tilde{\SigmaR}^{(4')}}(s) = (\tilde{P}^{(0)})_{21}(z), & s \in {}& \tilde{\SigmaR}^{(4')}, \label{eq:J(s)_4_2_tilde}
\end{align}
where $z = \tildeJ(s)$ is in $\Sigma^R_1 \cup \Sigma^R_2$ in \eqref{eq:J(s)_1_tilde}; in $(b + \epsilon, +\infty)$ in \eqref{eq:J(s)_2_tilde}; in $\partial D(b, \epsilon)$ in \eqref{eq:J(s)_3_tilde} and \eqref{eq:J(s)_3_2_tilde}; in $\partial D(0, \epsilon)$ in \eqref{eq:J(s)_4_tilde} and \eqref{eq:J(s)_4_2_tilde}. With the aid of these functions, we further define an operator $\Delta_{\tilde{\SigmaR}}$ such that for any complex-valued function $f(s)$ defined on $\tilde{\SigmaR}$, $\Delta_{\tilde{\SigmaR}}$ transforms it linearly into a function $\Delta_{\tilde{\SigmaR}}f$ that is also a complex-valued function defined on $\tilde{\SigmaR}$, with expression
\begin{equation} \label{eq:defn_Delta_SigmaR_tilde}
	(\Delta_{\tilde{\SigmaR}} f)(s) =
	\begin{cases}
		J_{\tilde{\SigmaR}^{(1')}}(s) f(\tilde{s}), & s \in \tilde{\SigmaR}^{(1)} \text{ and } \tilde{s} = \omega(\Iinv_1(\J(s))) \in \tilde{\SigmaR}^{(1')}, \\
		J_{\tilde{\SigmaR}^{(2)}}(s) f(\tilde{s}), & s \in \tilde{\SigmaR}^{(2)} \text{ and } \tilde{s} = \omega(\Iinv_2(\J(s))) \in \tilde{\SigmaR}^{(2')}, \\
		0, & s \in \tilde{\SigmaR}^{(1)} \cup \tilde{\SigmaR}^{(2')}, \\
		J^1_{\tilde{\SigmaR}^{(3)}}(s) f(s) + J^2_{\tilde{\SigmaR}^{(3)}}(s) f(\tilde{s}), & s \in \tilde{\SigmaR}^{(3)} \text{ and } \tilde{s} = \omega(\Iinv_2(\J(s))) \in \tilde{\SigmaR}^{(3')}, \\
		J^1_{\tilde{\SigmaR}^{(3')}}(s) f(s) + J^2_{\tilde{\SigmaR}^{(3')}}(s) f(\tilde{s}), & s \in \tilde{\SigmaR}^{(3')} \text{ and } \tilde{s} = \omega(\Iinv_1(\J(s))) \in \tilde{\SigmaR}^{(3)}, \\
		J^1_{\tilde{\SigmaR}^{(4)}}(s) f(s) + J^2_{\tilde{\SigmaR}^{(4)}}(s) f(\tilde{s}), & s \in \tilde{\SigmaR}^{(4)} \text{ and } \tilde{s} = \omega(\Iinv_2(\J(s))) \in \tilde{\SigmaR}^{(4')}, \\
		J^1_{\tilde{\SigmaR}^{(4')}}(s) f(s) + J^2_{\tilde{\SigmaR}^{(4')}}(s) f(\tilde{s}), & s \in \tilde{\SigmaR}^{(4')} \text{ and } \tilde{s} = \omega(\Iinv_1(\J(s))) \in \tilde{\SigmaR}^{(4)}.
	\end{cases}
\end{equation}
We note that all the functions $J_{\tilde{\SigmaR}^{(1')}}(s)$, \dots, $J^2_{\tilde{\SigmaR}^{(4')}}(s)$ that define $\Delta_{\tilde{\SigmaR}}$ in \eqref{eq:defn_Delta_SigmaR_tilde} are uniformly $\bigO(n^{-1})$. If we view $\Delta_{\tilde{\SigmaR}}$ as an operator from $L^2(\tilde{\SigmaR})$ to $L^2(\tilde{\SigmaR})$, then we have the estimate that for all large enough $n$, there is a constant $M_{\tilde{\SigmaR}}>0$ such that
\begin{equation}
	\lVert \Delta_{\tilde{\SigmaR}} \rVert_{L^2(\tilde{\SigmaR})} \leq M_{\tilde{\SigmaR}} n^{-1}.
\end{equation}

RH problem \ref{RHP:R_tilde} entails a scalar shifted RH problem for $\tilde{\R}$:

\begin{RHP} \hfill \label{RHP:R_scalar_tilde}
	\begin{enumerate}
		\item
		$\tilde{\R}(s)$ is analytic in $\compC \setminus \tilde{\SigmaR}$, where the contour $\tilde{\SigmaR}$ is defined in \eqref{eq:defn_SigmaR_tilde}.
		\item
		For $s \in \tilde{\SigmaR}$, we have
		\begin{equation}
			\tilde{\R}_+(s) - \tilde{\R}_-(s) = (\Delta_{\tilde{\SigmaR}} \tilde{\R}_-)(s),
		\end{equation}
		where $\Delta_{\tilde{\SigmaR}}$ is the operator defined in \eqref{eq:defn_Delta_SigmaR_tilde}.
		\item
		As $s \to \infty$, we have
		\begin{equation}
			\tilde{\R}(s) = 1 + \bigO(s^{-1}).
		\end{equation}
		\item
		As $s \to 0$, we have $\tilde{\R}(s) = \bigO(1)$.
	\end{enumerate}
\end{RHP}

Like Lemma \ref{lem:uniqueness_of_R}, we have

\begin{lem} \label{lem:uniqueness_of_R_tilde}
	The function $\tilde{\R}(s)$ defined in \eqref{eq:defn_R_by_IJ_tilde} is the unique solution of RH problem \ref{RHP:R_scalar_tilde} after trivial analytical extension.
\end{lem}
The proof is the same as that of Lemma \ref{lem:uniqueness_of_R}, and we omit it.

Finally, we have, analogous to Lemma \ref{lem:tRest}, that
\begin{lem}
	For all $s \in \compC \setminus \tilde{\SigmaR}$, we have the uniform convergence
	\begin{equation}\label{eq:esttR_tilde}
		\tilde{\R}(s)=1+\bigO(n^{-1}).
	\end{equation}
\end{lem}
The proof of this lemma is analogous to that of Lemma \ref{lem:tRest}, and we omit it. This lemma immediately yields that
\begin{align} \label{eq:P_1P_2_bounded_tilde}
	\tilde{R}_1(z) = {}& 1 + \bigO(n^{-1}) \quad \text{uniformly in $\paraP \setminus \Sigma^R$}, & \tilde{R}_2(z) = {}& 1 + \bigO(n^{-1}) \quad \text{uniformly in $\compC \setminus \Sigma^R$}.
\end{align}

\section{Proof of Theorems \ref{thm:main_summary} and \ref{thm:main}} \label{sec:proofs_main}
In this section, we will state and prove  Theorem \ref{thm:main}, which is a detailed version of Theorem \ref{thm:main_summary} stated in the introduction.

Since both $p^{(n)}_j(z)$ and $q^{(n)}_j(f(z))$ are real analytic functions and $p^{(n)}_j(\bar{z}) = \overline{p^{(n)}_j(z)}$, $q^{(n)}_j(f(\bar{z})) = \overline{q^{(n)}_j(f(z))}$, to prove their Plancherel-Rotach asymptotics, we only need to give their asymptotics in the upper half-plane and on the real line. To be precise, let $\delta > 0$ be a small enough constant, and let $C_{\delta} = \{ z \in \compC_+ \cup \realR \mid \lvert z \rvert \leq \delta \}$ and $D_{\delta} = \{ z \in \compC_+ \cup \realR \mid \lvert z - b \rvert \leq \delta \}$ be the two semicircles centered at $0$ and $b$ respectively, $B_{\delta} = \{ z \in \compC_+ \cup \realR \mid \Im z \leq \delta/2 \text{ and } \lvert z \rvert > \delta \text{ and } \lvert z - b \rvert > \delta \}$, and $A_{\delta} = (\compC_+ \cup \realR) \setminus (B_{\delta} \cup C_{\delta} \cup D_{\delta} )$. See Figure \ref{fig:four_regions} for the shapes of the regions. We assume that $V(z)$ and $h(z)$ are analytic in $B_{\delta} \cup C_{\delta} \cup D_{\delta}$.

\begin{figure}[htb]
	\centering
	\includegraphics{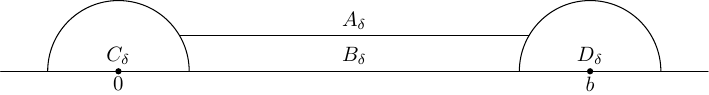}
	\caption[Four regions.]{The four regions in the upper complex plane where the asymptotics of $p^{(n)}_{n + k}(z)$ and $q^{(n)}_{n + k}(f(z))$ are given.}
	\label{fig:four_regions}
\end{figure}

Recall the functions $\gfn(z)$, $\tilde{\gfn}(z)$, $\phi(z)$, $f_b(z)$, $f_0(z)$ defined in \eqref{eq:defn_g_gtilde}, \eqref{eq:defn_phi}, \eqref{eq:f_b_prop}, \eqref{eq:f_0_prop}, respectively. Recall the contours $\gamma_1, \gamma_2$ defined in \eqref{eq:simplified}, $\gamma = \gamma_1 \cup \gamma_2$ with positive orientation, and $h(z)$ is a real analytic function on $[0, \infty)$. Let $\gamma'$ and $\gamma''$ be positively oriented contours such that $\gamma'$ encloses $\gamma$, and $\gamma''$ is enclosed by $\gamma$, such that $h(\J(s))$ is well-defined and analytic in the annular region between $\gamma'$ and $\gamma''$. We also assume that $\{ \Iinv_1(z), \Iinv_2(z) : z \in B_{\delta} \cup C_{\delta} \cup D_{\delta} \}$ lies inside the annular region. Then define
\begin{align}
	D(s) = {}& \exp \left( \frac{1}{2\pi i} \oint_{\gamma''} \log h(\J(\zeta)) \frac{d\zeta}{\zeta - s} \right), && \text{$s$ is outside $\gamma''$}, \label{eq:defn_D(s)} \\
	\tilde{D}(s) = {}& \exp \left( \frac{-1}{2\pi i} \oint_{\gamma'} \log h(\J(\zeta)) \frac{d\zeta}{\zeta - s} \right), && \text{$s$ is inside $\gamma'$}. \label{eq:defn_Dtilde(s)}
\end{align}
Between $\gamma'$ and $\gamma''$, both $D(s)$ and $\tilde{D}(s)$ are defined, and
\begin{equation} \label{eq:prod_DD}
	D(s) \tilde{D}(s) = h(\J(s))^{-1}, \quad \text{$s$ is between $\gamma'$ and $\gamma''$}.
\end{equation}
If $h(z)$ is defined in \eqref{eq:special_W}, then it can be verified that
\begin{align}
	D(s) = {}& \frac{1}{\sqrt{c}} \sqrt{\frac{s - 1}{s - s_1}} \frac{J_c(s)^{1/2}}{s^{1/4}}, & \tilde{D}(s) = {}& \sqrt{c} \sqrt{\frac{s - s_1}{s - 1}} \frac{s^{1/4}}{\sinh(J_c(s))^{1/2}}, &  \tilde{D}(1) = {}& \sqrt{\frac{c (1 - s_1)}{2 e^c}},
\end{align}
such that all power functions take the principal branch.

We define
\begin{equation} \label{eq:G_k_from_D_k}
	G_k(s) = \frac{(\frac{c^2}{4})^{\alpha + \frac{1}{2} + k} (s - s_1)^{\alpha + 1} s^{\frac{1}{2}} (s - 1)^k}{\J(s)^{\alpha + \frac{1}{2}} \sqrt{(s - s_1)(s - s_2)}} D(s), \quad \text{$s$ is outside $\gamma''$ and not on $\gamma_1$ or $[s_1, 1]$},
\end{equation}
where $\sqrt{(s - s_1)(s - s_2)}$ is analytic in $\compC \setminus \gamma_1$ and is $\sim s$ as $s \to \infty$, and $(s - s_1)^{\alpha + 1} s^{\frac{1}{2}}/\J(s)^{\alpha + \frac{1}{2}}$ is analytic in $\compC \setminus [s_1, 1]$, and is $\sim (4/c^2)^{\alpha + 1/2}$ as $s \to \infty$.
Similarly, we also define
\begin{equation} \label{eq:G_k_tilde_from_D_k}
	\tilde{G}_k(s) = \frac{(1 - s_1)^{\alpha + \frac{1}{2}} \sqrt{s_2 - 1} i \tilde{D}(1)^{-1} e^{kc}}{(s - s_1)^{\alpha} (s - 1)^k \sqrt{(s - s_1)(s - s_2)}} \tilde{D}(s), \quad \text{$s$ is inside $\gamma'$ and not on $\gamma_2$ or $(-\infty, s_1]$},
\end{equation}
where $\sqrt{(s - s_1)(s - s_2)}$ is analytic in $\compC \setminus \gamma_2$ and is $\sim s$ as $s \to \infty$, and $(s - s_1)^{\alpha}$ is analytic in $\compC \setminus (-\infty, s_1]$, and takes the principal branch.


Based on $G_k$ and $\tilde{G}_k$, we then define
\begin{align}
	r_k(x) = {}& 2\lvert G_{k, +}(\Iinv_+(x)) \rvert, & \theta_k(x) = {}& \arg(G_{k, +}(\Iinv_+(x))), \\
	\tilde{r}_k(x) = {}& 2\lvert \tilde{G}_{k, +}(\Iinv_-(x)) \rvert, & \tilde{\theta}_k(x) = {}& \arg(\tilde{G}_{k, +}(\Iinv_-(x))),
\end{align}
where $G_{k, +}(\Iinv_+(x))$ is the limit of $G_k(z)$ as $z$ approaches $\Iinv_+(x) \in \gamma_1$ in $\compC \setminus \overline{D}$, and $\tilde{G}_{k, +}(\Iinv_-(x))$ is the limit of $\tilde{G}_k(z)$ as $z$ approaches $\Iinv_-(x) \in \gamma_2$ in $D$.

\begin{thm} \label{thm:main}
	Suppose $V$ satisfies the same condition as in Theorem \ref{thm:main_summary}. As $n \to \infty$, we have the following asymptotics of $p^{(n)}_{n + k}(z)$ and $q^{(n)}_{n + k}(f(z))$, $k \in \intZ$, uniformly for $z$ in regions $A_{\delta}, B_{\delta}, C_{\delta}$ and $D_{\delta}$ (see Figure \ref{fig:four_regions}), if $\delta > 0$ is small enough.
	\begin{enumerate}
		\item
		In region $A_{\delta}$ we have
		\begin{align}
			p^{(n)}_{n + k}(z) = {}& (1 + \bigO(n^{-1})) G_k(\Iinv_1(z)) e^{n\gfn(z)}, && z \in A_{\delta}, \label{eq:p_A} \\
			q^{(n)}_{n + k}(f(z)) = {}& (1 + \bigO(n^{-1})) \tilde{G}_k(\Iinv_2(z)) e^{n\tilde{\gfn}(z)}, && z \in A_{\delta} \cap \paraP. \label{eq:q_A}
		\end{align}
		\item
		In region $B_{\delta}$ we have
		\begin{align}
			p^{(n)}_{n + k}(z) = {}& (1 + \bigO(n^{-1})) G_k(\Iinv_1(z)) e^{n\gfn(z)} + (1 + \bigO(n^{-1})) G_k(\Iinv_2(z)) e^{n(V(z) - \tilde{\gfn}(z) + \ell)}, \label{eq:p_B} \\
			q^{(n)}_{n + k}(f(z)) = {}& (1 + \bigO(n^{-1})) \tilde{G}_k(\Iinv_2(z)) e^{n\tilde{\gfn}(z)} + (1 + \bigO(n^{-1})) \tilde{G}_k(\Iinv_1(z)) e^{n(V(z) - \gfn(z) + \ell)}. \label{eq:q_B}
		\end{align}
		In particular, if $x \in (\delta, b - \delta)$, we have
		\begin{align}
			p^{(n)}_{n + k}(x) = {}& r_k(x) e^{n \int \log \lvert x - y \rvert d\mu(y)} \left[ \cos \left( n\pi \mu([x, b]) + \theta_k(x) \right) + \bigO(n^{-1}) \right], \label{eq:p_bulk_real_line} \\
			q^{(n)}_{n + k}(f(x)) = {}& \tilde{r}_k(x) e^{n \int \log \lvert f(x) - f(y) \rvert d\mu(y)} \left[ \cos \left( n\pi \mu([x, b]) + \tilde{\theta}_k(x) \right) + \bigO(n^{-1}) \right].
		\end{align}
		\item
		In region $D_{\delta}$ we have
		\begin{multline} \label{eq:p^n_n+k_soft}
			e^{-\frac{n}{2}(\gfn(z) - \tilde{\gfn}(z) + V(z) + \ell)} p^{(n)}_{n + k}(z) = \\
			\begin{aligned}
				\sqrt{\pi} & \left[ n^{\frac{1}{6}} f^{\frac{1}{4}}_b(z) \left( (1 + \bigO(n^{-1})) G_k(\Iinv_1(z)) - (1 + \bigO(1))i G_k(\Iinv_2(z)) \right) \Ai(n^{\frac{2}{3}} f_b(z)) \right. \\
				& - \left. n^{-\frac{1}{6}} f^{-\frac{1}{4}}_b(z) \left( (1 + \bigO(n^{-1})) G_k(\Iinv_1(z)) + (1 + \bigO(1))i G_k(\Iinv_2(z)) \right) \Ai'(n^{\frac{2}{3}} f_b(z)) \right], \\
			\end{aligned}
		\end{multline}
		\begin{multline} \label{eq:q^n_n+k_soft}
			e^{-\frac{n}{2}(\tilde{\gfn}(z) - \gfn(z) + V(z) + \ell)} q^{(n)}_{n + k}(f(z)) = \\
			\begin{aligned}
				\sqrt{\pi} & \left[ n^{\frac{1}{6}} f^{\frac{1}{4}}_b(z) \left( (1 + \bigO(n^{-1})) \tilde{G}_k(\Iinv_2(z)) - (1 + \bigO(1))i \tilde{G}_k(\Iinv_1(z)) \right) \Ai(n^{\frac{2}{3}} f_b(z)) \right. \\
				& - \left. n^{-\frac{1}{6}} f^{-\frac{1}{4}}_b(z) \left( (1 + \bigO(n^{-1})) \tilde{G}_k(\Iinv_2(z)) + (1 + \bigO(1))i \tilde{G}_k(\Iinv_1(z)) \right) \Ai'(n^{\frac{2}{3}} f_b(z)) \right].
			\end{aligned}
		\end{multline}
		In particular, if $z = b + f'_b(b)^{-1} n^{-2/3} t$ with $t$ bounded, then
		\begin{multline} \label{eq:asy_Airy_p}
			n^{-\frac{1}{6}} e^{-\frac{n}{2}(\gfn(z) - \tilde{\gfn}(z) + V(z) + \ell)} p^{(n)}_{n + k}(z) = \\
			2^{\frac{1}{4}} \sqrt{\pi} b^{\frac{1}{8}} s^{\frac{3}{8}}_2 c^{\frac{1}{2}} \left( \frac{c^2 (s_2 - s_1)}{4b} \right)^{\alpha + \frac{1}{2}} \left( \frac{c}{2} \right)^k f'_b(b)^{\frac{1}{4}} D(s_2) \left( \Ai(t) + \bigO(n^{-\frac{1}{3}}) \right),
		\end{multline}
		\begin{multline} \label{eq:asy_Airy_q}
			n^{-\frac{1}{6}} e^{-\frac{n}{2}(\tilde{\gfn}(z) - \gfn(z) + V(z) + \ell)} q^{(n)}_{n + k}(f(z)) = \\
			2^{\frac{3}{4}} \sqrt{\pi} \left( \frac{1 - s_1}{s_2 - s_1} \right)^{\alpha + \frac{1}{2}} \left( \frac{c}{2} \right)^k e^{kc} \left( \frac{b}{s_2} \right)^{\frac{1}{8}} f'_b(b)^{\frac{1}{4}} \frac{\tilde{D}(s_2)}{\tilde{D}(1)} \left( \Ai(t) + \bigO(n^{-\frac{1}{3}}) \right).
		\end{multline}
		Here $\Ai$ is the Airy function.
		\item \label{enu:thm:main_C}
		In region $C_{\delta}$ we have
		\begin{multline} \label{eq:p^n_n+k_hard}
			e^{-\frac{n}{2}(\gfn(z) - \tilde{\gfn}(z) + V(z) + \ell)} p^{(n)}_{n + k}(z) = \\
			\begin{aligned}
				\sqrt{\pi} & \left[ n^{\frac{1}{2}} f^{\frac{1}{4}}_0(z) \left( (1 + \bigO(n^{-1})) G_k(\Iinv_1(z)) - (1 + \bigO(1))i e^{\alpha \pi i} G_k(\Iinv_2(z)) \right) I_{\alpha}(2n \sqrt{f_0(z)}) \right. \\
				& + \left. n^{-\frac{1}{2}} f^{-\frac{1}{4}}_0(z) \left( (1 + \bigO(n^{-1})) G_k(\Iinv_1(z)) + (1 + \bigO(1))i e^{\alpha \pi i} G_k(\Iinv_2(z)) \right) I'_{\alpha}(2n \sqrt{f_0(z)}) \right],
			\end{aligned}
		\end{multline}
		\begin{multline} \label{eq:q^n_n+k_hard}
			e^{-\frac{n}{2}(\tilde{\gfn}(z) - \gfn(z) + V(z) + \ell)} q^{(n)}_{n + k}(f(z)) = \\
			\begin{aligned}
				& \sqrt{\pi} \left[ n^{\frac{1}{2}} f^{\frac{1}{4}}_0(z) \left( (1 + \bigO(n^{-1})) \tilde{G}_k(\Iinv_2(z)) - (1 + \bigO(1))i e^{-\alpha \pi i} \tilde{G}_k(\Iinv_1(z)) \right) I_{\alpha}(2n \sqrt{f_0(z)}) \right. \\
				& + \left. n^{-\frac{1}{2}} f^{-\frac{1}{4}}_0(z) \left( (1 + \bigO(n^{-1})) \tilde{G}_k(\Iinv_2(z)) + (1 + \bigO(1))i e^{-\alpha \pi i} \tilde{G}_k(\Iinv_1(z)) \right) I'_{\alpha}(2n \sqrt{f_0(z)}) \right].
			\end{aligned}
		\end{multline}
		In particular, if $z = -f'_0(0)^{-1} n^{-2} t$, with $t$ bounded, then
		\begin{multline} \label{eq:asy_Bessel_p}
			n^{-\frac{1}{2}} e^{-\frac{n}{2}(\gfn(z) - \tilde{\gfn}(z) + V(z) + \ell)} z^{\frac{\alpha}{2}} p^{(n)}_{n + k}(z) = 2\sqrt{\pi} \sqrt{\frac{-s_1}{s_2 - s_1}} \\
			\times \left( \frac{c(1 - s_1) \sqrt{-s_1}}{s_2 - s_1} \right)^{\alpha + \frac{1}{2}} \left( \frac{c^2}{4} (s_1 - 1) \right)^k D(s_1) (-f_0(0))^{\frac{1}{4}} \left( J_{\alpha}(2 \sqrt{t}) + \bigO(n^{-1}) \right),
		\end{multline}

		\begin{multline} \label{eq:asy_Bessel_q}
			n^{-\frac{1}{2}} e^{-\frac{n}{2}(\tilde{\gfn}(z) - \gfn(z) + V(z) + \ell)} z^{\frac{\alpha}{2}} q^{(n)}_{n + k}(f(z)) = 2\sqrt{\pi} \sqrt{\frac{s_2 - 1}{s_2 - s_1}} \\
			\times \left( \frac{c(s_2 - s_1)}{4 \sqrt{-s_1}} \right)^{\alpha + \frac{1}{2}} (s_1 - 1)^{-k} e^{kc} \frac{\tilde{D}(s_1)}{\tilde{D}(1)} (-f_0(z))^{\frac{1}{4}} \left( J_{\alpha}(2 \sqrt{t}) + \bigO(n^{-1}) \right).
		\end{multline}
		Here $J_{\alpha}$ is the Bessel function of the first kind, and $I_{\alpha}$ is the modified Bessel function of the first kind.
		\item
		\begin{equation} \label{eq:asy_h}
			e^{-n\ell} h^{(n)}_{n + k} = \frac{2\pi}{\tilde{D}(1)} \left( \frac{c^2(1 - s_1)}{4} \right)^{\alpha + \frac{1}{2}} \left( \frac{c^2}{4} \right)^{k + \frac{1}{4}} e^{kc} + \bigO(n^{-1}).
		\end{equation}
	\end{enumerate}
\end{thm}
In the rest of this section, we prove Theorem \ref{thm:main}. In the statement of Theorem \ref{thm:main}, $\compC_+ \cup \realR$ is divided into $A_{\delta}$, $B_{\delta}$, $C_{\delta}$ and $D_{\delta}$. We can consider the limit of $p^{(n)}_{n + k}(z)$ and $q^{(n)}_{n + k}(z)$ in the hard edge region $D(0, \epsilon)$, the soft edge region $D(b, \epsilon)$, the bulk region enclosed by $\Sigma^R_1, \Sigma^R_2, \partial D(0, \epsilon), \partial D(b, \epsilon)$ (see Figure \ref{fig:Sigma_R} and equation \eqref{eq:shape_Sigma^R_i}), and the outside region that is the complement of these three, since by deforming the shape of the contour $\Sigma^R$, these four regions cover $C_{\delta}, D_{\delta}, B_{\delta}, A_{\delta}$ respectively.

\subsection{Outside region}

For $z \in \compC$ that is outside the lens and outside $D(0, \epsilon)$ and $D(b, \epsilon)$, we have, by RH problems \ref{RHP:Y}, \ref{RHP:T}, \ref{RHP:S}, \ref{RHP:Q} and \ref{RHP:R},
\begin{equation} \label{eq: outer_region_proof}
	\begin{split}
		p^{(n)}_{n + k}(z) = {}& Y^{(n + k, n)}_1(z) = T_1(z) e^{n\gfn(z)} = S_1(z) e^{n\gfn(z)} = Q_1(z) P^{(\infty)}_1(z) e^{n\gfn(z)} \\
		= {}& R_1(z) P^{(\infty)}_1(z) e^{n\gfn(z)}.
	\end{split}
\end{equation}
Since $R_1(z) = 1 + \bigO(n^{-1})$ uniformly by \eqref{eq:P_1P_2_bounded} and $P^{(\infty)}_1(z) = G_k(\Iinv_1(z))$ by \eqref{eq:Pinfty_1_in_scalar}, we prove \eqref{eq:p_A}.

Similarly, for $z \in \paraP$ that is outside the lens and outside $D(0, \epsilon)$ and $D(b, \epsilon)$, we consider the counterpart of \eqref{eq: outer_region_proof} for $q^{(n)}_{n + k}(f(z))$. For later use in Section \ref{subsec:h^n_n+k}, we also consider $\tilde{C} q^{(n)}_{n + k}(f(z))$ for $z \in \compC$ that is outside the lens and outside $D(0, \epsilon)$ and $D(b, \epsilon)$. We have, by RH problems \ref{RHP:Y_tilde}, \ref{RHP:T_tilde}, \ref{RHP:S_tilde}, \ref{RHP:Q_tilde} and \ref{RHP:R_tilde}, as in \eqref{eq: outer_region_proof},
\begin{align} \label{eq: outer_region_proof_tilde}
	q^{(n)}_{n + k}(f(z)) = {}& \tilde{R}_1(z) \tilde{P}^{(\infty)}_1(z) e^{n\tilde{\gfn}(z)}, & \tilde{C} q^{(n)}_{n + k}(z) = {}& e^{n\ell} \tilde{R}_2(z) \tilde{P}^{(\infty)}_2(z) e^{-n\gfn(z)}.
\end{align}
Since $\tilde{R}_1(z) = 1 + \bigO(n^{-1})$ and $\tilde{R}_2(z) = 1 + \bigO(n^{-1})$ uniformly by \eqref{eq:P_1P_2_bounded_tilde}, $\tilde{P}^{(\infty)}_1(z) = \tilde{G}_k(\Iinv_2(z))$ by \eqref{eq:Pinfty_2_tilde_in_scalar}, and $\tilde{P}^{(\infty)}_2(z)$ has the expression given by \eqref{eq:Pinfty_1_tilde_in_scalar} and \eqref{eq:defn_P_scalar_tilde}, we prove \eqref{eq:q_A} and obtain

\begin{equation} \label{eq:asy_Cq}
	\tilde{C} q^{(n)}_{n + k}(z) = \frac{(1 - s_1)^{\alpha + \frac{1}{2}} \sqrt{s_2 - 1} i \tilde{D}(1)^{-1} e^{kc} z^{\alpha} e^{n\ell} e^{-n\gfn(z)}}{(\Iinv_1(z) - s_1)^{\alpha} (\Iinv_1(z) - 1)^k \sqrt{(\Iinv_1(z) - s_1)(\Iinv_1(z) - s_2)} D(\Iinv_1(z))} (1 + \bigO(n^{-1})).
\end{equation}

\subsection{Bulk region}

Similar to \eqref{eq: outer_region_proof}, for $z$ in the upper lens and outside $D(0, \epsilon)$ and $D(b, \epsilon)$,
\begin{equation} \label{eq:bulk_region_proof}
	\begin{split}
		p^{(n)}_{n + k}(z) = {}& Y^{(n + k, n)}_1(z) = T_1(z) e^{n\gfn(z)} = S_1(z) e^{n\gfn(z)} + z^{-\alpha} h(z)^{-1} f'(z) S_2(z) e^{n(V(z) - \tilde{\gfn}(z) + \ell)} \\
		= {}& Q_1(z) P^{(\infty)}_1(z) e^{n\gfn(z)} + z^{-\alpha} h(z)^{-1} f'(z) Q_2(z) P^{(\infty)}_2(z) e^{n(V(z) - \tilde{\gfn}(z) + \ell)} \\
		= {}& R_1(z) P^{(\infty)}_1(z) e^{n\gfn(z)} + z^{-\alpha} h(z)^{-1} f'(z) R_2(z) P^{(\infty)}_2(z) e^{n(V(z) - \tilde{\gfn}(z) + \ell)}.
	\end{split}
\end{equation}
As in \eqref{eq: outer_region_proof}, $R_1(z) = 1 + \bigO(n^{-1})$ and $R_2(z) = 1 + \bigO(n^{-1})$ uniformly by \eqref{eq:P_1P_2_bounded}. By \eqref{eq:Pinfty_1_in_scalar}, \eqref{eq:Pinfty_2_in_scalar}, \eqref{eq:defn_P_scalar}, \eqref{eq:G_k_from_D_k}, \eqref{eq:G_k_tilde_from_D_k} and \eqref{eq:prod_DD}, we have $P^{(\infty)}_1(z) = G_k(\Iinv_1(z))$ as in the outside region and
\begin{equation} \label{eq:P^infty_in_Gk}
	z^{-\alpha} h(z)^{-1} f'(z) P^{(\infty)}_2(z) = G_k(\Iinv_2(z)).
\end{equation}
Hence, we prove \eqref{eq:p_B}.

In particular, if $x \in (\epsilon, b - \epsilon)$ and $z \to x$ from above, we have by \eqref{eq:Euler-Lagrange} that $\lim_{z \to x} \gfn(z) = \gfn_+(x)$ and $\lim_{z \to x} V(z) - \tilde{\gfn}(z) + \ell = V(x) - \tilde{\gfn}_+(x) + \ell =  \gfn_-(x) = \overline{\gfn_+(x)}$. From the definition \eqref{eq:defn_g_gtilde} of $\gfn(z)$, we have $\gfn_{\pm}(x) = \int \log \lvert x - y \rvert d\mu(y) \pm \mu([x, b]) i$. On the other hand, as $z \to x$ from above, by \eqref{eq:Iinv_+} and \eqref{eq:Iinv_-}, $\Iinv_1(z)$ and $\Iinv_2(z)$ converge to $\Iinv_+(x)$ and $\Iinv_-(x)$ respectively. Noting that $\Iinv_-(x) = \overline{\Iinv_+(x)}$, and then $\lim_{z \to x} G_k(\Iinv_2(z)) = G_{k, +}(\Iinv_-(x)) = \overline{G_{k, +}(\Iinv_+(x))} = \overline{\lim_{z \to x} G_k(\Iinv_1(z))}$ and $\lim_{z \to x} R_2(z) = \R(\Iinv_-(x)) = \overline{\R(\Iinv_+(x))} = \overline{\lim_{z \to x} R_1(z)}$, we have
\begin{equation} \label{eq:p_n+k_on_R}
	\begin{split}
		p^{(n)}_{n + k}(x) = \lim_{z \to x \text{ in } \compC_+} p^{(n)}_{n + k}(z) = {}& 2\Re \left[ \R(\Iinv_+(x)) G_{k, +}(\Iinv_+(x)) e^{n\gfn_+(x)} \right] \\
		= {}& 2\Re \left((1 + \bigO(n^{-1})) G_{k, +}(\Iinv_+(x)) e^{n\gfn_+(x)} \right),
	\end{split}
\end{equation}
which implies \eqref{eq:p_bulk_real_line}.

Analogous to \eqref{eq:bulk_region_proof}, for $z$ in the upper lens and outside $D(0, \epsilon)$ and $D(b, \epsilon)$,
\begin{equation}
	q^{(n)}_{n + k}(f(z)) = \tilde{R}_1(z) \tilde{P}^{(\infty)}_1(z) e^{n\tilde{\gfn}(z)} + z^{-\alpha} h(z)^{-1} \tilde{R}_2(z) \tilde{P}^{(\infty)}_2(z) e^{n(V(z) - \gfn(z) + \ell)}.
\end{equation}
As in \eqref{eq: outer_region_proof_tilde}, $\tilde{R}_1(z) = 1 + \bigO(n^{-1})$ and $\tilde{R}_2(z) = 1 + \bigO(n^{-1})$ uniformly by \eqref{eq:P_1P_2_bounded_tilde}. By \eqref{eq:Pinfty_2_tilde_in_scalar}, \eqref{eq:Pinfty_1_tilde_in_scalar}, \eqref{eq:defn_P_scalar_tilde}, \eqref{eq:G_k_tilde_from_D_k} and \eqref{eq:prod_DD}, we have $\tilde{P}^{(\infty)}_1(z) = \tilde{G}_k(\Iinv_2(z))$ and
\begin{equation} \label{eq:P^infty_in_Gk_tilde}
	z^{-\alpha} h(z)^{-1} \tilde{P}^{(\infty)}_2(z) = \tilde{G}_k(\Iinv_1(z)).
\end{equation}
Hence, we prove \eqref{eq:q_B}.

In particular, if $x \in (\epsilon, b - \epsilon)$ and $z \to x$ from above, we have by \eqref{eq:Euler-Lagrange} that $\lim_{z \to x} \tilde{\gfn}(z) = \tilde{\gfn}_+(x)$ and $\lim_{z \to x} V(z) - \gfn(z) + \ell = V(x) - \gfn_+(x) + \ell =  \tilde{\gfn}_-(x) = \overline{\tilde{\gfn}_+(x)}$. From the definition \eqref{eq:defn_g_gtilde} of $\tilde{\gfn}(z)$, we have $\tilde{\gfn}_{\pm}(x) = \int \log \lvert f(x) - f(y) \rvert d\mu(y) \pm \mu([x, b]) i$. As in \eqref{eq:p_n+k_on_R}, we have
\begin{equation}
	q^{(n)}_{n + k}(f(x)) = \lim_{z \to x \text{ in } \compC_+} q^{(n)}_{n + k}(f(z)) = 2\Re \left((1 + \bigO(n^{-1})) \tilde{G}_{k, +}(\Iinv_+(x)) e^{n\tilde{\gfn}_+(x)} \right),
\end{equation}
which implies \eqref{eq:p_bulk_real_line}.

\subsection{Soft edge region}

As in \eqref{eq: outer_region_proof} and \eqref{eq:bulk_region_proof}, for $z \in D(b, \epsilon) \cap \compC_+$,
\begin{multline} \label{eq:P_in_Q_outside}
	p^{(n)}_{n + k}(z) = Y^{(n + k, n)}_1(z) = T_1(z) e^{n\gfn(z)} = S_1(z) e^{n\gfn(z)} = P^{(\infty)}_1(z) Q_1(z) e^{n\gfn(z)} \\
	+
	\begin{cases}
		0, & \text{outside the lens}, \\
		z^{-\alpha} h(z)^{-1} f'(z) P^{(\infty)}_2(z) Q_2(z) e^{n(V(z) - \tilde{\gfn}(z) + \ell}, & \text{inside the upper lens}.
	\end{cases}
\end{multline}
and then by \eqref{eq:V^b_in_Q_P}, \eqref{eq:P^b_in_E}, \eqref{eq:defn_E^b}, \eqref{eq:defn_Pmodel}, \eqref{eq:defn_g^(b)_i} and \eqref{eq:defn_phi} (also \eqref{eq:Airy_sum_zero} if $z$ is inside the upper lens)
\begin{multline} \label{eq:P_soft_asy}
	p^{(n)}_{n + k}(z) =
	\sqrt{\pi} \left[ n^{\frac{1}{6}} f^{\frac{1}{4}}_b(z) \left( P^{(\infty)}_1(z) V^{(b)}_1(z) - i \frac{f'(z)}{z^{\alpha} h(z)} P^{(\infty)}_2(z) V^{(b)}_2(z) \right) \Ai(n^{\frac{2}{3}} f_b(z)) \right. \\
	\left. - n^{-\frac{1}{6}} f^{-\frac{1}{4}}_b(z) \left( P^{(\infty)}_1(z) V^{(b)}_1(z) + i \frac{f'(z)}{z^{\alpha} h(z)} P^{(\infty)}_2(z) V^{(b)}_2(z) \right) \Ai'(n^{\frac{2}{3}} f_b(z)) \right] e^{\frac{n}{2}(\gfn(z) - \tilde{\gfn}(z) + V(z) + \ell)},
\end{multline}
By \eqref{eq:R_1_defn} and \eqref{eq:R_2_defn}, we have that $V_1(z) = R_1(z)$ and $V_2(z) = R_2(z)$, so they are both uniformly $1 + \bigO(n^{-1})$ by \eqref{eq:P_1P_2_bounded}. Also, we still have \eqref{eq:Pinfty_1_in_scalar} and \eqref{eq:P^infty_in_Gk} for $P^{(\infty)}_1, P^{(\infty)}_2$ as in the outside and bulk regions. Hence, we have \eqref{eq:p^n_n+k_soft}.

For $z \in \compC_+$ in the vicinity of $b$, we have the limits of $\Iinv_1(z)$ and $\Iinv_2(z)$ by \eqref{eq:I_1_at_b} and \eqref{eq:I_2_at_b} respectively,
and for $s$ in the vicinity of $s_2$, we have by \eqref{eq:G_k_from_D_k},
\begin{equation} \label{eq:G_k_soft_edge}
	G_k(s) = \left( \frac{c^2(s_2 - s_1)}{4b} \right)^{\alpha + \frac{1}{2}} \left( \frac{c}{2} \right)^k s^{\frac{1}{2}}_2 D(s_2) (s - s_2)^{-\frac{1}{2}} (1 + \bigO(s - s_2)).
\end{equation}
where $(s - s_2)^{-1/2}$ is positive on $(s_2, +\infty)$ and has the branch along $\gamma_1$. Hence, we derive \eqref{eq:asy_Airy_p}.

As in \eqref{eq:P_in_Q_outside}, for $z \in D(b, \epsilon) \cap \compC_+$,
\begin{multline} \label{eq:P_in_Q_outside_tilde}
	q^{(n)}_{n + k}(f(z)) = \tilde{P}^{(\infty)}_1(z) \tilde{Q}_1(z) e^{n\tilde{\gfn}(z)} \\
	+
	\begin{cases}
		0, & \text{outside the lens}, \\
		z^{-\alpha} h(z)^{-1} \tilde{P}^{(\infty)}_2(z) \tilde{Q}_2(z) e^{n(V(z) - \gfn(z) + \ell}, & \text{inside the upper lens}.
	\end{cases}
\end{multline}
and then by \eqref{eq:V^b_in_Q_P_tilde}, \eqref{eq:P^b_in_E_tilde}, \eqref{eq:defn_E^b_tilde}, \eqref{eq:defn_Pmodel_tilde}, \eqref{eq:defn_g^(b)_i_tilde} and \eqref{eq:defn_phi} (also \eqref{eq:Airy_sum_zero} if $z$ is inside the upper lens)
\begin{multline} \label{eq:P_soft_asy_tilde}
	q^{(n)}_{n + k}(f(z)) =
	\sqrt{\pi} \left[ n^{\frac{1}{6}} f^{\frac{1}{4}}_b(z) \left( \tilde{P}^{(\infty)}_1(z) \tilde{V}^{(b)}_1(z) - i \frac{1}{z^{\alpha} h(z)} \tilde{P}^{(\infty)}_2(z) \tilde{V}^{(b)}_2(z) \right) \Ai(n^{\frac{2}{3}} f_b(z)) \right. \\
	\left. - n^{-\frac{1}{6}} f^{-\frac{1}{4}}_b(z) \left( \tilde{P}^{(\infty)}_1(z) \tilde{V}^{(b)}_1(z) + i \frac{1}{z^{\alpha} h(z)} \tilde{P}^{(\infty)}_2(z) \tilde{V}^{(b)}_2(z) \right) \Ai'(n^{\frac{2}{3}} f_b(z)) \right] e^{\frac{n}{2}(\tilde{\gfn}(z) - \gfn(z) + V(z) + \ell)},
\end{multline}
By \eqref{eq:R_1_defn_tilde} and \eqref{eq:R_2_defn_tilde}, we have that $\tilde{V}_1(z) = \tilde{R}_1(z)$ and $\tilde{V}_2(z) = \tilde{R}_2(z)$, so they are both uniformly $1 + \bigO(n^{-1})$ by \eqref{eq:P_1P_2_bounded_tilde}. Also, we still have \eqref{eq:Pinfty_2_tilde_in_scalar} and \eqref{eq:P^infty_in_Gk_tilde} for $\tilde{P}^{(\infty)}_1, \tilde{P}^{(\infty)}_2$ as in the outside and bulk regions. Hence, we have \eqref{eq:q^n_n+k_soft}.

Similar to \eqref{eq:G_k_soft_edge}, we use the limits of $\Iinv_1(z)$ and $\Iinv_2(z)$, and have for $s$ in the vicinity of $s_2$
\begin{equation}
	\tilde{G}_k(s) = \left( \frac{1 - s_1}{s_2 - s_1} \right)^{\alpha + \frac{1}{2}} \left( \frac{c}{2} \right)^{k - \frac{1}{2}} e^{kc} \frac{\tilde{D}(s_2)}{\tilde{D}(1)} i (s - s_2)^{-\frac{1}{2}} (1 + \bigO(s - s_2)),
\end{equation}
where $(s - s_2)^{-1/2}$ is positive on $(s_2, +\infty)$ and has the branch along $\gamma_2$. Hence, we prove \eqref{eq:asy_Airy_q} in the same way as \eqref{eq:asy_Airy_p}.

\subsection{Hard edge region}

For $z \in D(0, \epsilon)$, \eqref{eq:P_in_Q_outside} still holds, and then by \eqref{eq:V_in_U}, \eqref{eq:P^0_in_E}, \eqref{eq:defn_E^0}, \eqref{eq:defn_Pmodel_0}, \eqref{eq:defn_g^(0)_i}, and \eqref{eq:defn_phi} (also \eqref{eq:Bessel_sum_zero} if $z$ is inside the upper lens)
\begin{multline} \label{eq:p^n_n+k_in_Bessel}
	p^{(n)}_{n + k}(z) = \sqrt{\pi} \left[ n^{\frac{1}{2}} f^{\frac{1}{4}}_0(z) \left( P^{(\infty)}_1(z) V^{(0)}_1(z) + i\frac{f'(z)}{(-z)^{\alpha} h(z)} P^{(\infty)}_2(z) V^{(0)}_2(z) \right) I_{\alpha}(2n \sqrt{f_0(z)}) \right. \\
	+ \left. n^{-\frac{1}{2}} f^{-\frac{1}{4}}_0(z) \left( P^{(\infty)}_1(z) V^{(0)}_1(z) - i\frac{f'(z)}{(-z)^{\alpha} h(z)} P^{(\infty)}_2(z) V^{(0)}_2(z) \right) I'_{\alpha}(2n \sqrt{f_0(z)}) \right] \\
	\times e^{\frac{n}{2}(\gfn(z) - \tilde{\gfn}(z) + V(z) + \ell)},
\end{multline}
where the $(-z)^{\alpha}$ factor takes the principal branch as $\arg (-z) \in (-\pi, \pi)$. As in the soft edge region, $V^{(0)}_1(z)$ and $V^{(0)}_2(z)$ are both uniformly $1 + \bigO(n^{-1})$ by \eqref{eq:P_1P_2_bounded}, and $P^{(\infty)}_1(z)$ and $P^{(\infty)}_2(z)$ are still given by \eqref{eq:Pinfty_1_in_scalar}, \eqref{eq:Pinfty_2_in_scalar}, \eqref{eq:defn_P_scalar}, \eqref{eq:G_k_from_D_k} and \eqref{eq:prod_DD}. Then we have \eqref{eq:p^n_n+k_hard}.

For $z \in \compC_+$ in the vicinity of $0$, we have the limits of $\Iinv_1(z)$ and $\Iinv_2(z)$ by \eqref{eq:I_1_at_0} and \eqref{eq:I_2_at_0} respectively,
and for $s$ in the vicinity of $s_1$, we have by \eqref{eq:G_k_from_D_k},
\begin{equation} \label{eq:G_k_hard_edge}
	G_k(s) = \left( \frac{4(1 - s_1)^2 (-s_1)}{(s_2 - s_1)^2} \right)^{\alpha + \frac{1}{2}} \left( \frac{c^2}{4} (s_1 - 1) \right)^k \sqrt{\frac{-s_1}{s_2 - s_1}} D(s_1) (s_1 - s)^{-\alpha - \frac{1}{2}} (1 + \bigO(s - s_1)),
\end{equation}
where $(s_1 - s)^{-\alpha - 1/2}$ is positive on $(-\infty, s_1)$ and has the branch cut along $\gamma_1$. Hence, we derive, if $z \in D(0, \epsilon) \cap \compC_+$ and $z = f'_0(0)^{-1} n^{-2} t$, with $t$ bounded, then uniformly
\begin{multline}
	n^{-\frac{1}{2}} e^{-\frac{n}{2}(\gfn(z) - \tilde{\gfn}(z) + V(z) + \ell)} (-z)^{\frac{\alpha}{2}} p^{(n)}_{n + k}(z) = \\
	2\sqrt{\pi} \sqrt{\frac{-s_1}{s_2 - s_1}} \left( \frac{c(1 - s_1) \sqrt{-s_1}}{s_2 - s_1} \right)^{\alpha + \frac{1}{2}} \left( \frac{c^2}{4} (s_1 - 1) \right)^k D(s_1) (-f_0(0))^{\frac{1}{4}} \left( I_{\alpha}(2 \sqrt{t}) + \bigO(n^{-1}) \right),
\end{multline}
which implies \eqref{eq:asy_Bessel_p} on $D(0, \epsilon) \cap (\compC_+ \cup \realR)$ by changing $I_{\alpha}$ into $J_{\alpha}$ as in \cite[10.27.6]{Boisvert-Clark-Lozier-Olver10}.

As in \eqref{eq:p^n_n+k_in_Bessel}, for $z \in D(0, \epsilon) \cap \compC_+$,
\begin{multline} \label{eq:P_hard_asy_tilde}
	q^{(n)}_{n + k}(f(z)) =
	\sqrt{\pi} \left[ n^{\frac{1}{2}} f^{\frac{1}{4}}_0(z) \left( \tilde{P}^{(\infty)}_1(z) \tilde{V}^{(0)}_1(z) + i \frac{1}{(-z)^{\alpha} h(z)} \tilde{P}^{(\infty)}_2(z) \tilde{V}^{(0)}_2(z) \right) I_{\alpha}(2n \sqrt{f_0(z)}) \right. \\
	\left. + n^{-\frac{1}{2}} f^{-\frac{1}{4}}_0(z) \left( \tilde{P}^{(\infty)}_1(z) \tilde{V}^{(0)}_1(z) - i \frac{1}{(-z)^{\alpha} h(z)} \tilde{P}^{(\infty)}_2(z) \tilde{V}^{(0)}_2(z) \right) I'_{\alpha}(2n \sqrt{f_0(z)}) \right] \\
	\times e^{\frac{n}{2}(\tilde{\gfn}(z) - \gfn(z) + V(z) + \ell)}.
\end{multline}
By \eqref{eq:R_1_defn_tilde} and \eqref{eq:R_2_defn_tilde}, we have that $\tilde{V}_1(z) = \tilde{R}_1(z)$ and $\tilde{V}_2(z) = \tilde{R}_2(z)$, so they are both uniformly $1 + \bigO(n^{-1})$ by \eqref{eq:P_1P_2_bounded_tilde}. Also, we still have \eqref{eq:P^infty_in_Gk_tilde} for $\tilde{P}^{(\infty)}_1, \tilde{P}^{(\infty)}_2$ as in the bulk region and the soft edge region. Hence, we have \eqref{eq:q^n_n+k_hard}.

Similar to \eqref{eq:G_k_hard_edge}, we also have for $s$ in the vicinity of $s_1$
\begin{equation}
	\tilde{G}_k(s) = (1 - s_1)^{\alpha + \frac{1}{2}} (s_1 - 1)^{-k} e^{kc} \sqrt{\frac{s_2 - 1}{s_2 - s_1}} \frac{\tilde{D}(s_1)}{\tilde{D}(1)} (s - s_1)^{-\alpha - \frac{1}{2}} (1 + \bigO(s - s_1)).
\end{equation}
Hence, we prove \eqref{eq:asy_Bessel_q} in the same way as \eqref{eq:asy_Bessel_p}.

\subsection{Computation of $h^{(n)}_{n + k}$} \label{subsec:h^n_n+k}
Recall that
$h^{(n)}_{n + k}$ is defined in \eqref{eq:defn_h}. By the limiting formulas \eqref{eq:p_A} for $p^{(n)}_{n + k}(z)$ and \eqref{eq:q_A} for $q^{(n)}_{n + k}(f(z))$, and the regularity condition \eqref{eq:gfn_outer} as well as the growth condition \eqref{eq:growth_at_infty}, there exists $\epsilon,\delta > 0$ such that for all $R > b + \delta$,
\begin{equation}
	\left\lvert h^{(n)}_{n + k} - h^{(n)}_{n + k}(R) \right\rvert = \bigO( e^{n \ell-\epsilon n}), \,\, \text{where} \,\, h^{(n)}_{n + k}(R) = \int^R_0 p^{(n)}_{n + k}(x) q^{(n)}_{n + k}(f(x)) W^{(n)}_{\alpha}(x) dx.
\end{equation}
Set $C_R$ to be the circular contour with positive orientation, centered at $0$ with radius $R$, and let $\hat{C}_R$ be its image under $\Iinv_1$, i.e.,
\begin{equation}\label{eq:defcrcr'}
	C_R: = \{ w \in \mathbb{C} : \lvert w \rvert = R \},\quad \hat{C}_R=\Iinv_1(C_R).
\end{equation}
If $R$ is large enough, then $\hat{C}_R$ is uniquely determined.
To prove \eqref{eq:asy_h}, we only need to compute
\begin{equation}
	h^{(n)}_{n + k}(R) = -\oint_{C_R} p^{(n)}_{n + k}(z) \tilde{C} q^{(n)}_{n + k}(z) dz,
\end{equation}
where $\tilde{C}q^{(n)}_{n+k}$ is the Cauchy transform of $q^{(n)}_{n+k}$ as defined in \eqref{eq:cauchyqj}.

With the help of \eqref{eq:p_A} and \eqref{eq:asy_Cq}, we have
\begin{equation}\label{int-1}
	\begin{split}
		h^{(n)}_{n + k}(R) = {}& e^{n\ell} \left( -\oint_{C_R} \frac{G_k(\Iinv_1(z)) (1 - s_1)^{\alpha + \frac{1}{2}} \sqrt{s_2 - 1} i \tilde{D}(1)^{-1} e^{kc} z^{\alpha} D(\Iinv_1(z))^{-1}}{(\Iinv_1(z) - s_1)^{\alpha} (\Iinv_1(z) - 1)^k \sqrt{(\Iinv_1(z) - s_1)(\Iinv_1(z) - s_2)}} dz + \bigO(n^{-1}) \right) \\
		= {}& e^{n\ell} \left( -\frac{\sqrt{s_2 - 1} i}{\tilde{D}(1)} \left( \frac{c^2(1 - s_1)}{4} \right)^{\alpha + \frac{1}{2}} \left( \frac{c^2}{4} \right)^k e^{kc} \oint_{C_R} \frac{\Iinv_1(z)^{1/2}}{z^{1/2} (\Iinv_1(z) - s_2)} dz + \bigO(n^{-1}) \right).
	\end{split}
\end{equation}
To compute the integral, we make a change of variables
$w = \Iinv_1(z)$. Then
$z^{1/2}=J_c(w)/2$, and
\begin{equation}\label{int0}
	\oint_{C_R} \frac{\Iinv_1(z)^{\frac{1}{2}}  }{z^{\frac{1}{2}} (\Iinv_1(z) - s_2)} dz = 2\oint_{\hat{C}_R} \frac{w^{\frac{1}{2}}}{J_c(w)(w - s_2) }  d \J(w).
\end{equation}
Using
\begin{equation}\label{eq:JJdev}
	\frac{d \J(w)}{dw} = \frac{J_c(w)}{2} \frac{d J_c(w)}{dw}
\end{equation}
and
\begin{equation}\label{eq:Jdev}
	\frac{d J_c(w)}{dw} = \frac{d}{dw} \left( c\sqrt{w} + \arcosh \frac{w + 1}{w - 1} \right) = \frac{c}{2} \frac{w - s_2}{\sqrt{w} (w - 1)},
\end{equation}
we see that
\begin{equation}\label{int1}
	\oint_{C_R} \frac{\Iinv_1(z)^{\frac{1}{2}}  }{z^{\frac{1}{2}} (\Iinv_1(z) - s_2)} dz
	=\frac{c}{2}\oint_{\hat{C}_R} (w-1)^{-1}dw.
\end{equation}
Substituting \eqref{int1} into the second line of \eqref{int-1}, we get
\begin{equation}
	\begin{split}
		e^{-n\ell}	h^{(n)}_{n + k}(R) &=   -\frac{\sqrt{s_2 - 1} c\pi }{\tilde{D}(1)} \left( \frac{c^2(1 - s_1)}{4} \right)^{\alpha + \frac{1}{2}} \left( \frac{c^2}{4} \right)^k e^{kc}  + \bigO(n^{-1}) \\
		&=\frac{2\pi}{\tilde{D}(1)} \left( \frac{c^2(1 - s_1)}{4} \right)^{\alpha + \frac{1}{2}} \left( \frac{c^2}{4} \right)^{k + \frac{1}{4}} e^{kc} + \bigO(n^{-1}) ,
	\end{split}
\end{equation}
where we have used the fact that $s_2=1+2/c$. This completes the proof of Theorem \ref{thm:main}.


\section{Bulk universality and limiting density}
\label{sec:bulk universality}
In this section, we first give pointwise and integral bounds for $p_{j}^{(n)}$ and $q_{j}^{(n)}$. Then we combine these bounds, Theorem \ref{thm:main}, and Proposition \ref{prop:DPP} to prove Theorem \ref{thm:sinekernel}. Throughout this section, we let
\begin{equation}\label{eq:uvdef}
	u = {} x^* + \frac{\xi}{\pi \psi(x^*) n}, \quad  v = {} x^* + \frac{\eta}{\pi \psi(x^*) n},
\end{equation}
where $x^* \in (0, b)$, and $\xi, \eta$ are in a compact $\Xi$ subset of $\compC$ ($x^*$ and $\Xi$ are both fixed). 

\subsection{Upper bounds on $p_{j}^{(n)}$ and $q^{(n)}_{j}$}\label{subsec:upperbds}
Recall that in Theorem \ref{thm:main}, we have asymptotic formulas for $p^{(n)}_j(x)$, $q^{(n)}_j(f(y))$, and $h^{(n)}_j$ if $j = n + \bigO(1)$. By a scaling argument, we also have the asymptotics for $p^{(n)}_j(x)$, $q^{(n)}_j(f(y))$, and $h^{(n)}_j$ if $j$ is of order $n$ (to be explained in detail shortly). Our asymptotic analysis based on Riemann-Hilbert problems is not convenient for estimating $p_j^{(n)}$ and $q^{(n)}_j$ for $j=\mathfrak{o}(n)$. However, a very crude estimate for small $j$ suffices in the proof of the bulk universality.

As mentioned below Proposition \ref{prop:char},
$p_j^{(n)}$ can be analyzed by transforming the function $V$ to $V_t(x) := t^{-1} V(x)$, where $t\in (0,\infty)$ is given by $n/j$ (both cases of $j\leq n$ and $j\geq n$ will be needed so that $t$ can be any positive real number).
For each $t \in (0, +\infty)$, we also let $\mu_t := \mu_{V, t}$ be the minimizer of $I_{V_t}$, that is, the functional defined by the right-hand side of \eqref{eq:energy_functional} with $V$ replaced by $V_t$. If $V$ satisfies assumption \eqref{eq:one-cut_regular_cond}, then all $V_t$ do so as well, and the argument in Section \ref{sec:constr_eq_measure} yields that $\mu_t$ is supported on an interval $[0, b_t]$ with a density function $\psi_t(x)$ that satisfies the one-cut regular with a hard edge condition stated in Requirement \ref{req:one-cut_reg}.
Then we can define the constant $\ell_t$ by \eqref{eq:defn_ell} with $V$ replaced by $V_t$, and functions $\gfn_t(z), \gfntilde_t(z)$ by \eqref{eq:defn_g_gtilde}, with $\psi(x)$ replaced by $\psi_t(x)$.
Theorem \ref{thm:main} (with $k$ set to 0,
$\gfntilde, \gfn$ replaced by $\gfntilde_t, \gfn_t$, and $n$ replaced by $tn$) now gives asymptotics of $p^{(n)}_j$.
(This process works well when $j$ is of order $n$, or equivalently, when $t$ is away from $0$ or $\infty$.)

For $z\in \mathbb{P}$ (the region shown in Figure \ref{fig:paraP}), we set
\begin{equation} \label{eq:Lambdat}
  \Lambda_t(z) =  \frac{t}{2} \left( \gfntilde_t(z) - \gfn_t(z) \right)=\frac{t}{2} \int_0^{b_t} \log \left( \frac{f(z)-f(y)}{z-y}\right)d\mu_t(y).
\end{equation}
Although $\gfn_t(z)$ and $\gfntilde_t(z)$ are discontinuous along the real axis, $\Lambda_t(z)$ is well defined there by analytic continuation. It is also clear that $\Lambda_1$ coincides with $\Lambda$ defined in \eqref{deflambda}.

We collect some properties of $\Lambda_t$ and $b_t$ in the following lemma.
\begin{lem}\label{lem:lip}
	Suppose that $s, t$ are in a compact subset of $(0, +\infty)$ and $u, v$ are given by \eqref{eq:uvdef}. There is a constant $C_0 > 0$ (depending on the compact subset), such that
	\begin{align} \label{eq:est_Ft}
		\lvert \Lambda_t(u) - \Lambda_t(v) \rvert < {}& C_0 \lvert u - v \rvert, & \lvert \Lambda_s(u) - \Lambda_t(u) \rvert < {}& C_0 \lvert s - t \rvert.
	\end{align}
	As $t \to 0_+$, we have the convergence $b_t \to 0_+$.
\end{lem}

\begin{proof}[Proof of Lemma \ref{lem:lip}]
	The definition of $u,v$ in \eqref{eq:uvdef} implies that they are away from $0$ and $b$. The Lipschitz continuity of $\Lambda_t(z)$ for $z$ in this regime then follows from
	 Lemma \ref{enu:lem:regular_g:1}. On the other hand,
	the Lipschitz continuity of $\Lambda_t(z)$ in $t$ stems from the continuity of $b_t$, which is a consequence of Theorem \ref{eq:find_b_psi} (and its proof in Lemma \ref{lem:determine_b}) through the implicit function theorem.

We now prove the second statement that $b_t\to 0$ as $t\to 0$.
Assume the contrary, i.e., there exists a sequence $t(n)\to 0$ such that $b_{t(n)}\geq \kappa>0$ for $n\geq 1$.
Applying \eqref{eq:defn_ell} to $x=\kappa/2$ and $x=\kappa$, we see that
\begin{equation}
	\begin{split}
&			\int \log|y-\kappa/2|^{-1}d\mu_{t(n)}(y)+\int \log|f(y)-f(\kappa/2)|^{-1}d\mu_{t(n)}(y)+ \frac{1}{t_n}V\left(\frac{\kappa}{2}\right) \\
=&
		\int \log|y-\kappa|^{-1}d\mu_{t(n)}(y)+\int \log|f(y)-f(\kappa)|^{-1}d\mu_{t(n)}(y)+ \frac{1}{t_n}V\left(\kappa \right).
	\end{split}
\end{equation}
Re-arranging this equality, we get
\begin{equation}\label{eq:vkappa}
	\begin{split}
			\frac{1}{t(n)} \left(V(\kappa)-V\left(\frac{\kappa}{2}\right)\right)=
 		\int  \left(\log \left| \frac{y-\kappa}{y-\kappa/2}\right|
 		+\log \left| \frac{f(y)-f(\kappa)}{f(y)-f(\kappa/2) } \right|\right)
 		d\mu_{t(n)}(y).
\end{split}
\end{equation}
By the condition \eqref{eq:one-cut_regular_cond} on $V$, we find that $V(\kappa)-V(\kappa/2)>0$. Thus, the left-hand side of \eqref{eq:vkappa} tends to $\infty$ as $t(n)\to 0$, while the right-hand side remains bounded. We arrive at a contradiction and deduce that $b_t\to 0 $ as $t\to 0$.
\end{proof}

	As mentioned at the beginning of this section, we will bound $p_{j}^{(n)}$ and $q_j^{(n)}$ according to
	the relative size of $j$ compared to $n$.
Fix a small $\epsilon \in (0, 1)$. The asymptotics result in Theorem \ref{thm:main} implies that all zeros of $p^{(n)}_{\lfloor \epsilon n \rfloor}(x)$ and $q^{(n)}_{\lfloor \epsilon n \rfloor}(f(x))$ are in the interval $[0, 1.1b_{\epsilon}]$, if $n$ is large enough. We note that functions $\{f(x)^i\}_{i\geq 0}$ form a Markov system\footnote{In other words, for any $N\geq 1$, any real linear combination $\sum_{i=0}^N a_i f(x)^i$ has at most $N$ zeros for $x\in \mathbb{R}_+$. This follows from the facts that any polynomial of degree $n$ has at most $n$ real roots and that $f$ is strictly monotone on $\mathbb{R}_+$. } on $\mathbb{R}_+$ in the sense of \cite[second definition in Section 4.4]{Nikishin-Sorokin91}. The result in \cite{Kershaw70} shows that the zeros of $p^{(n)}_k(x)$ and $p^{(n)}_{k - 1}(x)$ interlace. Similarly, the zeros of $q^{(n)}_k(f(x))$ and $q^{(n)}_{k - 1}(f(x))$ interlace. Hence, for all $j = 1, 2, \dotsc, \lfloor \epsilon n \rfloor$, the zeros of $p^{(n)}_j(x)$ and $q^{(n)}_j(f(x))$ are all in $[0, 1.1b_{\epsilon}]$.

To state the bound for $p_j^{(n)}$, it is more convenient to perform a transformation
\begin{align}
  \p_j(x) = {}& e^{-\frac{n \ell_{j/n}}{2}} \sqrt{W^{(n)}_{\alpha}(x)} p^{(n)}_j(x), & \q_j(x) = {}& \frac{1}{h^{(n)}_j} e^{\frac{n \ell_{j/n}}{2}} \sqrt{W^{(n)}_{\alpha}(x)} q^{(n)}_j(f(x)).
\end{align}

We have the following estimates that are analogous to \cite[Corollary 2.7]{Claeys-Wang22}. We emphasize that the same constant $C_1$ appears in both \eqref{eq:L2_p_and_q_outside} and \eqref{eq:p_k_q_k_estimate}.
We also remind the reader that the $u,v$ appearing in \eqref{eq:p_k_q_k_estimate} and \eqref{eq:p_kq_kpointbd} are given by \eqref{eq:uvdef} (so that \eqref{eq:p_k_q_k_estimate} and \eqref{eq:p_kq_kpointbd} hold uniformly over all $\xi,\eta\in \Xi$).
\begin{lem} \label{lem:est_p_q}
  Let $n$ be large enough and $\epsilon \in (0, 1)$ a constant. There exist $C_1, C_2, C_3, C_4 > 0$ such that the following inequalities \eqref{eq:L2_p_and_q_outside} - \eqref{eq:p_kq_kpointbd} hold true. For all $k = 0, 1, \dotsc, \lfloor 1.1n \rfloor$, we have
  \begin{align}
&    \int^{\infty}_{C_2} \lvert \p_k(x) \rvert^2 \cosh(2\sqrt{x}) dx < {} e^{-C_1 n}, \quad \int^{\infty}_{C_2} \lvert \q_k(x) \rvert^2 \cosh(2\sqrt{x}) dx < {} e^{-C_1 n}, \label{eq:L2_p_and_q_outside} \\
&\int^{C_2}_0 \lvert \p_k(x) \rvert^2 dx < {} e^{C_3 n}, \quad \  \int^{C_2}_0 \lvert \q_k(x) \rvert^2 dx < {} e^{C_3 n}, \label{eq:L2_p_and_q_inside} \\
& \lvert e^{n\Lambda(u)} \p_k(u) \rvert < {} e^{\frac{1}{3} C_1 n}, \quad \ \lvert e^{-n\Lambda(v)} \q_k(v) \rvert < {} e^{\frac{1}{3} C_1 n}.\label{eq:p_k_q_k_estimate}
  \end{align}
  Moreover, for all $ \lfloor \epsilon n \rfloor \leq k \leq \lfloor 1.1n \rfloor$, we have
  \begin{equation}
    \int^{C_2}_0 \lvert e^{n \Lambda_{k/n}(x)} \p_k(x) \rvert^2 dx < {} C_4, \quad \quad  \int^{C_2}_0 \lvert e^{-n \Lambda_{k/n}(x)} \q_k(x) \rvert^2 dx < {} C_4, \label{eq:p_k_q_k_integral_estimate}
  \end{equation}
and
\begin{equation}
	| e^{n \Lambda_{k/n}(u)} \p_k(u)| < C_4, \quad\quad | e^{-n \Lambda_{k/n}(v)} \q_k(v) |< C_4.
	\label{eq:p_kq_kpointbd}
\end{equation}
\end{lem}

\begin{proof}

  For $\lfloor \epsilon n \rfloor \leq k \leq \lfloor 1.1n \rfloor$, by substituting $k$ for $n$ and $V_{k/n}$ for $V$ in Theorem \ref{thm:main}, we find that $	| e^{n \Lambda_{k/n}(x)} \p_k(x)|$ and
  $ | e^{-n \Lambda_{k/n}(x)} \q_k(x) |$ are  upper bounded both pointwise (in the bulk) and in the $L^2$ norm.
  Combining this observation with the growth condition \eqref{eq:growth_at_infty} on $V$ and the continuity of $\Lambda_t$ in \eqref{eq:est_Ft}, we can first choose a large enough $C_1$ such that \eqref{eq:p_k_q_k_estimate} holds. Then we let $C_2$ be large enough to guarantee
  \eqref{eq:L2_p_and_q_outside}. Finally, we choose large enough $C_3$ and $C_4$ to validate
   \eqref{eq:L2_p_and_q_inside}, \eqref{eq:p_k_q_k_integral_estimate}, and \eqref{eq:p_kq_kpointbd}.

  For $0 \leq k \leq \lfloor \epsilon n \rfloor$, we can repeat the argument in \cite[Proof of Lemma 3.1]{Claeys-Wang22} (and modify the constants if necessary) to obtain the estimates \eqref{eq:L2_p_and_q_outside}, \eqref{eq:L2_p_and_q_inside}, \eqref{eq:p_k_q_k_estimate}. They essentially follow from the simple facts that $p_k(z)$ (resp.~$q_k(f(z))$) is a monic polynomial in $z$ (resp.~$f(z)$) of degree $k$, with all zeros on $[0, 1.1 b_{\epsilon}]$, and that $V$ attains its global minimum at $0$.

\end{proof}

Using the estimate \eqref{eq:p_k_q_k_estimate}, we have for all $j, k = 0, 1, \dotsc, \lfloor 1.1n \rfloor$
\begin{equation} \label{eq:est_pq_product}
  \lvert e^{n(\Lambda(u) - \Lambda(v))} \p_j(u) \q_k(v) \rvert = \lvert e^{n\Lambda(u)} \p_j(u) \rvert \lvert e^{-n\Lambda(v)} \q_k(v) \rvert <  e^{\frac{2}{3} C_1 n}.
\end{equation}

In addition, by \eqref{eq:est_Ft} and \eqref{eq:p_kq_kpointbd}, for $\epsilon n\leq j\leq k\leq 1.1n$, we get
\begin{multline} \label{eq:est_Pk_qj_around_n}
  \lvert e^{n(\Lambda(u) -\Lambda(v))} \p_k(u) \q_j(v) \rvert = \lvert e^{n(\Lambda_1(u) - \Lambda_{k/n}(u))} \rvert \lvert e^{-n(\Lambda_1(v) - \Lambda_{j/n}(v))} \rvert \\
  \times \lvert e^{n \Lambda_{k/n}(u)} \p_k(u) e^{-n \Lambda_{j/n}(v)} \q_j(v) \rvert \leq e^{C_0(k - j)} C'.
\end{multline}

\subsection{Sine limit of $K_n$ in the bulk}
We now use the bounds in the previous section and Theorem \ref{thm:main} to prove \eqref{eq:_sine_kernel_pre}, which would also complete the proof of Theorem \ref{thm:sinekernel}.
Recall the definition of $K_n(u, v)$ in \eqref{eq:corr_kernel}.
\subsubsection{An equivalent form of \eqref{eq:_sine_kernel_pre}}
Like the truncated exponential sum $\exp_k(x) = \sum^k_{i = 0} x^i/i!$ used in \cite[Equation (4.3)]{Claeys-Wang22},
we consider a truncated version $f_k$ of $f(x)=\sinh^2(\sqrt{x})$,
\begin{equation} \label{eq:f_exp_relation}
	f_k(x) = \sum^k_{i = 0} c_i x^i = \frac{1}{4}(\exp_{2k}(2\sqrt{x}) + \exp_{2k}(-2\sqrt{x})) - \frac{1}{2},
\end{equation}
where
\begin{equation}
\quad \sinh^2(\sqrt{x}) = \sum^{\infty}_{i = 0} c_i x^i,
\end{equation}
is the Taylor expansion of $f(x)$ at $0$. One can show that for $k, l \in \intZ_+$ and $z = (x + iy)^2$,
\begin{align} \label{eq:simple_est}
  \lvert f_k(z) \rvert \leq {}& \cosh(2x), & \lvert f(z) - f_k(z) \rvert \leq {}& \cosh(2x), & \lvert f_k(z) - f_l(z) \rvert \leq \cosh(2x),
\end{align}
and with $C_2$ defined in Lemma \ref{lem:est_p_q}, there exists a constant $C_5$, depending on $C_2$, such that for all $k < l \in \intZ_+$,
\begin{align} \label{eq:not_simple_est}
  \max_{\lvert z \rvert < C_2} \lvert f(z) - f_k(z) \rvert < {}& C_5 e^{-\frac{k}{2} \log(k)}, &  \max_{\lvert z \rvert < C_2} \lvert f_l(z) - f_k(z) \rvert < {}& C_5 e^{-\frac{k}{2} \log(k)}.
\end{align}
Equations \eqref{eq:simple_est} and \eqref{eq:not_simple_est} can be proved directly by Stirling's formula, or by the estimates \cite[Equations (4.4) and (4.5)]{Claeys-Wang22} of $\exp_k(z)$ together with the relation \eqref{eq:f_exp_relation}.

Let $\delta$ be a fixed small positive number. From estimates \eqref{eq:not_simple_est} for $f_{\lfloor \delta n \rfloor}(u)$, \eqref{eq:est_Ft} for $e^{n\Lambda(u)} e^{-n\Lambda(v)}$, and \eqref{eq:p_k_q_k_estimate} for $\p_j(u) \q_j(v)$, we find that \eqref{eq:_sine_kernel_pre} is equivalent to
\begin{equation} \label{eq:essential_sine_kernel}
  \lim_{n \to \infty} (f_{\lfloor \delta n \rfloor}(u) - f(v)) e^{n\Lambda(u)} e^{-n\Lambda(v)} \sum^{n - 1}_{k = 0} \p_k(u) \q_k(v) = \frac{f'(x^*)}{\pi} \sin (\xi - \eta).
\end{equation}
It remains to prove \eqref{eq:essential_sine_kernel}.
We define
\begin{align}
	\tilde{a}_{j, k} = {}& \int_0^{\infty} \p_k(y) \q_j(y) f(y) dy, & \tilde{b}_{j, k} = {}& \int_0^{\infty} \p_j(x) \q_k(x) f_{\lfloor \delta n \rfloor}(x) dx.
\end{align}
Since 	$f(y) \q_j(y)$ is a weighted polynomial of degree $j+1$ in $f(y)$ and  $f_{\lfloor \delta n \rfloor}(x) \p_j(x) $ is a weighted polynomial of degree $j+\lfloor \delta n \rfloor$ in $x$, we  can expand
\begin{align}
	f(y) \q_j(y) = {}& \sum^{j + 1}_{k = 0} \tilde{a}_{j, k} \q_k(y), & f_{\lfloor \delta n \rfloor}(x) \p_j(x) = \sum^{j + \lfloor \delta n \rfloor}_{k = 0} \tilde{b}_{j, k} \p_k(x).
\end{align}
Thus we see that $\tilde{a}_{j, k} = 0$ if $k > j + 1$ and $\tilde{b}_{j, k} = 0$ if $k > j + \lfloor \delta n \rfloor$. 
Then we express, with $M$ any positive integer,
\begin{equation}\label{eq:ess2}
	(f_{\lfloor \delta n \rfloor}(u) - f(v)) \sum^{n - 1}_{k = 0} \p_k(u) \q_k(v) = J_1(u, v) + J^{(M)}_2(u, v) + J^{(M)}_3(u, v),
\end{equation}
where
\begin{align}
	J_1(u, v) = {}& \sum^{n - 1}_{j, k = 0} (\tilde{b}_{k, j} - \tilde{a}_{j, k}) \p_j(u) \q_k(v), \\
	J^{(M)}_2(u, v) = {}& \sum^{n - 1 - M}_{j = n - \lfloor \delta n \rfloor} \sum^{j + \lfloor \delta n \rfloor}_{k = n} \tilde{b}_{j, k} \p_k(u) \q_j(v), \\
	J^{(M)}_3(u, v) = {}& \sum^{n - 1}_{j = n - M} \sum^{j + M}_{k = n} \tilde{b}_{j, k} \p_k(u) \q_j(v) - \tilde{a}_{n - 1, n} \p_{n - 1}(u) \q_n(v).
\end{align}

\subsubsection{Estimates of the error terms $J_1(u,v)$ and $J^{(M)}_2(u, v) $}
For all $j, k \leq \lfloor 1.1 n \rfloor$, using \eqref{eq:simple_est} for $f(x) - f_{\lfloor \delta n \rfloor}(x)$ and \eqref{eq:L2_p_and_q_outside} for the $L^2$ norm estimate of $\p_k$ and $\q_j$, we find that
\begin{equation} \label{eq:est_I_1}
  \begin{split}
    I^{(j, k)}_1 := {}& \left\lvert \int^{\infty}_{C_2} (f(x) - f_{\lfloor \delta n \rfloor}(x)) \p_k(x) \q_j(x) dx \right\rvert \\
    \leq {}& \left\lvert \int^{\infty}_{C_2} \cosh(2\sqrt{x}) \p_k(x) \q_j(x) dx \right\rvert \\
    \leq {}& e^{-C_1 n}.
  \end{split}
\end{equation}
In addition, using \eqref{eq:not_simple_est} for $f(x) - f_{\lfloor \delta n \rfloor}(x)$ and \eqref{eq:L2_p_and_q_inside} for the $L^2$ norm estimate of $\p_k$ and $\q_j$,
\begin{equation} \label{eq:est_I_2}
  \begin{split}
    I^{(j, k)}_2 :={}& \left\lvert \int^{C_2}_0 (f(x) - f_{\lfloor \delta n \rfloor}(x)) \p_k(x) \q_j(x) dx \right\rvert \\
    \leq {}& C_5 e^{-\frac{\lfloor \delta n \rfloor}{2} \log(\lfloor \delta n \rfloor)} \left\lvert \int^{C_2}_0 \p_k(x) \q_j(x) dx \right\rvert \\
    \leq {}& C_5 e^{-\frac{\lfloor \delta n \rfloor}{2} \log(\lfloor \delta n \rfloor) + C_3 n}.
  \end{split}
\end{equation}
Hence, we have
\begin{equation} \label{eq:b-a_in_I1_I2}
  \begin{split}
    \lvert \tilde{b}_{k, j} - \tilde{a}_{j, k} \rvert = \left\lvert \int^{\infty}_0 (f(x) - f_{\lfloor \delta n \rfloor}(x)) \p_k(x) \q_j(x) dx \right\rvert \leq {}& I^{(j, k)}_1 + I^{(j, k)}_2 \\
   \leq  {}& e^{-C_1 n} + C_5 e^{-\frac{\lfloor \delta n \rfloor}{2} \log(\lfloor \delta n \rfloor) + C_3 n}.
  \end{split}
\end{equation}
We have the following lemma showing that $J_1(u, v)$ and $J_2^{(M)}(u, v)$ are negligible on the right-hand side of \eqref{eq:ess2}. Again, the $u,v$ here are given by \eqref{eq:uvdef} and they differ from $x^*$ by $C n^{-1/3}$.
\begin{lem} \label{lem:sine_pre_est}
  \begin{enumerate}
  \item \label{enu:lem:sine_pre_est:1}
    \begin{equation}
      \lim_{n \to \infty} e^{n(\Lambda(u) - \Lambda(v))} J_1(u, v) = 0.
    \end{equation}
  \item \label{enu:lem:sine_pre_est:2}
    For any $\epsilon > 0$, there exists $M_0$ such that for all large enough $n$, if $M > M_0$, then
    \begin{equation}
      \lvert e^{n(\Lambda(u) - \Lambda(v))} J^{(M)}_2(u, v) \rvert < \epsilon.
    \end{equation}
  \end{enumerate}
\end{lem}
\begin{proof}
  Using \eqref{eq:b-a_in_I1_I2} and \eqref{eq:est_pq_product} for $j, k = 0, \dotsc, n - 1$, we derive
  \begin{equation}
    \lvert e^{n(\Lambda(u) - \Lambda(v))} J_1(u, v) \rvert \leq  n^2 \left( e^{-C_1 n} + C_5 e^{-\frac{\lfloor \delta n \rfloor}{2} \log(\lfloor \delta n \rfloor) + C_3 n} \right) e^{\frac{2}{3} C_1 n} = \mathfrak{o}(1).
  \end{equation}
  and prove Part \ref{enu:lem:sine_pre_est:1}.

  For Part \ref{enu:lem:sine_pre_est:2}, we first use the same argument to derive that
  \begin{equation}
    \left\lvert e^{n(\Lambda(u) - \Lambda(v))} J_2(u, v) - e^{n(\Lambda(u) - \Lambda(v))} \sum^{n - 1 - M}_{j = n - \lfloor \delta n \rfloor} \sum^{j + \lfloor \delta n \rfloor}_{k = n } \tilde{a}_{k, j} \p_j(u) \q_k(v) \right\rvert = \mathfrak{o}(1),
  \end{equation}
  Hence, we only need to show that if $M$ is large enough, then
  \begin{equation} \label{eq:J_2_approx_ineq}
    \left\lvert e^{n(\Lambda(u) - \Lambda(v))} \sum^{n - 1 - M}_{j = n - \lfloor \delta n \rfloor} \sum^{j + \lfloor \delta n \rfloor}_{k = n } \tilde{a}_{k, j} \p_j(u) \q_k(v) \right\rvert < \epsilon.
  \end{equation}
  To prove \eqref{eq:J_2_approx_ineq}, we estimate $\tilde{a}_{k, j}$ and $e^{n(\Lambda(u) - \Lambda(v))} \p_j(u) \q_k(v)$ separately. First consider $\tilde{a}_{k, j}$. For $k > j$, we denote
  \begin{align}
    I^{(j, k)}_3 = {}& \left\lvert \int^{\infty}_{C_2} (f(x) - f_{k - j - 1}(x)) \p_j(x) \q_k(x) dx \right\rvert, \\
    I^{(j, k)}_4 = {}& \left\lvert \int^{C_2}_0 (f(x) - f_{k - j - 1}(x)) \p_j(x) \q_k(x) dx \right\rvert.
  \end{align}
  By the estimate \eqref{eq:simple_est} of $f(x) - f_{k - j - 1}(x)$ and the $L^2$ norm estimate \eqref{eq:L2_p_and_q_outside}, similarly to the estimate of $I^{(j, k)}_1$ in \eqref{eq:est_I_1}, we have
  \begin{equation} \label{eq:est_I3}
    I^{(j, k)}_3 \leq \left\lvert \int^{\infty}_{C_2} \p_j(x) \q_k(x) \cosh(2\sqrt{x}) dx \right\rvert \leq e^{-C_1 n}.
  \end{equation}
  Using \eqref{eq:est_Ft} for the estimate of $\Lambda_{k/n}(x) - \Lambda_{j/n}(x)$, the estimate \eqref{eq:not_simple_est} for $f(x) - f_{k - j - 1}(x)$, and the $L^2$ norm estimate \eqref{eq:p_kq_kpointbd},
  \begin{equation} \label{eq:est_I4}
    \begin{split}
      I^{(j, k)}_4 = {}& \left\lvert \int^{C_2}_0 e^{n \Lambda_{j/n}(x)} \p_j(x) e^{-n\Lambda_{k/n}(x)} \q_k(x) e^{n(\Lambda_{k/n}(x) - \Lambda_{j/n}(x))} (f(x) - f_{k - j - 1}(x)) dx \right\rvert \\
      \leq {}& C_4 C_5 e^{-\frac{k - j - 1}{2} \log(k - j - 1) + C_0(k - j)}.
    \end{split}
  \end{equation}
  We note that if $k > j$, then $\int^{\infty}_0 f_{k - j - 1}(x) \p_j(x) \q_k(x) dx = 0$, implying
  \begin{equation} \label{eq:a_in_I3_I4}
    \lvert \tilde{a}_{k, j} \rvert = \left\lvert \int^{\infty}_0 (f(x) - f_{k - j - 1}(x)) \p_j(x) \q_k(x) dx \right\rvert \leq I^{(j, k)}_3 + I^{(j, k)}_4.
  \end{equation}

We deduce from \eqref{eq:est_Pk_qj_around_n}, \eqref{eq:a_in_I3_I4}, \eqref{eq:est_I3}, and \eqref{eq:est_I4} that
  \begin{equation}
    \lvert \tilde{a}_{k, j} e^{n(\Lambda(u) - \Lambda(v))} \p_j(u) \q_k(v) \rvert \leq \left( e^{-C_1 n} + C_4 C_5 e^{-\frac{k - j - 1}{2} \log(k - j - 1) + C_0(k - j)} \right) e^{C_0(k - j)} C'.
  \end{equation}
  Thus we derive \eqref{eq:J_2_approx_ineq} for large enough $M$.
\end{proof}

\subsubsection{Estimates of the main term $J^{(M)}_3(u, v)$}
By Lemma \ref{lem:sine_pre_est}, it
 remains to find the approximation of $e^{n(\Lambda(u) - \Lambda(v))} J^{(M)}_3(u, v)$ to prove \eqref{eq:essential_sine_kernel}. From the estimate \eqref{eq:b-a_in_I1_I2} for $\lvert \tilde{b}_{k, j} - \tilde{a}_{j, k} \rvert$, and \eqref{eq:est_Pk_qj_around_n} for $e^{n(\Lambda(u) - \Lambda(v))} \p_k(u) \q_j(v)$, we find that for any fixed $M$,
\begin{equation}\label{eq:j3ap1}
  \lvert e^{n(\Lambda(u) - \Lambda(v))} J^{(M)}_3(u, v) - \Jhat^{(M)}_3(u, v) \rvert = \mathfrak{o}(1),
\end{equation}
where
\begin{equation}\label{eq:j3ap2}
  \begin{split}
    \Jhat^{(M)}_3(u, v) = {}& e^{n(\Lambda(u) - \Lambda(v))} \sum^{n - 1}_{j = n - M} \sum^{j + M}_{k = n} \tilde{a}_{k, j} \p_k(u) \q_j(v) - \tilde{a}_{n - 1, n} \p_{n - 1}(u) \q_n(v) \\
    = {}& \sum^{n - 1}_{j = n - M} \sum^{j + M}_{k = n} a_{k, j} \phat_k(u) \qhat_j(v) - a_{n - 1, n} \phat_{n - 1}(u) \qhat_n(v).
  \end{split}
\end{equation}
Here we have set
\begin{align}
  \phat_k(u) := {}& e^{n\Lambda(u)} e^{-\frac{n \ell_{k/n}}{2}} \sqrt{W^{(n)}_{\alpha}(x)} p^{(n)}_j(x), & \qhat_j(v) := {}& e^{-n\Lambda(v)} \frac{1}{h^{(n)}_j} e^{\frac{n \ell_{j/n}}{2}} \sqrt{W^{(n)}_{\alpha}(x)} q^{(n)}_j(f(x)),
\end{align}
and
\begin{equation}
  a_{k, j} := \int_0^{\infty} \phat_j(x) \qhat_k(x) f(x) dx = \frac{1}{h^{(n)}_k} \int_{0}^{\infty} p^{(n)}_j(x) q^{(n)}_k(x) W^{(n)}_{\alpha}(x) f(x) dx.
\end{equation}

We now find the limit of $a_{n+k,n+j}$ as $n\to\infty$ for fixed $j,k$. Consider the function
	\begin{equation}
f(\J(w))=		 \sinh^2\left(\frac{1}{2} J_c(w) \right)  = \sinh^2 \left( \frac{1}{2} \left( c\sqrt{w} + \arcosh \frac{w + 1}{w - 1} \right) \right).
	\end{equation}
One can show that $ \sinh^2(\frac{1}{2} J_c(w))$ is a meromorphic function, with only one pole of order $1$ at $w = 1$. Thus, we may expand the function $f(\J(w))$ around $1$ as follows,
\begin{equation}\label{Fcoeff}
f(\J(w))=  \frac{1}{4} \left( e^{c\sqrt{w}} \frac{\sqrt{w} + 1}{\sqrt{w} - 1} - 2 + e^{-c\sqrt{w}} \frac{\sqrt{w} - 1}{\sqrt{w} + 1} \right) = \sum^{\infty}_{k = -1} d_k (w - 1)^k,
\end{equation}
where $d_k,k\geq -1$ are the Laurent coefficients of $f(\J(w))$ around $1$. In particular, $d_{-1} = e^c$.

\begin{lem}\label{lem:j3ap3}
  We have the following convergence for fixed integers $j,k$,
  \begin{equation} \label{eq:asy_of_a_jk}
   \mathfrak{a}_{k,j}:= 	\lim_{n\to\infty}a_{n+k, n+j} =
    \begin{cases}
      0, & k < j - 1, \\
      \left(\frac{c^2}{4}\right)^{j-k}d_{k-j}, & k \geq j - 1.
    \end{cases}
  \end{equation}
 We also have the following asymptotics for fixed $j,k$ (recall that $u=x^*+\xi/(n\pi \psi(x^*))$ and $v=x^*+\eta/(n\pi \psi(x^*))$),
  \begin{align}
    \phat_k(u) = {}& \sqrt{(x^*)^{\alpha} h(x^*)} \left( G_k(\Iinv_+(x^*)) e^{-i\xi} e^{n\pi i \mu([x^*, b])} (1 + \bigO(n^{-1}) \right. \notag \\
                   & \left. + G_k(\Iinv_-(x^*)) e^{i\xi} e^{-n\pi i \mu([x^*, b])} (1 + \bigO(n^{-1}) \right), \label{eq:asy_of_phat_k} \\
    \qhat_j(v) = {}& \frac{\sqrt{(x^*)^{\alpha} h(x^*)}}{\frac{2\pi}{\tilde{D}(1)} \left( \frac{c^2(1 - s_1)}{4} \right)^{\alpha + \frac{1}{2}} \left( \frac{c^2}{4} \right)^{j + \frac{1}{4}} e^{jc}} \left( \tilde{G}_j(\Iinv_-(x^*)) e^{-i\eta} e^{n\pi i \mu([x^*, b])} (1 + \bigO(n^{-1}) \right. \notag \\
                   & \left. + \tilde{G}_j(\Iinv_+(x^*)) e^{i\eta} e^{-n\pi i \mu([x^*, b])} (1 + \bigO(n^{-1}) \right). \label{eq:asy_of_qhat_j}
  \end{align}
\end{lem}
\begin{proof}
  The fact that $\mathfrak{a}_{k,j}\equiv 0$ for $k<j-1$ follows from the biorthogonality condition \eqref{eq:biorthogonality}. As in the computation of $h_{n+k}^{(n)}$ in Section \ref{subsec:h^n_n+k}, there exists $\epsilon>0$ such that for all $R$ large enough,
  \begin{equation}
    \int^{\infty}_0 f(x)p^{(n)}_{n + j}(x) q^{(n)}_{n + k}(f(x)) W^{(n)}_{\alpha}(x) dx= 	\int^{R}_0 f(x)p^{(n)}_{n + j}(x) q^{(n)}_{n + k}(f(x)) W^{(n)}_{\alpha}(x) dx+\bigO(e^{\ell n-\epsilon n}).
  \end{equation}

  Dividing the right-hand side of the above equation by $e^{n\ell}$ and sending $n\to\infty$, we get
  \begin{equation}\label{akj-1}
    \begin{split}
      &  -\oint_{C_R} \frac{G_j(\Iinv_1(z)) (1 - s_1)^{\alpha + \frac{1}{2}} \sqrt{s_2 - 1} i \tilde{D}(1)^{-1} e^{kc} z^{\alpha} D(\Iinv_1(z))^{-1}}{(\Iinv_1(z) - s_1)^{\alpha} (\Iinv_1(z) - 1)^k \sqrt{(\Iinv_1(z) - s_1)(\Iinv_1(z) - s_2)}} f(z) dz  \\
      = {}&  -\frac{\sqrt{s_2 - 1} i}{\tilde{D}(1)} \left( \frac{c^2(1 - s_1)}{4} \right)^{\alpha + \frac{1}{2}} \left( \frac{c^2}{4} \right)^j e^{kc} \oint_{C_R} \frac{\Iinv_1(z)^{1/2}}{z^{1/2} (\Iinv_1(z) - s_2)} f(z)dz\\
      = {}& -\frac{\sqrt{s_2 - 1} i}{\tilde{D}(1)} \left( \frac{c^2(1 - s_1)}{4} \right)^{\alpha + \frac{1}{2}} \left( \frac{c^2}{4} \right)^j e^{kc} \oint_{\hat{C}_R}
            \frac{c}{2} (w-1)^{j-k-1} \sinh^2 \left(\frac{1}{2}J_c(w)\right)dw
    .
    \end{split}
  \end{equation}
Here the contours $C_R$ and $\hat{C}_R$ are the same as \eqref{eq:defcrcr'}. Combining \eqref{eq:asy_h}, \eqref{Fcoeff}, and \eqref{akj-1}, we obtain, for $k\geq j-1$,
  \begin{equation}
    \lim_{n\to\infty} a_{n+k,n+j}= \left(\frac{c^2}{4}\right)^{j-k}	 \frac{1}{2\pi \i}	\oint_{C'_R} (w-1)^{j-k-1} \sinh^2 \left(\frac{1}{2}J_c(w)\right)dw=
    \left(\frac{c^2}{4} \right)^j  d_{k - j}.
  \end{equation}
  This proves \eqref{eq:asy_of_a_jk}. Equations	\eqref{eq:asy_of_phat_k} and \eqref{eq:asy_of_qhat_j} are direct consequences of Theorem \ref{thm:main}.
\end{proof}

The final ingredient for proving \eqref{eq:essential_sine_kernel} is the following identity.
\begin{lem}\label{lem:cancel}
	We have
	\begin{equation}\label{eq:cancel1}
		\sum_{k\geq 0,j<0}
\mathfrak{a}_{k,j}		 \left( \frac{c^2}{4}\right)^{-j} e^{-jc}G_k(\Iinv_+(x^*)) \tilde{G}_j(\Iinv_-(x^*))-\mathfrak{a}_{-1,0}
G_{-1}(\Iinv_+(x^*)) \tilde{G}_0(\Iinv_-(x^*)) =0,
	\end{equation}
and
	\begin{equation}\label{eq:cancel2}
	\sum_{k\geq 0,j<0}
\mathfrak{a}_{k,j}		 \left( \frac{c^2}{4}\right)^{-j} e^{-jc}G_k(\Iinv_-(x^*)) \tilde{G}_j(\Iinv_+(x^*))-\mathfrak{a}_{-1,0}
G_{-1}(\Iinv_-(x^*)) \tilde{G}_0(\Iinv_+(x^*)) =0.
\end{equation}
\end{lem}
\begin{proof}
	We only prove \eqref{eq:cancel1} since \eqref{eq:cancel2} can be obtained by complex conjugation.
	By Theorem \ref{thm:main} and the expressions of $G$ and $\tilde{G}$ in \eqref{eq:G_k_from_D_k} and \eqref{eq:G_k_tilde_from_D_k}, we find that the value of the expression on the left-hand side of \eqref{eq:cancel1} is proportional to
	\begin{equation}
		\begin{split}
			 &	\sum_{k\geq 0,j<0}d_{k-j}
			(\Iinv_{+}(x^*)-1)^{k}	(\Iinv_{-}(x^*)-1)^{-j}-d_{-1}(\Iinv_{+}(x^*)-1)^{-1}
			\\
			= & \frac{\Iinv_{-}(x^*)-1}{\Iinv_{+}(x^*)-\Iinv_{-}(x^*)} \sum_{\ell=-1}^{\infty} d_{\ell} \left(
			(\Iinv_{+}(x^*)-1)^{\ell}-	(\Iinv_{-}(x^*)-1)^{\ell}
			\right)\\
			= & \frac{\Iinv_{-}(x^*)-1}{\Iinv_{+}(x^*)-\Iinv_{-}(x^*)}\left(f(J_c(\Iinv_{+}(x^*)))-f(J_c(\Iinv_{-}(x^*)))\right)=0,
		\end{split}
	\end{equation}
where we have used the fact that $d_{\ell},\ell\geq -1$ are the Laurent coefficients of $f(\J(\cdot))$ around $1$ (see \eqref{Fcoeff}) in the last equality.
\end{proof}

Combining Lemmas \ref{lem:sine_pre_est}, \ref{lem:j3ap3} and \ref{lem:cancel}, and using
equations \eqref{eq:j3ap1}, \eqref{eq:j3ap2}, we obtain
\begin{equation}\label{eq:final1}
	\begin{split}
&		  \lim_{n \to \infty} (f_{\lfloor \delta n \rfloor}(u) - f(v)) e^{n\Lambda(u)} e^{-n\Lambda(v)} \sum^{n - 1}_{k = 0} \p_k(u) \q_k(v) \\
=& \lim_{M\to\infty} 	\lim_{n\to\infty}
 	  \hat{J}_3^{(M)}(u,v)\\
 	  :=&  \exp(i (\eta-\xi))J_4+\exp(i (\xi-\eta)) J_5,
	\end{split}
\end{equation}
where
\begin{equation}
	J_4=\frac{i\Iinv_+(x^*)^{1/2} }{(\Iinv_+(x^*)-s_2)c\pi \sqrt{x^*} }
	\left(-d_{-1} (\Iinv_+(x^*)-1)^{-1}+
	\sum_{k\geq 0,j<0}  d_{k-j} (\Iinv_+(x^*)-1)^{k-j}\right),
	\end{equation}
and $J_5$ is the complex conjugate of $J_4$.
We now calculate
\begin{equation}
  \begin{split}
   J_4 = {}& \frac{i}{c\pi \sqrt{x^*}} \frac{\Iinv_{+}(x^*)^{1/2} (\Iinv_{+}(x^*)-1) }{\Iinv_{\pm}(x^*)-s_2 } \sum_{\ell=-1}^{\infty}  \ell d_{\ell}  (\Iinv_{+}(x^*)-1)^{\ell-1}  \\
    = {}& \frac{i}{c\pi \sqrt{x^*}} \frac{\Iinv_{+}(x^*)^{1/2} (\Iinv_{+}(x^*)-1) }{\Iinv_{+}(x^*)-s_2 } f'(\J(\Iinv_{+}(x^*))	\J'(\Iinv_{+}(x^*))\\
    = {}& \frac{i}{c\pi \sqrt{x^*}} \frac{\Iinv_{+}(x^*)^{1/2} (\Iinv_{+}(x^*)-1) }{\Iinv_{+}(x^*)-s_2 }
          f'(x^*) \frac{J_c(\Iinv_{+}(x^*))}{2}
          \frac{c}{2} \frac{\Iinv_{+}(x^*)-s_2}
          {\Iinv_{+}(x^*)^{1/2} (\Iinv_{+}(x^*)-1) } \\
    = {}& \frac{if'(x^*)}{2\pi }.
  \end{split}
\end{equation}
Thus the limit in \eqref{eq:final1} is equal to
\begin{equation}
\frac{f'(x^*)}{2\pi} \left( i\exp(i (\eta-\xi)) -i\exp(i (\xi-\eta))    \right)=\frac{f'(x^*)}{\pi} \sin(\xi-\eta),
\end{equation}
proving \eqref{eq:essential_sine_kernel} and therefore also Theorem \ref{thm:sinekernel}.
	
\section{LLN and CLT for linear statistics and applications to the conductance of the disordered wire model \eqref{eq:special_W}}

In this section, we first give a law of large numbers and a central limit theorem for holomorphic linear statistics of biorthogonal ensembles of the type \eqref{eq:jpdf}. In particular, we will elaborate on the third part of Theorem \ref{eq:find_b_psi} and prove Theorem \ref{thm:variance_general}. Then we apply these results to the special case of \eqref{eq:special_W} with test function $\cosh^{-2}(\sqrt{x})$, and prove \eqref{eq:meanconduc} of Theorem \ref{thm:eq_measure_linear} (Ohm's law) and Theorem \ref{thm:UCF1} (universal conductance fluctuations). As mentioned in Section \ref{sec:Intro}, we will apply existing results on limit theories of biorthogonal ensembles from \cite{Borot-Guionnet-Kozlowski13, Breuer-Duits13, Lambert18}, which we also include here as Propositions \ref{thm:mean}, \ref{prop:clt1}, and \ref{thm:var} below.

\subsection{Limit results for the mean and variance of linear statistics}\label{subsec:formula1}

\paragraph{Results from \cite{Borot-Guionnet-Kozlowski13}}

Note that our biorthogonal ensemble falls within the framework of interacting particles with Coulomb gas interaction studied in \cite{Borot-Guionnet-Kozlowski13}, under the category of `unconstrained model' with one cut, i.e., the case of $g=0$ in \cite{Borot-Guionnet-Kozlowski13}. See the first example in \cite[Page 10457]{Borot-Guionnet-Kozlowski13}. The interaction function $T$ specified in  \cite[Equation (1.2)]{Borot-Guionnet-Kozlowski13}  is now given by a linear combination of $\log W_n^{\alpha}(x)$ and  $\log ((f(x)-f(y)/(x-y))$, where $W_n^{\alpha}$ is the weight given in \eqref{eq:general_weight},

A standing assumption throughout \cite{Borot-Guionnet-Kozlowski13} 
is Hypothesis 2.1 therein, which requires that the equilibrium measure $\mu=\mu_V$ is the unique minimizer of the functional $I_V$ defined in \eqref{eq:energy_functional}. This condition is guaranteed by our Proposition \ref{thm:potential_theory}. 
Consequently, by \cite[Theorem 2.2]{Borot-Guionnet-Kozlowski13}, we have the following law of large numbers for linear statistics.



\begin{prop}\label{thm:mean}
  For any continuous function $\phi: \mathbb{R}\to \mathbb{R}$, we have
  \begin{equation}
    \frac{1}{n}\sum_{i=1}^n \phi(\lambda_i) \to  \int_0^b \phi(x)\psi(x)dx  \,\,\mbox{  in probability}, \quad \text{as } n\to\infty,
  \end{equation}  
  where $\psi(x)$ is the density of the equilibrium measure given in \eqref{eq:density_formula}. 
\end{prop}
\begin{rmk}\label{rmk:llnalpha}
	It is also assumed in \cite[Hypothesis 2.1]{Borot-Guionnet-Kozlowski13} that the weight $W_n^{\alpha}$, or rather $\log W_n^{\alpha}$, is continuous,  which would require $\alpha=0$. However, the large deviation result \cite[Theorem 2.1]{Borot-Guionnet-Kozlowski13} (which is the key ingredient of \cite[Theorem 2.2]{Borot-Guionnet-Kozlowski13})  is robust and can be extended to general $\alpha > -1$, thus establishing Proposition \ref{thm:mean} for all $\alpha>-1$.
\end{rmk}

For the fluctuations of linear statistics,  \cite[Proposition 8.2]{Borot-Guionnet-Kozlowski13} gives a central limit for any holomorphic test function $\phi$, provides that  the equilibrium measure  $\mu$ is off-critical, i.e., the density of $\mu$ is strictly positive in the interior of its support. This condition is also satisfied, according to Theorem \ref{thm:one-cut_rgular_w_hard_edge}  and Requirement \ref{req:one-cut_reg}. The following proposition is thus a direct consequence of \cite[Proposition 8.2]{Borot-Guionnet-Kozlowski13}, where we  denote $M_1[\varphi]$ and $M_2[\varphi]$ appearing  there as $m(\phi)$ and $\sigma^2(\phi)$. (The condition $w(\phi)=0$ there holds trivially since we are in the one-cut regime $g=0$ here).

\begin{prop}\label{prop:clt1}
  Assume that $\alpha=0$ and that $\phi$ is a holomorphic function in a neighborhood of $[0,b]$ in $\mathbb{C}$. Then there exists $m(\phi)\in \mathbb{R}$ and $\sigma^2(\phi)\geq 0$ such that  
  \begin{equation} \label{eq:mean_variance_BGK}
    \sum_{i=1}^n \phi(\lambda_i) -n \int_0^b \phi(x)\psi(x)dx \to \N(m(\phi),\sigma^2(\phi)) \mbox{ in distribution}, \mbox{ as }n\to\infty,  
  \end{equation}
  such that $m(\cdot)$ and $\sigma^2(\cdot)$ are continuous\footnote{
  	The continuity here follows from the contour integral representations in \cite[Proposition 8.2]{Borot-Guionnet-Kozlowski13}: $m(\phi)$ and 
$\sigma^2(\phi)$ can be obtained by integrating $\phi$ against certain fixed functions (arising from Schwinger-Dyson equations) that are analytic off the support of the equilibrium measure.  
} with respect to the uniform norm,
  \begin{equation}
  \forall\, \phi_n \to \phi \mbox{ uniformly in a neighborhood of }[0,b], \quad
  \lim_{n\to\infty}m(\phi_n)=m(\phi)\mbox{ and }
   \lim_{n\to\infty}\sigma^2(\phi_n)=\sigma^2(\phi).
  \end{equation}
\end{prop}

\begin{rmk}\label{rmk:a=0}
	We only state the case of $\alpha=0$ in the above proposition, since  \cite[Proposition 8.2]{Borot-Guionnet-Kozlowski13} needs the analyticity of the interaction function (which    only allows for $\alpha = 0$ so that $\log W_n^{\alpha}$ is analytic) and the universal conductance fluctuation \eqref{eq:UCF} is also exactly concerned with this class.  But  Proposition \ref{prop:clt1} is actually valid for all $\alpha>-1$. Indeed, by carefully checking the proofs, the Schwinger-Dyson equation argument for \cite[Proposition 8.2]{Borot-Guionnet-Kozlowski13} works for the weight \eqref{eq:general_weight} with general $\alpha>-1$.  
\end{rmk}

\paragraph{Proof of Theorem \ref{thm:variance_general}}

When $\phi$ is a polynomial, there is a more explicit formula in \cite{Breuer-Duits13, Lambert18} for $\sigma^2(\phi)$ utilizing the right limit of the recurrence matrix for the biorthogonal polynomials $\{ p^{(n)}_j \}_{j\geq 0}$ and $\{ q^{(n)}_j \}_{j\geq 0}$ defined in \eqref{eq:biorthogonality}, and it is valid for all $\alpha > -1$. As will be clear soon, it would be more convenient to perform a rescaling: $p^{(n)}_j \to \bar p^{(n)}_j:= (4/c^2)^j p^{(n)}_j, q^{(n)}_j \to \bar q^{(n)}_j:=(c^2/4)^j q^{(n)}_j$. For any $n$, we define the (infinite) recurrence matrix $\mathcal{R}^{(n)}$ by the relation
\begin{equation}
  x(\bar p^{(n)}_1(x), \bar p^{(n)}_2(x), \bar p^{(n)}_3(x),\ldots)^T= \mathcal{R}^{(n)}(\bar p^{(n)}_1(x), \bar p^{(n)}_2(x), \bar p^{(n)}_3(x),\ldots)^T, \,\, \forall x\geq 0. 
\end{equation}
Let $r_k,k\geq -1$ be the Laurent coefficients of $\J(w)$ in the annulus $\vert w-1\vert \gg 1$, 
\begin{equation}
\J(w)=\sum_{k=-1}^{\infty} r_k(w-1)^{-k}.
\end{equation}
\begin{lem}\label{lem:recu}
  The recurrence matrix $\mathcal{R}^{(n)}$ satisfies $\mathcal{R}^{(n)}_{j, k}=0$ for all $j<k-1$. In addition, $\mathcal{R}^{(n)}$ has a right limit\footnote{An infinite matrix $\mathcal{M}$ is said to be a right limit of a sequence of infinite matrices $\mathcal{M}^{(n)}$, if there exists a subsequence $n_{\ell} \to \infty$, such that for all $j,k \in \mathbb{Z}$, we have 
$$
\lim_{\ell\to\infty} \mathcal{M}^{(n)}_{n_{\ell}+j, n_{\ell}+k}=\mathcal{M}_{j,k}.
$$  
} $\mathcal{R}$ with $\mathcal{R}_{j,k}=r_{j-k}$ for all $j,k\in \mathbb{Z}$. More precisely, we have
  \begin{equation}
    \lim_{n\to\infty}  \mathcal{R}^{(n)}_{n+j,n+k}=r_{j-k},  \quad \forall j,k\in \mathbb{Z}.
  \end{equation}
\end{lem}

\begin{proof}
  By definition, we have
  \begin{equation}
    \begin{split}
      \mathcal{R}^{(n)}_{j, k} = {}& \left(\int^{\infty}_0 x\bar p^{(n)}_{n + j}(x) \bar q^{(n)}_{n + k}(f(x)) W^{(n)}_{\alpha}(x) dx\right) 	\left( \int^{\infty}_0 \bar p^{(n)}_{n + k}(x) \bar q^{(n)}_{n + k}(f(x)) W^{(n)}_{\alpha}(x) dx \right)^{-1}\\
      = {}& \left(\frac{4}{c^2}\right)^{j-k} \left(\int^{\infty}_0 x p^{(n)}_{n + j}(x)  q^{(n)}_{n + k}(f(x)) W^{(n)}_{\alpha}(x) dx\right) 	\left( \int^{\infty}_0  p^{(n)}_{n + k}(x)  q^{(n)}_{n + k}(f(x)) W^{(n)}_{\alpha}(x) dx \right)^{-1}.
    \end{split}
  \end{equation}
  The fact that $\mathcal{R}^{(n)}_{j, k}=0$ for all $j<k-1$ follows directly from \eqref{eq:biorthogonality}. Repeating the arguments leading to \eqref{akj-1} (and recalling the contour $\hat{C}_R$ in \eqref{eq:defcrcr'}), we find that
  \begin{equation}\label{eq:xrecu}
    \begin{split}
      &	\lim_{n\to\infty}
        e^{-n\ell}	 \int^{\infty}_0 xp^{(n)}_{n + j}(x) q^{(n)}_{n + k}(f(x)) W^{(n)}_{\alpha}(x) dx\\
      = {}& 	-\frac{\sqrt{s_2 - 1} \i}{\tilde{D}(1)} \left( \frac{c^2(1 - s_1)}{4} \right)^{\alpha + \frac{1}{2}} \left( \frac{c^2}{4} \right)^j e^{kc} \oint_{\hat{C}_R}
         \frac{c}{2} (w-1)^{j-k-1} \J(w)dw\\
      = {}& \frac{\sqrt{s_2 - 1} }{\tilde{D}(1)} \left( \frac{c^2(1 - s_1)}{4} \right)^{\alpha + \frac{1}{2}} \left( \frac{c^2}{4} \right)^j e^{kc} c\pi r_{j-k}.
    \end{split}
  \end{equation}
  Lemma \ref{lem:recu} now follows by dividing \eqref{eq:xrecu} by $e^{-n\ell}h^{(n)}_{n+k}$ and applying \eqref{eq:asy_h}.
\end{proof}

Lemma \ref{lem:recu} implies that we can identify the Laurent matrix $\mathcal{R}$ with $\J(w)$ (in the sense of \cite[Section 2.1]{Breuer-Duits13} and \cite[Definition 2.5]{Lambert18}). (The function $\J(w)$ itself corresponds to the notion of Laurent polynomials \cite[Definition 2.5]{Breuer-Duits13}.) Then \cite[Theorem 2.1]{Breuer-Duits13} (see also \cite[Theorem 2.6]{Lambert18} for an alternative proof) yields the following central limit theorem with an explicit expression for the variance. \footnote{The results in \cite{Breuer-Duits13, Lambert18} are originally stated for the case when $\mathcal{R}^{(n)}$ has a uniformly bounded number of nonzero diagonals. A closer inspection of the proofs in \cite{Lambert18} reveals that the result still holds if the number of nonzero upper diagonals or lower diagonals is bounded, which is the case for the $\mathcal{R}^{(n)}$ studied here.}
\begin{prop}\label{thm:var}
  Given any polynomial $P$ with real coefficients, the centered linear statistic converges in distribution to a Gaussian limit as $n\to\infty$,
  \begin{equation} \label{eq:variance_poly}
    \sum_{i=1}^n P(\lambda_i) -\mathbb{E} \left(\sum_{i=1}^n P(\lambda_i)\right)
    \to \N (0, \sigma^2(P) ) \quad \mbox{ with } \sigma^2(P) = \sum_{k=1}^{\infty}kP(\J(w))_k P(\J(w))_{-k},
  \end{equation}
  where $P(\J(w))_k$ is the coefficient in front of $(w-1)^{k}$ in the expansion of $P(\J(w))$ for $\vert w\vert \gg 1$.  
\end{prop}

 The original formula \ref{eq:variance_poly} is less convenient for the purpose of taking the limit $\sigma^2(P_n)$ for a sequence of polynomials $P_n$. To solve this issue, we now consider an alternative representation.  Let $\mathcal{C}(D)$ and $\mathcal{C}'(D)$ be the contours specified in Theorem \ref{thm:variance_general}. Then
\begin{equation}
  P(\J(w))_k= {} \frac{1}{2\pi \i}\oint_{\mathcal{C}'(D)} \frac{P(\J(u))}{(u-1)^k} \frac{du}{u - 1},\quad  P(\J(w))_{-k} = {} \frac{1}{2\pi \i}\oint_{\mathcal{C}(D)} P(\J(v)) (v - 1)^k \frac{dv}{v - 1},
\end{equation}
and thus  
\begin{equation}\label{eq:sigp2}
  \begin{split}
    \sigma^2(P) = {}& \frac{1}{(2\pi i)^2} \oint_{\mathcal{C}'(D)} \frac{du}{u - 1} \oint_{\mathcal{C}(D)} \frac{dv}{v - 1} P(\J(u)) P(\J(v)) \sum^{\infty}_{k = 1} k \left( \frac{v - 1}{u - 1} \right)^k \\
    = {}& \frac{1}{(2\pi i)^2} \oint_{\mathcal{C}'(D)} du \oint_{\mathcal{C}(D)} dv P(\J(u)) P(\J(v)) \frac{1}{(u - v)^2}.
  \end{split}
\end{equation}
Hence, the variance function $\sigma^2(\phi)$ in \eqref{eq:mean_variance_BGK} has a double contour integral formula if $\phi$ is a polynomial, and Theorem \ref{thm:variance_general} is proved for a polynomial $\phi$ without the restriction $\alpha = 0$ (since Proposition \ref{thm:var} is valid for all $\alpha>-1$). Since it is known that $\sigma^2(\phi)$ depends on $\phi$ continuously in its uniform norm (Proposition \ref{prop:clt1}), this contour integral formula extends to all real analytic functions on $[0, b]$. By this extension, we establish Theorem \ref{thm:variance_general}.

\begin{rmk}\label{rmk:congen}
	From the proof we can see that
	the double contour integral formula \eqref{eq:sigform}  is valid  for general  biorthogonal ensembles satisfying the conditions of \cite[Proposition 8.2]{Borot-Guionnet-Kozlowski13}, with $\J$ replaced by corresponding Laurent polynomials. 
\end{rmk}

\subsection{Applications to the conductance of disordered wires}

In this section, we shall use Proposition \ref{thm:mean} (or alternatively Part \ref{enu:eq:find_b_psi_3} of Theorem \ref{eq:find_b_psi}) and Theorem \ref{thm:variance_general} to rigorously establish Ohm's law \eqref{eq:Ohm_law} and universal conductance fluctuation \eqref{eq:UCF}, respectively. 
We will take the test function to be $\cosh^{-2}(\sqrt{x})$, which is holomorphic in a neighborhood of the positive real axis. 

Recall that with the specialization \eqref{eq:special_W}, the parameter $c=2\Mconst$ and the equilibrium measure is given in \eqref{eq:b_linear_case} and \eqref{eq:psi_linear_case} of Theorem \ref{thm:eq_measure_linear}.

\subsubsection{Proof of \eqref{eq:meanconduc} of Theorem \ref{thm:eq_measure_linear}: verifying Ohm's law \eqref{eq:Ohm_law}}\label{subsubsec:meanconduc}
By Proposition \ref{thm:mean} and the boundedness of $\cosh^{-2}(\sqrt{x})$ in $[0,\infty)$, we see that
\begin{equation}\label{eq:integral_eq}
 \lim_{n \to \infty} \frac{1}{n} \E[C_n(\Mconst)]=
		\int^b_0 \frac{1}{\cosh^2(\sqrt{x})} \frac{1}{\pi} \Im\sqrt{\frac{\Iinv_+(x)}{x}}dx.
\end{equation} 
To evaluate this integral, we let $\gamma_3$ be a flat contour surrounding $[0, b]$ and set $\tilde{\gamma}_3=\Iinv_1(\gamma_3)$ to be the image of $\gamma_3$ under the mapping $\Iinv_1$. The integral in \eqref{eq:integral_eq} is equal to

\begin{equation} \label{eq:int_form_mean}
		\begin{split}
			 \frac{-1}{2\pi i} \oint_{\gamma_3} \frac{\sqrt{\Iinv(z)}}{\sqrt{z} \cosh^2 \sqrt{z}} dz 
			= {}& \frac{-1}{2\pi i} \oint_{\tilde{\gamma}_3} \frac{\sqrt{w}}{\sqrt{\J(w)}\cosh^2 \sqrt{\J(w)}} d\J(w) \\
			= {}& \frac{-1}{2\pi i} \oint_{\tilde{\gamma}_3} \frac{\sqrt{w}}{\frac{1}{2} J_c(w)}\cdot  \frac{\frac{1}{2} J_c(w)}{\cosh^2(\frac{1}{2} J_c(w))} \cdot \frac{c}{2} \frac{w-(1+2/c)}{\sqrt{w}(w-1)} dw \\
			= {}& \frac{-1}{2\pi i} \oint_{\tilde{\gamma}_3} \frac{cw - (c + 2)}{(w - 1)(\cosh(J_c(w)) + 1)} dw,
		\end{split}
	\end{equation}
where we have used \eqref{eq:Jdev} and \eqref{eq:JJdev} in the second equality. From \eqref{Fcoeff}, we see 
\begin{equation}\label{eq:formulacoshjc}
	\cosh(J_c(w))+1= \frac{1}{2} \left( e^{c\sqrt{w}} \frac{\sqrt{w} + 1}{\sqrt{w} - 1} + 2 + e^{-c\sqrt{w}} \frac{\sqrt{w} - 1}{\sqrt{w} + 1} \right).
\end{equation}
Inside $\tilde{\gamma}_3$, there is only one zero for $\cosh(J_c(w))+1$ at $w=0$. Using 
\begin{equation}
e^{c\sqrt{w}}=1+c\sqrt{w}+\frac{c^2w}{2}+\bigO(w^{3/2}),\,\,\frac{\sqrt{w}+1}{\sqrt{w}-1}=-(1+2\sqrt{w}+2w+\bigO(w^{3/2})),
\end{equation}
one sees that 
\begin{equation}\label{eq:jcw1}
	\cosh(J_c(w))+1 =-(c+2)^2 w/2+\bigO(w^2) \mbox{ as } w\to 0.
\end{equation} On the other hand, for $w\to 1$, $(w-1)(\cosh(J_c(w))+1)$ is bounded away from $0$, as is clear from \eqref{Fcoeff}. Thus the last line of \eqref{eq:int_form_mean}
is equal to $2(c+2)^{-1}$ and we get (using $c=2\Mconst$)
\begin{equation}
 \lim_{n \to \infty} \frac{1}{n} \E[C_n(\Mconst)]=\frac{2}{c+2}=\frac{2}{2\Mconst+2}=\frac{1}{\Mconst+1},
\end{equation}
proving \eqref{eq:meanconduc} of Theorem \ref{thm:eq_measure_linear}. Ohm's law
  \eqref{eq:Ohm_law} now follows immediately.

\subsubsection{Proof of Theorem \ref{thm:UCF1}: verifying universal conductance fluctuation \eqref{eq:UCF}}\label{subsubsec:pfucf}

We let
\begin{equation}
	T(w) = 	\frac{2}{\cosh(J_c(w)) +1}.
\end{equation}
By Theorem \ref{thm:variance_general}, we need to compute the limiting variance
  \begin{equation}\label{eq:sigma1}
	\sigma^2(T) = \frac{1}{(2\pi i)^2} \oint_{\gamma_4 } du \oint_{\gamma_5} dv T(u) T(v) \frac{1}{(u - v)^2},
\end{equation}
where $\gamma_4$ is a positively oriented smooth Jordan curve encircling another smooth Jordan curve $\gamma_5$, which in turn encircles $\partial D$. 

Within $\gamma_4$, $T(w)$ has a zero at $w =1$, and has only one pole at $w = 0$ with order $1$. Using \eqref{eq:jcw1}, the residue at $0$ is
\begin{equation}
	\Res_{w = 0} T(w) = -\left( \frac{c}{2} + 1 \right)^{-2}.
\end{equation}
Thus, the integration with respect to $v$ in \eqref{eq:sigma1} gives
\begin{equation}\label{eq:sigmaT1}
	\sigma^2(T)=-\left(\frac{c}{2}+1\right)^{-2} \frac{1}{ 2\pi i } \oint_{\gamma_4 }   T(u)  \frac{1}{u^2}du.
\end{equation}
To evaluate this integral, we note that,
outside of $\gamma_4$, by Lemma \ref{lem:J_c}, $T(w)$ has infinitely many poles for $w\leq s_1$, in the form of
\begin{equation}
	t^*_k = \Iinv_1(-(k + \frac{1}{2})^2 \pi^2)=J_c^{-1}( (2k+1)\pi i ) , \quad k = 0, 1, 2, \dotsc.
\end{equation}
It follows that 
the function
\begin{equation}
	\tilde{T}(w) := T(w) + \left( \frac{c}{2} + 1 \right)^{-2} \frac{1}{w }
\end{equation}
is holomorphic for $|w-1|<1-t_1^*$. Moreover, combining \eqref{eq:sigma1} and the fact
\begin{equation}
\oint_{\gamma_4} \frac{1}{w^3} dw=0,
\end{equation}
we see that
\begin{equation}\label{eq:varphifin}
	\sigma^2(T) = -\left(\frac{c}{2}+1\right)^{-2}  
	\frac{1}{2\pi i} \oint_{\gamma_4} \tilde{T}(u) \frac{1}{u^2} du = -\left(\frac{c}{2}+1\right)^{-2} \tilde{T}'(0).
\end{equation}


By \eqref{eq:formulacoshjc}, we can	write
	\begin{equation}
		\begin{split}
			\tilde{T}(w) = {}& \frac{2}{\cosh(J_c(w))+1 } + \left( \frac{c}{2} + 1 \right)^{-2} \frac{1}{w} \\
			= {}& \frac{2(w-1)}{(w + 1)\cosh(c\sqrt{w}) + 2\sqrt{w}\sinh(c\sqrt{w}) + w - 1}  + \left( \frac{c}{2} + 1 \right)^{-2} \frac{1}{w}.
		\end{split}
	\end{equation}
	Using the expansion
        \begin{equation}
	\cosh(c\sqrt{w})= 1+\frac{c^2 w}{2}+\frac{c^4 w^2}{24}+\frac{c^6 w^3}{720}+\bigO(w^4),\,
	\sqrt{w}\sinh(c\sqrt{w})=cw+\frac{c^3 w^{2}}{6}+ \frac{c^5 w^3}{120}+\bigO(w^{4}),
	\end{equation}
	we have the following asymptotic for the denominator
\begin{equation}
	\begin{split}
&		(w + 1)\cosh(c\sqrt{w}) + 2\sqrt{w}\sinh(c\sqrt{w}) + w - 1\\
		= & \frac{1}{2}(c+2)^2 w+\left( \frac{c^2}{2}+\frac{c^3}{3}+\frac{c^4}{24}\right)w^2+\left(\frac{c^4}{24}+\frac{c^5}{60}+\frac{c^6}{720}\right)w^3 +\bigO(w^4).
	\end{split}
\end{equation}
After some calculations, one finds that
\begin{equation}
\begin{split}
			\tilde{T}(w) =&\left(\frac{c}{2}+1 \right)^{-3} \left(\frac{c^3}{8}+\frac{c^2}{4}+\frac{c}{2}+1\right)
			-\left(\frac{c}{2}+1\right)^{-4}
			\left(\frac{c^6}{960} + \frac{c^5}{80} + \frac{c^4}{16} + \frac{c^3}{6} + \frac{c^2}{4}  \right)w+\bigO(w^2)\\
			=&\left(\frac{c}{2}+1 \right)^{-3} \left(\frac{c^3}{8}+\frac{c^2}{4}+\frac{c}{2}+1\right)- \frac{1}{15} \left(\frac{c}{2}+1\right)^2 \left( 1 - \frac{3c + 1}{(\frac{c}{2} + 1)^6} \right)w+\bigO(w^2).
		\end{split}
\end{equation}
Using \eqref{eq:varphifin} and inserting the value $c=2\Mconst$, we get
\begin{equation}
\sigma^2(T)=
	\frac{1}{15} \left( 1 - \frac{3c + 1}{(\frac{c}{2} + 1)^6} \right)=	\frac{1}{15}\left(1-\frac{6\Mconst+1}{(\Mconst+1)^6}\right), 
\end{equation}
proving \eqref{eq:wireclt1}. The asymptotic expansion of the partition function for the biorthogonal ensemble (\cite[Theorem 1.1]{Borot-Guionnet-Kozlowski13}) implies\footnote{The exponential moment of linear statistics can be written as the ratio of two partition functions with different potentials, which can then be computed using the asymptotic expansion for partition functions.} that 
the exponential moment of $C_n(\Mconst)-\mathbb{E}(C_n(\Mconst))$ is uniformly bounded, and thus 
 the sequence
$[C_n(\Mconst)-\mathbb{E}(C_n(\Mconst))]^2$ is uniformly integrable. Consequently, we also have 
\begin{equation}
\lim_{n\to\infty}\Var(C_n(\Mconst))=
\sigma^2(T)=\frac{1}{15}\left(1-\frac{6\Mconst+1}{(\Mconst+1)^6}\right),
\end{equation}
proving \eqref{eq:wireclt2}. The universal conductance fluctuation \eqref{eq:UCF} directly follows by
taking $\Mconst\to\infty$.

	\appendix
	\section{Properties of $\J[x](s)$ and related functions} \label{sec:J_x_prop}

\paragraph{Proof of Part \ref{enu:lem:J_c:4} of Lemma \ref{lem:J_c}}

Below we give a constructive description of $\gamma_1(x)$. Any $s \in \compC_+$ can be represented as
\begin{equation} \label{eq:s_in_uv}
	s = \frac{\cosh(u + iv) + 1}{\cosh(u + iv) - 1}, \quad u \in (0, \infty) \text{ and } v \in (-\pi, 0).
\end{equation}
Then the condition $s \in \gamma_1(x) \subseteq \compC_+$ is equivalent to $\Im J_x(s) = 0$, which can be expressed as $x\sin v/(\cosh u - \cos v) - v = 0$, or equivalently,
\begin{equation} \label{eq:gamma_1_eq}
	\cosh u = x \frac{\sin v}{v} + \cos v.
\end{equation}
By direct computation,
we see that the right-hand side of \eqref{eq:gamma_1_eq} is an increasing function on $(-\pi, 0)$, and its limits at $-\pi$ and $0$ are $-1$ and $x + 1$ respectively. Hence, \eqref{eq:gamma_1_eq} has a solution only if $v \in (v^*, 0)$, where $v^* \in (-\pi, 0)$ is the solution to $1 = x\sin(v)/v + \cos v$. We then construct the curve
\begin{equation} \label{eq:gamma'_1x}
	\gamma'_1(x) = \{ \frac{\cosh(u(v) + iv) + 1}{\cosh(u(v) + iv) - 1} : v \in [v^*, 0] \}, \quad \text{where} \quad u(v) = \arcosh \left( x \frac{\sin v}{v} + \cos v \right).
\end{equation}
which lies in $\compC_+$ and connects $\gamma'_1(v^*) = (\cosh(iv^*) + 1)/(\cosh(iv^*) - 1) = s_1(x)$ and $\gamma'_1(0) = 1 + 2/x = s_2(x)$.

By expressing $s$ in $u, v$ in \eqref{eq:s_in_uv} and parametrizing $u$ by $v$ as in \eqref{eq:gamma'_1x}, we have that $J_x(s) \mid_{s \in \gamma'_1(x)}$ is parametrized by $v \in (v^*, 0)$. Then we can compute
\begin{equation} \label{eq:derivative_on_gamma_1}
	\frac{d}{dv} J_x(s(u(v), v)) = -\frac{x + 1 - \cosh(u(v) + iv)}{\cosh(u(v) + iv) - 1} (u'(v) + i).
\end{equation}
Since by the construction, we know that $J_x(s) \in \realR$ if $s \in \gamma'_1(x)$, we know that the derivative above is real valued wherever it is well defined on $(v^*, 0)$. On the other hand, we have that $x + 1 - \cosh(u(v) + iv) \neq 0$ and $u'(v) + i \neq 0$ on $(v^*, 0)$. Hence the derivative in \eqref{eq:derivative_on_gamma_1} is real, non-vanishing, and continuous on $(v^*, 0)$, and then it has to be always positive or always negative there. By comparing $J_x(s_1(x))$ and $J_x(s_2(x))$, we conclude that the derivative is always positive on $(v^*, 0)$.

Now we see that the curve $\gamma'_1(x)$ satisfies the properties of $\gamma_1(x)$ described in Part \ref{enu:lem:J_c:4} of Lemma \ref{lem:J_c}, and it is the only candidate for $\gamma_1(x)$. So we let $\gamma_1(x)$ be $\gamma'_1(x)$ constructed in \eqref{eq:gamma'_1x}, and prove constructively Part \ref{enu:lem:J_c:4} of Lemma \ref{lem:J_c}.

\paragraph{Proof of Lemma \ref{lem:conformal_mapping}}

Both $\compC_+$ and $\compC_+ \setminus D$ are simply connected regions. From Parts \ref{enu:lem:J_c:1}, \ref{enu:lem:J_c:2} and \ref{enu:lem:J_c:4} of Lemma \ref{lem:J_c}, we have that $\J[x](s)$ maps $(-\infty, s_1(x)] \cup \gamma_1(x) \cup [s_2(x), \infty)$, the boundary of $\compC_+ \setminus D$, homeomorphically to $\realR$, the boundary of $\compC_+$. Also we have $J[x](s) = (x^2/4)s + \bigO(1)$. For any large enough $R \in \realR_+$, we let $C_R$ be the region enclosed by the semicircle $S_R := \{ R e^{i\theta} : \theta \in (0, \pi) \}$ and the lower boundary $(-R, s_1(s)] \cup \gamma_1(x) \cup [s_2(x), R)$. We then have that $\J[x](s)$ maps the boundary of $C_R$ homeomorphically to its image lying in $\compC_+ \cup \realR$. Hence by standard arguments in complex analysis, $\J[x](s)$ maps $C_R$ analytically and bijectively to its image lying in $\compC_+$. Letting $R \to \infty$, we have that $\J[x](s)$ maps $\compC_+ \setminus D$ analytically and bijectively to $\compC_+$.

Similarly, we can prove that $\J[x](s)$ maps $\compC_+ \cap D$ bijectively to $\compC_- \cap \paraP$.

At last, since $\J[x](\bar{s}) = \overline{\J[x](s)}$, the results above can be extended to $\compC \setminus \bar{D}$ and $\compC \cap D$. Hence we finish the proof of Lemma \ref{lem:conformal_mapping}.

\paragraph{Limiting shape of $\gamma_1(x)$ as $x \to \infty$ and $x \to 0$}

With the help of expression \eqref{eq:gamma'_1x} of $\gamma_1(x) = \gamma'_1(x)$, we have the following results by direct calculation.
\begin{enumerate}
	\item 
	As $x \to \infty$, $s_1(x) = -(\frac{\pi}{x})^2 (1 + \bigO(x^{-1}))$ and $\mathfrak{b}(x) = \frac{x^2}{4}(1 + \bigO(x^{-1}))$. 
	Given $\epsilon > 0$, we have
	\begin{equation}
		\Iinv_{x, +}(y) = \frac{4}{x^2} \left( y(1 + \bigO(x^{-1})) + i \pi \sqrt{y} (1 + \bigO(x^{-1})) \right), \quad y \in \left( \frac{\epsilon x^2}{4}, \frac{(1 - \epsilon) x^2}{4} \right),
	\end{equation}
	where the two $\bigO(x^{-1})$ terms are uniform for $y \in (\epsilon x^2/4, (1 - \epsilon) x^2/4)$.
	\item 
	As $x \to 0_+$, $s_1(x) = -\frac{2}{x} (1 + \bigO(x))$ and $\mathfrak{b}(x) = 2x(1 + \bigO(x))$. 
	Given $\epsilon > 0$, we have
	\begin{equation}
		\Iinv_{x, +}(y) = \frac{2}{x^2} \left( (y - x) + \sqrt{y(2x - y)} i \right) (1 + \bigO(x)), \quad y \in (\epsilon x, (2 - \epsilon)x),
	\end{equation}
	where the $1 + \bigO(x)$ term is uniform for $y \in (\epsilon x, (2 - \epsilon)x)$.
\end{enumerate}

\paragraph{Limit behaviour of $\J[x](s)$ around $s_1(x)$ and $s_2(x)$, and its inverse functions around $0$ and $\mathfrak{b}(x)$}

By direct computation, we have the following.
\begin{enumerate}
	\item 
	In the vicinity of $s_1(x)$
	\begin{equation}
		\J[x](s) = \frac{x^2 (s_2(x) - s_1(x))^2}{16 s_1(x) (1 - s_1(x))^2} (s - s_1(x))^2 + \bigO((s - s_1(x))^3), 
	\end{equation}
	and around $0$, for a small enough $\epsilon > 0$,
	\begin{align}
		\Iinv_{x, \pm}(y) = {}& s_1(x) \pm e_1(x) \sqrt{y} i + \bigO(y), && y \in (0, \epsilon), \\
		\Iinv_{x, 1}(z) = {}& s_1(x) - e_1(x) (-z)^{\frac{1}{2}} + \bigO(z), && z \in N(0, \epsilon) \setminus [0, \epsilon), \label{eq:I_1_at_0} \\
		\Iinv_{x, 2}(z) = {}& s_1(x) + e_1(x) (-z)^{\frac{1}{2}} + \bigO(z), && z \in N(0, \epsilon) \setminus [0, \epsilon). \label{eq:I_2_at_0}
	\end{align}
	where
	\begin{equation} \label{eq:defn_e_1}
		e_1(x) = \frac{4\sqrt{-s_1(x)}(1 - s_1(x))}{x(s_2(x) - s_1(x))}.
	\end{equation}
	
	\item 
	In the vicinity of $s_2(x)$
	\begin{equation}
		\J[x](s) = \mathfrak{b}(x) + \frac{x^2 \mathfrak{b}(x)^{1/2}}{8 s_2(x)^{1/2}} (s - s_2(x))^2 + \bigO((s - s_2(x))^3), 
	\end{equation}
	and around $\mathfrak{b}(x)$, for a small enough $\epsilon > 0$,
	\begin{align}
		\Iinv_{x, \pm}(y) = {}& s_2(x) \pm d_1(x) \sqrt{\mathfrak{b}(x) - y} i + \bigO(\mathfrak{b}(x) - y), && y \in (\mathfrak{b}(x) - \epsilon, \mathfrak{b}(x)), \\
		\Iinv_{x, 1}(z) = {}& s_2(x) + d_1(x) \sqrt{z - \mathfrak{b}(x)} + \bigO(\mathfrak{b}(x) - z), && z \in N(\mathfrak{b}(x), \epsilon) \setminus [\mathfrak{b}(x) - \epsilon, \mathfrak{b}(x)), \label{eq:I_1_at_b} \\
		\Iinv_{x, 2}(z) = {}& s_2(x) - d_1(x) \sqrt{z - \mathfrak{b}(x)} + \bigO(\mathfrak{b}(x) - z), && z \in N(\mathfrak{b}(x), \epsilon) \setminus [\mathfrak{b}(x) - \epsilon, \mathfrak{b}(x)). \label{eq:I_2_at_b}
	\end{align}
	where
	\begin{equation} \label{eq:defn_d_1}
		d_1(x) = \frac{2\sqrt{2} s_2(x)^{1/4}}{x \mathfrak{b}(x)^{1/4}}.
	\end{equation}
	
\end{enumerate}

\section{Local universal parametrices}

\subsection{The Airy parametrix}\label{app:Airy}

In this subsection, let $y_0$, $y_1$ and $y_2$ be the functions defined by
\begin{equation}
	y_0(\zeta) = \sqrt{2\pi}e^{-\frac{\pi i}{4}} \Ai(\zeta), \quad y_1(\zeta) = \sqrt{2\pi}e^{-\frac{\pi i}{4}} \omega\Ai(\omega \zeta), \quad y_2(\zeta) = \sqrt{2\pi}e^{-\frac{\pi i}{4}} \omega^2\Ai(\omega^2 \zeta),
\end{equation}
where $\Ai$ is the usual Airy function (cf. \cite[Chapter 9]{Boisvert-Clark-Lozier-Olver10}) and $\omega=e^{2\pi i/3}$. We then define a $2\times 2$ matrix-valued function $\Psi^{(\Ai)}$ by
\begin{equation} \label{eq:defn_Psi_Ai}
	\Psi^{(\Ai)}(\zeta)
	= \left\{
	\begin{array}{ll}
		\begin{pmatrix}
			y_0(\zeta) &  -y_2(\zeta) \\
			y_0'(\zeta) & -y_2'(\zeta)
		\end{pmatrix}, & \hbox{$\arg \zeta \in (0,\frac{2\pi}{3})$,} \\
		\begin{pmatrix}
			-y_1(\zeta) &  -y_2(\zeta) \\
			-y_1'(\zeta) & -y_2'(\zeta)
		\end{pmatrix}, & \hbox{$\arg \zeta \in (\frac{2\pi}{3},\pi)$,} \\
		\begin{pmatrix}
			-y_2(\zeta) &  y_1(\zeta) \\
			-y_2'(\zeta) & y_1'(\zeta)
		\end{pmatrix}, & \hbox{$\arg \zeta \in (-\pi,-\frac{2\pi}{3})$,} \\
		\begin{pmatrix}
			y_0(\zeta) &  y_1(\zeta) \\
			y_0'(\zeta) & y_1'(\zeta)
		\end{pmatrix}, & \hbox{$\arg \zeta \in  (-\frac{2\pi}{3},0)$.}
	\end{array}
	\right.
\end{equation}
It is well-known that $\det (\Psi^{(\Ai)}(z))=1$ and $\Psi^{(\Ai)}(\zeta)$ is the unique solution of the following $2 \times 2$ RH problem; cf. \cite[Section 7.6]{Deift99}.
\begin{RHP} \hfill\label{rhp:Ai}
	\begin{enumerate}
		\item
		$ \Psi(\zeta)$ is analytic in $\mathbb{C} \setminus \Gamma_{\Ai}$, where the contour $\Gamma_{\Ai}$ is defined in
		\begin{equation} \label{def:AiryContour}
			\Gamma_{\Ai}:=e^{-\frac{2\pi i}{3}}[0,+\infty) \cup \mathbb{R} \cup e^{\frac{2\pi i}{3}}[0,+\infty)
		\end{equation}
		with the orientation shown in Figure \ref{fig:jumps-Psi-A}.
		\item
		For $z \in \Gamma_{\Ai}$, we have
		\begin{equation} \label{eq:Ai_jump}
			\Psi_+(\zeta) =  \Psi_-(\zeta)
			\begin{cases}
				\begin{pmatrix}
					1 & 1 \\
					0 & 1
				\end{pmatrix},
				& \arg \zeta =0, \\
				\begin{pmatrix}
					1 & 0 \\
					1 & 1
				\end{pmatrix},
				& \arg \zeta = \pm \frac{2\pi }{3}, \\
				\begin{pmatrix}
					0 & 1  \\
					-1 & 0
				\end{pmatrix},
				& \arg \zeta = \pi.
			\end{cases}
		\end{equation}
		\item
		As $\zeta \to \infty$, we have
		\begin{equation}
			\Psi(\zeta) = \zeta^{-\frac{1}{4} \sigma_3} \frac{1}{\sqrt{2}}
			\begin{pmatrix}
				1 & 1 \\
				-1 & 1
			\end{pmatrix}
			e^{-\frac{\pi i}{4} \sigma_3}(I+\bigO(\zeta^{-\frac32}))e^{-\frac23 \zeta^{\frac32}\sigma_3}.
		\end{equation}
		\item
		As $\zeta \to 0$, we have $\Psi_{i, j}(\zeta) = \bigO(1)$, where $i, j = 1, 2$.
	\end{enumerate}
\end{RHP}
We note that the jump condition \eqref{eq:Ai_jump} can be derived from the identity \cite[Equation (7.116)]{Deift99}
\begin{equation} \label{eq:Airy_sum_zero}
	\Ai(\zeta) + \omega \Ai(\omega \zeta) + \omega^2 \Ai(\omega^2 \zeta) = 0.
\end{equation}

\begin{figure}[htb]
	\begin{minipage}{0.55\linewidth}
		\centering
		\includegraphics{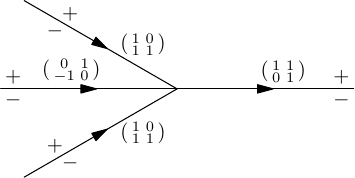}
		\caption{The jump contour $\Gamma_{\Ai}$ for the RH problem \ref{rhp:Ai} for $\Psi^{(\Ai)}$.}
		\label{fig:jumps-Psi-A}
	\end{minipage}
	\hspace{\stretch{1}}
	\begin{minipage}{0.4\linewidth}
		\centering
		\includegraphics{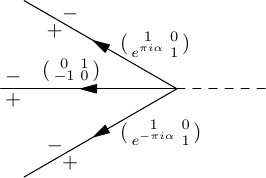}
		\caption{The jump contour $\Gamma_{\Be}$ for the RH problem \ref{RHP:Bessel} for $\Psi^{(\Be)}_{\alpha}$.}
		\label{fig:jumps-Psi-B}
	\end{minipage}
\end{figure}

\subsection{The Bessel parametrix}

The Bessel kernel used in our paper is essentially the $\tilde{\Psi}$ constructed in \cite[Equation (6.51)]{Kuijlaars-McLaughlin-Van_Assche-Vanlessen04}. In this subsection, let $w_0$, $w_1$, $w_2$ and $w_3$ be the functions defined by
\begin{align}
	w_0(\zeta) = {}& I_{\alpha}(2\zeta^{\frac{1}{2}}), & w_1(\zeta) = {}& \frac{-i}{\pi} K_{\alpha}(2\zeta^{\frac{1}{2}}), & w_2(\zeta) = {}& \frac{1}{2} H^{(1)}_{\alpha}(2(-\zeta)^{\frac{1}{2}}), & w_3(\zeta) = {}& \frac{1}{2} H^{(2)}_{\alpha}(2(-\zeta)^{\frac{1}{2}}),
\end{align}
where all $\zeta^{1/2}$ are defined by the principal branch on the sector $\arg \zeta \in (-\pi, \pi)$, $I_{\alpha}$ and $K_{\alpha}$ are modified Bessel functions of order $\alpha$, and $H^{(1)}_{\alpha}$ and $H^{(2)}_{\alpha}$ are Hankel functions of order $\alpha$ of the first and second kind, respectively. We then define a $2 \times 2$ matrix-valued function $\Psi^{(\Be)}_{\alpha}$ by
\begin{equation}\label{eq:psibeal}
	\Psi^{(\Be)}_{\alpha}(\zeta) =
	\begin{pmatrix}
		1 & 0 \\
		0 & -2\pi i \zeta
	\end{pmatrix}
	\times
	\begin{cases}
		\begin{pmatrix}
			w_0(\zeta) & w_1(\zeta) \\
			w'_0(\zeta) & w'_1(\zeta)
		\end{pmatrix},
		& \arg \zeta \in (-\frac{2\pi}{3}, \frac{2\pi}{3}), \\
		\begin{pmatrix}
			w_2(\zeta) & -w_3(\zeta) \\
			w'_2(\zeta) & -w'_3(\zeta)
		\end{pmatrix}
		e^{\frac{\alpha}{2} \pi i \sigma_3}, & \arg \zeta \in (\frac{2\pi}{3}, \pi), \\
		\begin{pmatrix}
			w_3(\zeta) & w_2(\zeta) \\
			w'_3(\zeta) & w'_2(\zeta)
		\end{pmatrix}
		e^{-\frac{\alpha}{2} \pi i \sigma_3}, & \arg \zeta \in (-\pi, -\frac{2\pi}{3}).
	\end{cases}
\end{equation}
It is well known that $\det(\Psi^{(\Be)}_{\alpha}) = 1$ and $\Psi^{(\Be)}_{\alpha}$ is the unique solution of the following RH problem:

\begin{RHP} \hfill \label{RHP:Bessel}
	\begin{enumerate}
		\item
		$\Psi(\zeta)$ is analytic in $\compC \setminus \Gamma_{\Be}$, where the contour $\Gamma_{\Be}$ is defined in 
		\begin{equation} \label{def:BesselContour}
			\Gamma_{\Be}:=e^{-\frac{2\pi i}{3}}[0,+\infty) \cup \mathbb{R_-} \cup e^{\frac{2\pi i}{3}}[0,+\infty)
		\end{equation}
		with the orientation shown in Figure \ref{fig:jumps-Psi-B}.
		\item
		For $z \in \Gamma_{\Be}$, we have
		\begin{equation} \label{eq:Bessel_jump}
			\Psi_+(\zeta) =  \Psi_-(\zeta)
			\begin{cases}
				\begin{pmatrix}
					1 & 0 \\
					e^{\pi i \alpha} & 1
				\end{pmatrix},
				& \arg \zeta = \frac{2\pi }{3}, \\
				\begin{pmatrix}
					1 & 0 \\
					e^{-\pi i \alpha} & 1
				\end{pmatrix},
				& \arg \zeta = -\frac{2\pi }{3}, \\
				\begin{pmatrix}
					0 & 1  \\
					-1 & 0
				\end{pmatrix},
				& \arg \zeta = \pi.
			\end{cases}
		\end{equation}
		\item
		As $\zeta \to \infty$, we have
		\begin{equation}
			\Psi(\zeta) = (2\pi)^{-\frac{1}{2} \sigma_3} \zeta^{-\frac{1}{4} \sigma_3} \frac{1}{\sqrt{2}}
			\begin{pmatrix}
				1 & -i \\
				-i & 1
			\end{pmatrix}
			(I+\bigO(\zeta^{-\frac{1}{2}}))e^{2 \zeta^{\frac{1}{2}} \sigma_3}.
		\end{equation}
		\item
		As $\zeta \to 0$, if $\alpha \in (-1, 0)$, then
		\begin{align}
			\Psi(\zeta) = {}& 
			\begin{pmatrix}
				\bigO(\zeta^{\alpha/2}) & \bigO(\zeta^{\alpha/2}) \\
				\bigO(\zeta^{\alpha/2}) & \bigO(\zeta^{\alpha/2})
			\end{pmatrix}, &
			\Psi(\zeta)^{-1} = {}& 
			\begin{pmatrix}
				\bigO(\zeta^{\alpha/2}) & \bigO(\zeta^{\alpha/2}) \\
				\bigO(\zeta^{\alpha/2}) & \bigO(\zeta^{\alpha/2})
			\end{pmatrix},
		\end{align}
		if $\alpha = 0$, then
		\begin{align}
			\Psi(\zeta) = {}&
			\begin{pmatrix}
				\bigO(\log \zeta) & \bigO(\log \zeta) \\
				\bigO(\log \zeta) & \bigO(\log \zeta)
			\end{pmatrix}, &
			\Psi(\zeta)^{-1} = {}&
			\begin{pmatrix}
				\bigO(\log \zeta) & \bigO(\log \zeta) \\
				\bigO(\log \zeta) & \bigO(\log \zeta)
			\end{pmatrix}, &
		\end{align}
		and if $\alpha > 0$, then in the sector $\arg \zeta \in (-2\pi/3, 2\pi/3)$,
		\begin{align}
			\Psi(\zeta) = {}&
			\begin{pmatrix}
				\bigO(\zeta^{\alpha/2}) & \bigO(\zeta^{-\alpha/2}) \\
				\bigO(\zeta^{\alpha/2}) & \bigO(\zeta^{-\alpha/2})
			\end{pmatrix}, &
			\Psi(\zeta)^{-1} = {}&
			\begin{pmatrix}
				\bigO(\zeta^{-\alpha/2}) & \bigO(\zeta^{-\alpha/2}) \\
				\bigO(\zeta^{\alpha/2}) & \bigO(\zeta^{\alpha/2})
			\end{pmatrix},
		\end{align}
		and in the sector $\arg \zeta \in (2\pi/3, \pi)$ or in the sector $\arg \zeta \in (-\pi, -2\pi/3)$,
		\begin{align}
			\Psi(\zeta) = {}&
			\begin{pmatrix}
				\bigO(\zeta^{-\alpha/2}) & \bigO(\zeta^{-\alpha/2}) \\
				\bigO(\zeta^{-\alpha/2}) & \bigO(\zeta^{-\alpha/2})
			\end{pmatrix}, &
			\Psi(\zeta)^{-1} = {}&
			\begin{pmatrix}
				\bigO(\zeta^{-\alpha/2}) & \bigO(\zeta^{-\alpha/2}) \\
				\bigO(\zeta^{-\alpha/2}) & \bigO(\zeta^{-\alpha/2})
			\end{pmatrix}.
		\end{align}
	\end{enumerate}
\end{RHP}
Like \eqref{eq:Airy_sum_zero}, the jump condition \eqref{eq:Bessel_jump} can be derived from the identity \cite[Proof of Theorem 6.3]{Kuijlaars-McLaughlin-Van_Assche-Vanlessen04}
\begin{equation} \label{eq:Bessel_sum_zero}
	w_2(\zeta) + w_3(\zeta) = w_0(\zeta).
\end{equation}
	
	{\small \textbf{Declaration of interests} The authors declare that they have no known competing financial interests or personal relationships that could appears to influence the work reported in this paper.
		
		\textbf{Data availability} Data sharing not applicable to this article as no datasets were generated or analyzed during the current study.}
	
	\bibliographystyle{plain}
	\bibliography{../bibliography/bibliography.bib}

\end{document}